%% file: sheafology_arxiv.tex
\documentclass	[nonacm,
			 natbib=false,
			 ]{acmart}						
\usepackage[utf8]{inputenc}
\usepackage[T1]{fontenc}
\usepackage{mathtools}
\usepackage{amsthm, 
		    amsmath,
		    bbm,
		    stmaryrd,
		    mathrsfs,
		    dsfont}
\usepackage{thmtools}
\usepackage[disable]{todonotes}
\usepackage{tcolorbox}
\usepackage{wrapfig}
\usepackage{hyperref}
\hypersetup{hidelinks,
	           final}
\usepackage[capitalise,
		    nameinlink]{cleveref}
\usepackage{tikz-cd}
\usetikzlibrary{automata,
  arrows.meta,
  positioning,
  patterns,
  patterns.meta,
  calc}
\usepackage{general-macros}
\usepackage{local-macros}
\usepackage{thmstyle}
\usepackage[createShortEnv,commandRef=Cref,conf={no link to proof}]{proof-at-the-end}
\pratendSetGlobal{text proof={Proof of \string\pratendRef{thm:prAtEnd\pratendcountercurrent} on \string\cpageref{thm:prAtEnd\pratendcountercurrent}}}

\renewcommand*{\phi}{\varphi}
\RequirePackage[datamodel=acmdatamodel,
			   style=acmnumeric, 
			   citestyle=numeric-comp,
			   maxnames=99,
			   backend=biber]{biblatex}
\usepackage[author=anonymous,nomargin,marginclue,footnote]{fixme}
\FXRegisterAuthor{cf}{acf}{CF}
\FXRegisterAuthor{bs}{abs}{BvS}
\FXRegisterAuthor{hb}{ahb}{HB}
\newtcolorbox[auto counter,
		      crefname={diagram}{diagrams},
		      Crefname={Diagram}{Diagrams}]{diagrambox}[2][]{%
		      colback=white,colframe=purple!75!black,fonttitle=\bfseries,
		      title=Diagram~\thetcbcounter: #2,#1}
\NewDocumentEnvironment{diagram}{o m m}{
  \IfNoValueTF{#1}{
    \begin{diagrambox}[float,label=#2]{#3}
  }{
    \begin{diagrambox}[float,floatplacement=#1,label=#2]{#3}
  }
}{
  \end{diagrambox}
}

\numberwithin{equation}{section}
\begin{document}
\title{Separation Logic of Generic Resources via Sheafeology}
\author[B.~van~Starkenburg]{Berend {van~Starkenburg}}
\orcid{0009-0001-5379-193X}
\affiliation{%
\institution{Leiden University}
\department{Leiden Institute of Advanced Computer Science}
\city{Leiden}
\country{The Netherlands}}
\email{bcrvan.starkenburg@liacs.leidenuniv.nl}
\author[H.~Basold]{Henning Basold} 
\orcid{0000-0001-7610-8331}
\affiliation{%
\institution{Leiden University}
\department{Leiden Institute of Advanced Computer Science}
\city{Leiden}
\country{The Netherlands}}
\email{h.basold@liacs.leidenuniv.nl}
\author[C.~Ford]{Chase Ford} 
\orcid{0000-0003-3892-5917}
\affiliation{%
\institution{Leiden University}
\department{Leiden Institute of Advanced Computer Science}
\city{Leiden}
\country{The Netherlands}}
\email{m.c.ford@liacs.leidenuniv.nl}
\begin{abstract}
Separation logic was originally conceived in order to make the verification of pointer programs scalable to large systems and it has been extremely effective at that.
The principle idea is that programs typically do not have to interact with the whole available memory, but only small parts on which reasoning can thus be focused.
This idea is implemented in separation logic by extending first-order logic with separating connectives, which inspect local regions of memory.
It turns that this approach not only applies to pointer programs, but also to programs involving other resource structures.
Various theories have been put forward to extract and apply the ideas of separation logic more broadly.
This resulted in algebraic abstractions of memory and many variants of separation logic for, e.g., concurrent programs and stochastic processes.
However, none of the existing approaches formulate the combination of first-order logic with separating connectives in a theory that could immediately yield program logics for different resources.
In this paper, we propose a framework based on the idea that separation logic can obtained by making first-order logic resource-aware.
First-order logic can be understood in terms of categorical logic, specifically fibrations.
Our contribution is to make these resource-aware by developing categorical logic internally in categories of sheaves, which is what we call sheafeology.
The role of sheaves is to model views on resources, through which resources can be localised and combined, which enables the scalability promised by separation logic.
We contribute constructions of an internal fibration in sheaf categories that models predicates on resources, and that admits first-order and separating connectives.
Thereby, we attain a general framework of separation logic for generic resources, a claim we substantiate by instantiating our framework to various memory models and random variables.
\end{abstract}

\received{10 July 2025}
\received[revised]{}
\received[accepted]{}
\maketitle							
\input{introduction}
\input{background}     
\input{framework-internal-logic}
\input{interpreting-logic-generalised}
\input{probability-sheaves}
\input{conclusion}
\printbibliography{}
\clearpage
\appendix
\textbf{\huge Appendix}
\input{notation}
\input{fibration}
\input{lan.tex}
\section{Proof details}
\label{sec:proofdetails-internalcategories}
This appendix section provides detailed proofs corresponding to the results presented in ~\cref{sec:internalcategoriessheaves}.

\printProofs[internal-categories-sheaves]
\section{Proof details}
\label{sec:proofdetails-internallogic}
This appendix section provides detailed proofs corresponding to the results presented in ~\cref{sec:internallogic}.

\printProofs[internal-logic]
\end{document}

%% file: introduction.tex
\section{Introduction}
\label{sec:introduction}

A ubiquitous challenge in program verification is to achieve scalability for reasoning about the use and sharing of program resources.
Prototypical examples of such resources include memory with pointer structures that are shared between concurrent processes and the combination of processes with probabilistic behaviour.
In both cases, scalability of reasoning can be achieved via logical connectives that can separate processes and their local resources, while allowing properties of processes on local resources to be combined into properties of composed processes.
For pointer programs, one can use Reynold's separation logic, while probabilistic separation logic can handle compositional reasoning on stochastic processes.
In this paper, we propose a separation logic of generic resources based on categorical logic internal in categories of sheaves, or \emph{sheafeology} for short. 
The ensuing framework affords a uniform perspective on a host of memory models and probabilistic behaviour by instantiation.

Separation logic and its semantics comes in many flavours, not only differing in the kind of resources one can reason about but also how they are treated.
For instance, memory models can be strict (all memory cells allocated), partial (allocation at runtime) and support pointer structures (locations can be stored in memory).
Also the differences between finite, finitely supported and infinite memory have been debated~\cite{dHd23:LogicSeparationLogic, HT16:CompletenessFirstOrderAbstract, Weber04:MechanizedProgramVerification}.
Similar questions come up for probabilistic reasoning: Does one need finite distributions, countable distributions or proper measures~\cite{LiEA24}?
The combination of probabilistic memory and the introduction of concurrency then multiplies the options.
In order to push the idea of computing with resources and using separation logic to scale reasoning about localisation further, it is clear that we need a better understanding of how resourceful computations and localised reasoning via separation work for generic resources.
To this end, we propose sheafeology as a common framework for the uniform treatment of different models and flavours of separation logic.
In what follows, we briefly outline sheafeology and how it underpins the framework sketched above.
Sheafeology combines two main ideas.

\begin{wrapfigure}[11]{r}{2.8cm}
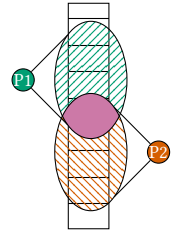

  \centering
  \vspace*{-0.3cm}
  \memoryViews{}
  \vspace*{-0.2cm}
  \caption{\small Two views on regions of shared memory}
  \label{fig:memory-views}
\end{wrapfigure}
\paragraph{Sheaves}
Resources in separation logic have the crucial property that views on them can be localised and compatible views can be combined.
This is illustrated in \cref{fig:memory-views} on the right, which displays the views of two processes P1 and P2 on shared memory.
These processes can operate independently on the disjoint regions but also interact via the overlapping region.
If the overlapping region is empty, then separating conjunction can immediately localise reasoning to the two processes.
However, if the region is inhabited, one needs to employ a weak form of separating conjunction~\cite{BrotherstonEA20} that requires compatibility between the actions of the two processes.

The language of sheaves provides a uniform viewpoint on both cases.
In the simplest case, suppose that $\Loc$ is some set of memory locations and $\Val$ a set of values that can be stored in memory.
For a region $U \subseteq \Loc$, we can model the possible assignments of values to locations as the set $\IntHom{U}{\Val}$ of maps $\sigma \from U \to \Val$.
These are also called (strict) \emph{heaps} and we write $MU$ for $\IntHom{U}{\Val}$.
Heaps can be \emph{localised} to sub-regions $V \subseteq U$ simply by restricting $\sigma \in MU$ to $\restr{\sigma}{V} \in MV$.
Given regions $U_{1}$ and $U_{2}$ for the two processes and heaps $\sigma_{1} \in M U_{1}$ and $\sigma_{2} \in MU_{2}$, we can \emph{amalgamate} these heaps into one heap $\sigma \in M(U_{1} \cup U_{2})$ if $\restr{\sigma_{1}}{U_{1} \cap U_{2}} = \restr{\sigma_{2}}{U_{1} \cap U_{2}}$.
In this case, we define $\sigma$ as follows.
\begin{equation*}
  \sigma(x) =
  \begin{cases}
    \sigma_{1}(x), & x \in U_{1} \\
    \sigma_{2}(x), & x \in U_{2}
  \end{cases}
\end{equation*}

Without going into much detail, we can equip the powerset $\Pow{\Loc}$ with a map $J$ that assigns to each $U \subseteq \Loc$ a set $J(U)$ of \emph{covers} of $U$.
A cover of $U$ is a family $\set{U_{i}}_{i \in I}$ with $U_{i} \subseteq U$ and $\bigcup_{i \in I} U_{i} = U$ for an arbitrary index set $I$.
The strict heaps form a \emph{sheaf} for $(\Pow{\Loc}, J)$, in the sense that for all covers $\set{U_{i}}_{i \in I} \in J(U)$ and all heaps $\sigma_{i} \from U_{i} \to \Val$ with $\restr{\sigma_{i}}{U_{i} \cap U_{j}} = \restr{\sigma_{j}}{U_{i} \cap U_{j}}$ for all $i, j \in I$, there is a unique $\sigma \from U \to \Val$ with $\restr{\sigma}{U_{i}} = \sigma_{i}$.

The sheaf-theoretic view allows us to handle different flavours of memory models in one framework.
Partial memory is obtained by moving from total maps in $MU$ to \emph{partial maps} given by
\begin{equation*}
  M_{p}U = \IntHom{U}{\Val + \set{\bot}} \, ,
\end{equation*}
where $\Val + \set{\bot}$ adjoins a new element $\bot$ to $\Val$ that signals ``undefined''.
Common models of separation logic further restrict attention to finitely supported heaps $M_{f}U \subseteq \IntHom{U}{\Val + \set{\bot}}$, where $\sigma \in M_{f}U$ if the set $\setDef{x \in U}{\sigma \neq \bot}$ is finite.
These finitely supported heaps form a sheaf for \emph{finite covers} $J_{f}(U)$, which are covers over finite index sets.
Similarly, one can restrict finite probability distributions using renormalisation and amalgamate compatible distributions on finite covers.
The more general case of probability measure requires more care, but can also be handled using sheaves~\cite{LiEA24}.
This sheaf-theoretic view thus enables us to model different flavours of resources in one framework, where covers determine the local viewpoints on resources and sheaves model resources that can be restricted and amalgamated according to those viewpoints.

\paragraph{Internal categorical logic}

Separation logic typically comprises three kinds of connectives: first-order logic atoms and connectives, separating connectives and model specific atoms like the points-to predicate~\cite{Reynolds02,Reynolds05}:
\begin{equation*}
  \phi \Coloneqq
  \textcolor{WongDarkPink!80!black}{\underbrace{s \pointsTo t \mid \dotsm}_{\text{model specific atoms}}}
  \quad \mid
  \textcolor{WongRedOrange}{\underbrace{\phi \sepconj \phi \mid \phi \sepimp \phi}_{\text{separating connectives}}}
  \quad \mid \quad
  \textcolor{WongDarkBlue}{
    \underbrace{\bot
      \mid s \doteq t
      \mid \phi \land \phi
      \mid \phi \to \phi
      \mid \all{x} \phi
      \mid \dotsm
    }_{\text{first-order logic atoms and connectives}}
  }
\end{equation*}

An established categorical logic view on first-order logic is via fibrations~\cite{Jacobs99}.
Let us illustrate the idea on the fibration of predicates over a category $\SetC$ of sets and maps.
We define the category $\SetPred$ to have as objects pairs $(X, Q)$ where $Q \subseteq X$ is a predicate, and its morphisms $(X, Q) \to (Y, R)$ are maps $f \from X \to Y$ such that the image of $Q$ under $f$ lies in $R$ or, in symbols, $\img{f}(Q) \subseteq R$.
This category comes with a projection functor $p \from \SetPred \to \Set$ with $p(X, Q) = Q$.
The projection allows us to partition $\SetPred$ into its \emph{fibres} $\SetPred_{X}$, which are subcategories containing only objects of the form $(X, Q)$ and morphism $(X, Q) \to (X, R)$ of the form $\id_{X}$ with $Q = \img{\id_{X}}(Q) \subseteq R$.
General substitution in first-order logic is handled via \emph{reindexing}, but let us focus on the special case of \emph{weakening}.
Let $\pi_{Y} \from X \times Y \to X$ be the projection from the product onto~$X$.
This map induces a functor $\reidx{\pi_{Y}} \from \SetPred_{X} \to \SetPred_{X \times Y}$ by taking preimages:
$\reidx{\pi_{Y}}(X, Q) = (X \times Y, \preimg{\pi_{Y}}(Q))$.
Explicitly, 
\begin{equation*}
  \preimg{\pi_{Y}}(Q) = \setDef{(x, y)}{x \in Q} \, .
\end{equation*}
That is, the functor $\reidx{\pi_{Y}}$ implements weakening.
The other way around, we can define a functor $\forall_{Y} \from \SetPred_{X \times Y} \to \SetPred_{X}$ that is a right adjoint of $\reidx{\pi_{Y}}$ as follows:
\begin{equation*}
  \forall_{Y}(X \times Y, Q) = (X, \setDef{x \in X}{\all{y \in Y} (x, y) \in Q}.
\end{equation*}
Fibrations provide a general definition for reindexing, which enables an abstract theory of existential and universal quantification as, respectively, left and right adjoint functors to weakening functors.
The equality predicate can also be handled in a similar spirit.
Finally, the propositional connectives are functors on fibres that are preserved by reindexing.
For instances, conjunction is a functor $\land_{X} \from \SetPred_{X} \times \SetPred_{X} \to \SetPred_{X}$, such that $Q \land_{Y} R$ is a categorical product for all $Q, R \subseteq Y$ and such that for all $f \from X \to Y$ we have $\reidx{f}(Q \land_{Y} R) = \reidx{f}(Q) \land_{X} \reidx{f}(R)$.
The latter equation should be compared to the usual recursive definition of substitution with $(\phi \land \psi)[\gamma] = \phi[\gamma] \land \psi[\gamma]$.
Fibrations provide us with a framework for first-order logic, but they are lacking the separating connectives.

In order to integrate separating connectives with fibrations, we need to make the latter resource-aware.
Simply put, we develop categorical logic internally in categories of sheaves, the advertised sheafeology, and prove that there is an internal analogue of the predicate fibration.
To be slightly more precise, given a pair $(\CC, J)$ of a small category with covers, called a \emph{site}, we show that there are internal categories $\Univ$ and $\Pred$ in the category $\Sh[1]{\CC, J}$ of sheaves on this site that together form an internal fibration $\Pred \to \Univ$.
Here, $\Sh[1]{\AC, J}$ is the category of sheaves mapping to a category $\Set[1]$ of large sets, whose set-theoretic technicalities are deferred to \cref{sec:prelim}.
The sheaf $\Univ$ can be seen as replacing the category $\Set$ by a \emph{universe}, over which we can model resource-aware sets like memory or probability spaces.
Using the internal category $\Pred$ of predicates on resource-aware sets, we establish general conditions that allow us to construct internal monoids on such predicates that resemble resource monoids and yield our generic model of separation logic.

Finally, the model-specific atoms are interpreted as internal predicates for the corresponding resource model.
For instance, the points-to predicate can be interpreted internally as element of the fibre $\Pred_{M}$, $\Pred_{M_{p}}$ or $\Pred_{M_{f}}$, depending on the flavour of memory model that is chosen.
In the first case, the points-to does not include a check if a location is allocated, since allocation is not part of the model.
However, allocation is tested in the latter two cases.
For probability distributions, a specific atom would be the comparison of a random variable against a fixed distribution~\cite{LiEA24}.
By recasting separation logic via sheafeology, that is, categorical logic internal to sheaf categories, we obtain a framework for modelling generic resources and separation logic for reasoning about them.


\paragraph*{Related Work}
\label{sec:related-work}
This paper develops a categorical logic internally in sheaf categories to reason about resources such as memory and probabilistic data.
Instead of proposing new logics tailored to specific applications, we present a general semantic framework in which such logics arise internally from categorical structure.
This includes familiar instances such as separation logic~\cite{Reynolds02,OHearnEA09} and BI~\cite{IO01}, which introduced spatial connectives for heap ownership.
Pym~\cite{CP09:AlgebraLogicResourcebased, Pym02:SemanticsProofTheory} models the multiplicative structure via Day convolution on a resource category.
In contrast, present it as monoids on predicates over internalised small sheaves, which model resources in an ambient large sheaf category.
Logical connectives are given by Pym using the Heyting algebra on the subobject classifier, while we use the Heyting algebra structure on internalised predicates on small sheaves.

Separation algebras~\cite{COY07:LocalActionAbstract, DHA09:FreshLookSeparation} provide an axiomatisation for separation logic. 
The local view and the separating conjunction as a monoid reminds of our set up. However, our framework has a different perspective for separation, while the disjointness axiom is not enforced.

Our work extends categorical logic based on fibrations~\cite{Jacobs99} to internal categories using 2-categorical techniques~\cite{LR20}.
Predicates are elements in fibres of internal fibrations, and the logical connectives are operations on those fibres.
In particular, separating conjunction becomes a monoid for Day convolution~\cite{Day70}.
Pointer semantics derived from BI-hyperdoctrines~\cite{BieringEA07} feature as special cases.
Convolution-based models that assume disjointness~\cite{RW25} can be recovered by using a site with coproducts which is the case for the Schanuel topos~\cite{LiEA23}.
Unlike the internal logic of a topos~\cite[Chap.~VII]{MM94:SheavesGeometryLogic}, where predicates correspond to global subobjects classified by the subobject classifier, our framework defines a fibered internal logic over a resource sheaf. 
Logical structure arise not only from the Heyting algebra of subobjects, but also from monoidal reindexing and Day convolution, making the logic resource-aware.

In our framework, both weak and strong models of separating conjunction~\cite{BrotherstonEA20, PZ22} can be realised by varying the monoidal structure on the resource sheaf: disjoint gluing yields strong separation, while compatible amalgamation gives rise to weak separation.

Iris~\cite{KTB17:InteractiveProofsHigherorder,JungEA18} and Actris~\cite{HBK22:ActrisAsynchronousSessionType} are frameworks for higher-order concurrent separation logic, supporting modularity, invariants, and mechanised verification in Rocq based on the implementation of semantic approaches to separation logic and step-indexing in Rocq. 
Even though our framework is based on abstract category theory, it may inform new implementation techniques to enable new application areas of Iris and Actris.

In probabilistic settings, separation logic has been interpreted via product~\cite{BartheEA19} and independent combination~\cite{LiEA23}. 
Li et al.~\cite{LiEA24} capture the models as resource models for Day convolution in sheaf categories and unite them. We adopt a similar categorical setup and recover the product model internally. Other sheaf-based approaches to probabilistic reasoning, such as those using atomic sheaves for conditional independence\cite{Simpson24}, are also relevant.

The internal setting is well-suited for reasoning about computational effects.
Prior work includes local state monads~\cite{MM15,PP01} and algebraic theories~\cite{Power05,HP06}.
Extending our framework to internalise these structures could yield sheaves of predicate transformers~\cite{AK20}, 
supporting weakest precondition reasoning.
We also envision use cases where program logics for effects might be captured semantically within our framework~\cite{ND21}.

Game-semantic treatments of concurrency~\cite{MS18,MS18b,MS20} model interaction via execution traces.
Our framework differs in focus: it provides an internal account of predicate and resource structure, rather than externally specifying execution behaviour or communication.
Our approach opens up the possibility of treating resource-aware computations and reasoning in one framework, which will enable us to provide sheaf models of concurrent~\cite{TO15} and distributed systems~\cite{SS09} that have previously been treated externally on an application-specific basis. 

\todo[inline, disable]{~\cite{ND21} Could be a potential use case but I did not yet have time to analyse how and when.}

\paragraph*{Outline}
\label{sec:outline}
\cref{sec:prelim} reviews the basic ingredients, such as set-theoretic assumptions, category-theoretic terminology, sheaves and internal monoids, and outlines the view on resources as sheaves.
In \cref{sec:internalcategoriessheaves}, we introduce the main technical tools to carry out sheafeology, namely internal categories and fibrations, provide the main technical constructions of an internal universe and predicate fibration, and illustrate these constructions on basic memory models and resource-aware first-order logic.
\cref{sec:internallogic} elaborates on the view on separating connectives as internal monoids on the predicates that stem from the fibration in the previous section.
We round off this section with model-specific atoms for memory, and then provide an extensive example for probabilistic reasoning in \cref{sec:probability}.
In \cref{sec:conclusion}, we give a brief summary and an outlook on future work.

%% file: background.tex
\section{A Sheaf-Theoretic View on Memory}
\label{sec:prelim}
While we assume familiarity with basic concepts of category theory~\cite{MacLane98:CatWorkingMath,Borceux94,Riehl16:CatTheoryinContext}, we will use this section to review the necessary background on sheaves~\cite[Chapter 3]{MM94:SheavesGeometryLogic}.
We will incrementally develop a sheaf-based model of memory, in line with~\cref{sec:introduction}, in the examples.

Category theory is, form the logical perspective, a first-order axiomatisable theory.
However, subtleties arise when we want to reflect category theory back into itself.
A famous instance of this are the hom-functor and the Yoneda embedding, both of which require that the morphisms between to fixed objects are elements of some category of sets.
Since this plays a central role in sheaf theory and thus our programme of resource-aware logic, we first need to fix our set-theoretic assumptions.
There are different possible choices of foundations for category theory, like Grothendieck universes~\cite[Expos\'{e} 1]{GV72}, each having different strengths~\cite{Schulman08}.
Our approach is inspired by the elementary theory of categories and sets (ETCS)~\cite{Leinster14:RethinkingSetTheory} and cumulative universes in type theory:
We assume that for every natural number $k$ there is a category $\Set[k]$ of \emph{$k$-sets and maps}, such that 1) each category is closed under the usual set-theoretic constructions, such as powerset, limits and colimits, 2) $\Set[k] \subseteq \Set[k+1]$ and 3) the collection of objects of $\Set[k]$ is an object in $\Set[k+1]$.
By default, $\Set$ will refer to $\Set[0]$ and we may call its elements \emph{small sets}.


We then use the following terminology for categories.
A \emph{$k$-category} $\CC$ comprises a set $\CC_{0} \in \Set[k]$ of objects and for all $A, B \in \CC_{0}$ a set $\CC(A, B) \in \Set[k]$ of morphisms, subject to the usual axioms of categories.
If $\CC$ is a $(k+1)$-category, such that $\CC(A, B) \in \Set[k]$, then we call $\CC$ locally $k$-small.
As usual, $0$-categories are called \emph{small} and locally $0$-small categories are just called locally small.


For categories $\CC$, we employ the usual notation:
We shorten $A \in \CC_{0}$ to just $A \in \CC$.
Morphisms $f \in \CC(A, B)$ are denoted by $f \from A \to B$, and
we write $g\comp f$ or $gf$ for the composite morphism of $g \from B \to C$ with $f$.
The identity morphism of $A \in \CC$ is denoted $\id_A$.
The \emph{opposite category} of~$\CC$ is denoted $\op{\CC}$: it has the same objects as~$\CC$,
and~$\mor{\op{\CC}}(A, B) = \mor{\CC}(B, A)$ for all objects~$A,B$ in $\CC$.

\subsection{Presheaves}
\label{sec:presheaves}

Presheaves formalise the first important ingredient to separation logic: the localisation of resources to along relations of views.
Given a $k$-small category $\CC$, we denote by $\FunCat{\CC}{\Cat{B}}$ the locally $k$-small category with functors from~$\CC$ into~$\Cat{B}$ as objects and natural transformations between them as morphisms.
A \emph{$k$-presheaf on~$\CC$} is an object in the category
\begin{equation*}
\PSh[k]{\CC} \coloneqq \FunCat{\op{\CC}}{\Set[k]}.
\end{equation*}
Unless we need to explicitly handle set levels, we will leave out the subscript level and just write $\PSh{\CC}$ instead of $\PSh[k]{\CC}$.
We interpret presheaves in this paper as follows.
The objects $A$ in $\CC$ model views on resources and morphisms their relations.
Given a presheaf $F \in \PSh{\CC}$, we understand $FA$ as the set of resources for the view $A$, while for a morphism $f \from A \to B$ we can read the map $Ff \from FB \to FA$ as restriction or localisation.

\begin{example}
\label{expl:powerset}
For a set~$X$, we view the powerset $\pow X$ as the small category with subsets~$A\subseteq X$ as objects and the inclusions~$\iota_{A, B}\colon A\hra B$ as morphisms (also denoted~$\iota$ if~$A$ and~$B$ are clear from context).
Thus, there is a unique morphism~$A\to B$ if~$A\subseteq B$, and otherwise~$\pow X(A, B)$ is empty.

Fix a set $\Loc$ of \emph{locations} and a set $\Val$ of \emph{values}, and let $\caL$ be the powerset $\pow\Loc$.
We define the \emph{strict memory presheaf} $M\colon\op{\caL}\to\Set$ as follows.
For $U\subseteq\Loc$ set $M(U) \coloneqq [U, \Val]$, where $\IntHom{U}{\Val}$ is the set of all maps~$U\to\Val$.
The action $M$ is given by pre-composition with inclusions, which corresponds to restriction:
\begin{equation*}
M(\iota_{U_1, U_2})\colon [U_2, \Val]\to[U_1, \Val], \qquad
M(\iota_{U_1, U_2})(f) \coloneqq (U_1\xra{\iota}U_2\xra{f}\Val).
\end{equation*}
\end{example}

Resources that are independent of the view can be modelled by constant presheaves.
\begin{example}
\label{expl:constant}
For any set~$X$, the assignment~$A\mapsto X$ is the object-part of a presheaf on~$\Cat{A}$ with the action defined by~$K(f):=\id_X$ for every morphism~$f\colon A\to B$.
We refer to~$K$ as the \emph{constant presheaf at~$X$}.
The constant presheaf at a terminal object~$1:=\{*\}$ of~$\Set$ is a terminal object in~$\PSh{\Cat{A}}$.
\end{example}

A essential tool in understand presheaves is given by viewing relations between views themselves as resources, which then allows us to probe resources.
\begin{example}
  \label{expl:yoneda}
  For a locally $k$-small category $\CC$, we refer to the functor~$\yo\colon\CC\to\PSh[k]{\CC}$, which sends $A\in\CC$ to the \emph{corepresentable presheaf}~$\CC(-, A)\colon\op{\CC}\to\Set$, as the \emph{Yoneda embedding} ($\yo$ is pronounced ``yo'').
  Explicitly, on objects $\Yo(A)(B) = \CC(B, A)$ and it acts on a morphism $f\colon B\to C$ by pre-composition: $\Yo(A)(f)(g) = B\xra{f} C\xra{g} A$.
We often write $\Yo_{A}$ for the application of $\Yo$ to $A$.
A fundamental result is the so-called Yoneda lemma, which asserts that the elements of a presheaf $F$, that is, the local resources, can be probed by natural transformations, in that there is an isomorphism $FA \cong \PSh[k]{\CC}(\Yo_{A}, F)$ of sets that is natural in $A$.
\end{example}

\subsection{Sheaves}
\label{sec:sheaves}

As outlined in the introduction, the second ingredient that we need, besides localisation, is the amalgamation of resources.
This is an operation supported by sheaves, which form a full subcategory of presheaves.
In order to define them, we need to formalise the notion of cover, for which we will use the general notion of Grothendieck coverages that capture the essential properties of covers on topological spaces but allow for arbitrary categories instead of merely a poset of open sets.

Let $\CC$ be a (small) category.
A \emph{sieve} of an object~$A\in\CC$ is a family $S$ of morphisms of the form $S = \set{f \from A_f\to A}$ that is closed under pre-composition:~$f \comp g \in S$ for every~$f \in S$ and every morphism~$g\colon B\to A_f$ in~$\CC$.
We note that a sieve of~$A$ is equivalently a \emph{subfunctor} $S\hra\yo_{A}$, which is a functor with a natural transformation into the Yoneda presheaf (\cref{expl:yoneda}) with all components being inclusion maps.

Given a sieve $S$ of $A$, and a morphism~$h\colon B\to A$ in~$\CC$, we write~$\pbCover{h}(S)$ for the \emph{pullback} of $S$ along $h$, given by $\pbCover{h}(S) = \setDef{g \from C \to B}{h \comp g \in S}$.
A \emph{(Grothendieck) coverage} of~$\CC$ is an assignment of a set~$J(A)$ of sieves to each~$A\in\CC$ satisfying the following \emph{saturation conditions}.
\begin{enumerate}
\item
the \emph{maximal sieve}~$\Yo_{A}$ is in~$J(A)$;
\item
$\pbCover{h}(S)\in J(B)$ for each morphism~$h\colon B\to A$ and each $S\in J(A)$;
\item
$R\in J(A)$ if there exists~$S\in J(A)$ such that~$\pbCover{h}(R)\in J(B)$ for all~$(h\colon B\to A)\in S$.
\end{enumerate}
We refer to these axioms as \emph{maximality, stability}, and \emph{transitivity}, respectively.
If we speak of a sieve that is member of a specified coverage, then we call it a \emph{cover}.

\begin{definition}
\label{defn:sheaf}
A \emph{(Grothendieck) site} is a pair~$(\Cat{C}, J)$ where~$\Cat{C}$ is a small category~$\Cat{C}$ and~$J$ is a coverage.
We call~$(\Cat{C}, J)$ \emph{cartesian} if~$\Cat{C}$ is finitely complete.
For~$F\in\PSh[k]{\Cat{C}}$ a ($k$-)presheaf, an \emph{$F$-compatible family} for a cover $S = \set{f\colon A_f\to A}$ in $J(A)$ is a family of points~$\set{x_f\in F(A_f)}_{f \in S}$ such that for every commutative square of the form
\begin{equation*}
\begin{tikzcd}[column sep =20mm]
B\arrow[r, "h"]\arrow[d, swap, "k"]	&A_g\arrow[d, "g"]		\\
A_f\arrow[r, "f"]					&A
\end{tikzcd}
\end{equation*}
we have~$F(k)(x_f) = F(h)(x_g)$ in~$FB$.
A \emph{($k$-)sheaf} over~$(\Cat{C}, J)$ is a presheaf~$F\in\PSh[k]{\Cat{C}}$ such that for any object~$A\in\Cat{A}$ and any $F$-compatible family $\set{x_{f}}_{f \in S}$ for cover $S \in J(A)$, there exists a unique \emph{amalgamation} $a\in FA$, such that

\noindent ~$Ff(a) = x_f$ for all~$f \in S$.
We write~$\Sh[k]{\Cat{C}, J}$ for the full subcategory of~$\PSh[k]{\Cat{C}}$ spanned by sheaves over~$(\Cat{C}, J)$.
\end{definition}

Various concepts of coverage are found in the literature on sheaves which are obtained by dropping some, or even all, of the saturation conditions detailed above (see, e.g., Johnstone~\cite{Johnstone02} for a comprehensive overview).
First of all, the condition that a cover is closed under pre-composition is not fulfilled by the usual understanding of cover on open sets of a topological space.
In general, let us call a family of morphisms $(f \from A_f\to A)$ a \emph{pre-cover} of an object~$A$.
A sieve can be obtained from a pre-cover by closing under pre-composition, which amounts in the case of pre-covers in topological spaces to taking the downwards closure.
Second, the saturation conditions on coverages can be weakened.
By a \emph{pre-coverage} we understand an assignment of a set of pre-covers~$J(A)$ to each object~$A$ satisfying the following condition:
for every pre-cover $(f \from A_f\to A)\in J(A)$ and every morphism~$h \from B\to A$ in~$\Cat{A}$, there exists a pre-cover~$(g \from B_g\to B)\in J(B)$ such that for every $g$ there is an $f$ and a $k \from B_{g} \to A_{f}$ with $fk = hg$.
Given a pre-coverage, we can always saturate it to obtain a Grothendieck coverage.
The latter are simpler to use in proofs, while pre-coverages are often easier to obtain in practice.
These notions of coverage are essentially equivalent: every pre-coverage induces a unique Grothendieck coverage with equivalent categories of sheaves.
We pass freely between these equivalent presentations of sheaves without further mention.

\begin{example}
\label{expl:site-powerset}
We obtain a coverage $J$ of~$\pow X$ (\cref{expl:powerset}) by defining letting $J(A)$ consist of downwards closed covers $U$, in the topological sense, of $A$ by subsets of $A$: 
\begin{equation*}
  \bigcup_{f \in U} A_{f} = A \text{ and if } B \subseteq A_{f} \text{ then } \iota_{B, A} \in U
\end{equation*}
\end{example}

We can use the site on the powerset to model the resource of strict memory as sheaf.
\begin{example}
\label{expl:memory-cover}
Consider the site~$(\caL, J)$ where~$J$ is the coverage of~\cref{expl:site-powerset} above.
A~$M$-compatible family (\cref{expl:powerset}) of a cover~$S = \set{f\colon A_f\to A}$ is given by a family of points
$a_f\in[A_f, \Val]$ such that for any pair~$f,g\in S$ and any~$B\subseteq A_f \cap A_g$ we have
\begin{equation*}
M(\iota_{B, A_f})(a_f) = a_f\comp \iota_{B, A_f}
					       = a_g\comp\iota_{B, A_g} = M(\iota_{B, A_g})(a_g).
\end{equation*}
That is, the restriction of $a_{f}$, given by $\restr{a_f}{B} = (B\hra A_f\xra{a_f} \Val)$, and the restriction $\restr{a_g}{B}$ are equal.
This shows that $M$ is a sheaf, which we call the \emph{strict memory sheaf}.
\end{example}

So far, we have only used the powerset and a coverage that corresponds to the discrete topology.
However, it is crucial that we allow general sites in order to enable, e.g., finitary memory models.

\begin{example}
  \label{ex:partial-finitary-memory}
  Let $\caL$ again be the powerset of locations and $\Val$ the set of values from \cref{expl:powerset}.
  A \emph{partial memory} or \emph{heap} on $U \subseteq \Loc$ is a map $\sigma \from U \to \Val + \set{\bot}$ into the coproduct of $\Val$ with $\set{\bot}$.
  The fresh element $\bot$ is used to signal that a location is not allocated.
  We define the \emph{support} of $\sigma$ by $\supp \sigma = \setDef{x \in U}{\sigma(x) \neq \bot}$.
  This allows us to define the presheaves $M_{p}$ and $M_{f}$ on $\caL$ of heaps and \emph{finitary heaps} as follows.
  \begin{equation*}
    M_{p} U = \IntHom{U}{\Val + \set{\bot}}
    \quad \text{and} \quad
    M_{f} U = \setDef{\sigma \in M_{p}U}{\supp \sigma \text{ is finite}}
  \end{equation*}
  The restriction along inclusion works precisely as for the strict memory.
  Since $M_{p}U$ contains all maps, just like $M$, it is a sheaf for $(\caL, J)$ from \cref{expl:memory-cover}.

  However, $M_{f}$ is not a sheaf on that site because an arbitrary family of finitary heaps cannot be amalgamated into one finitary heap.
  Let us instead define the coverage $J_{f}$ of finite covers.
  This works for an arbitrary distributive lattice, not necessarily bounded, but we just illustrate it for $\caL$.
  A finite pre-cover of a set $U \subseteq \Loc$ is a finite set $S \subseteq \Pow\caL$, such that $\bigcup S = U$.
  We obtain from such a pre-cover $S$ a cover by taking the downwards closure $\downset{S}$ under the inclusion order.
  One can easily check that the assignment
  \begin{equation*}
    J_{f}(U) = \setDef{\downset{S}}{S \text{ finite pre-cover of } U}
  \end{equation*}
  yields a coverage and thus a site $(\caL, J_{f})$.
  Moreover, $M_{f}$ is a sheaf on this site because for every $\downset{S} \in J_{f}(U)$ and compatible family $\sigma_{V} \in M_{f}V$, we can simply define its amalgamation $\sigma$ by $\sigma(x) = \sigma_{V}(x)$ for $V \in S$ and $x \in V$.
  This map is finitely supported because $S$ is finite and the unique amalgamation just as in \cref{expl:memory-cover}.
\end{example}

Other models for finitary heaps are obtained using sheaves on (a small skeleton on) the opposite category of finite sets and injective maps with an appropriate site~\cite{LiEA24}.
Similarly, one cane view random variables as resources using a sheaves on a specific site, see \cref{sec:probability}.

\subsection{Monoidal Categories and Day Convolution}
\label{sec:day-convolution}

Previous algebraic abstractions of separation logic have given semantics in terms of separation algebras~\cite{COY07:LocalActionAbstract, DHA09:FreshLookSeparation} and resource monoids~\cite{CP09:AlgebraLogicResourcebased}.
We will integrate separating conjunction as well as a monoid operation, but internally in sheaf categories.
Crucial to the usual notion of a monoid carried by a set $A$, is that the monoid operation can be seen as a map $A \times A \to A$ on the product and the monoid unit as a map $\ast \to A$ from a singleton set to $A$.
To make such operations resource-aware and internalise them into the category of sheaves, we need an analogue of the product and the singleton set.
One could use just products of sheaves, as the category $\Sh{\CC, J}$ is complete, but that is not the right perspective for separating conjunction.
For instance, for a sheaf $\Pred_{M}$ of predicates on the memory sheaf $M$, which we can informally see as sheaf with $\Pred_{M}(U) = \Pow{MU}$, separating conjunction aims to split the locations across two predicates to be combined:
\begin{equation*}
  P_{1} \ast_{U_{1}, U_{2}} P_{2} = \setDef{\sigma \in M (U_{1} \cup U_{2})}{\restr{\sigma}{U_{k}} \in P_{k}},
  \text{ where } P_{k} \in \Pred_{M}(U_{k})
\end{equation*}
This splitting can be achieved with so-called Day convolution $\Day$, which is a monoidal tensor but not a product on (pre)sheaves.
Our aim is then to establish separating conjunction as a monoid $\Pred_{M} \Day \Pred_{M} \to \Pred_{M}$ for a certain sheaf of predicates on $M$ or more generally on $\Pred_{F}$ for sheaves $F$ that model resources.

Suppose that~$\CC$ comes equipped with the structure of a \emph{symmetric monoidal category}~\cite{MacLane98:CatWorkingMath}.
This means that~$\CC$ has a \emph{tensor bifunctor}~$\tens \from \CC \times \CC \to \CC$,
a \emph{unit object}~$I\in\CC$, and isomorphisms
\begin{equation*}
\alpha \from (A\tens B)\tens C\to A\tens (B\tens C),
\qquad
\lambda \from I\tens A\to A,
\quad\text{and}\quad
\rho \from A\tens I\to A
\end{equation*}
that are natural in~$A,B,C\in\CC$, called the \emph{associator, left unitor}, and \emph{right unitor}, respectively.
Symmetry means that~$\CC$ additionally admits a natural isomorphism~$\beta_{A,B} \from A \tens B\to B \tens A$, called the \emph{braiding}, such that $\beta_{A,B}^{-1} = \beta_{B, A}$.
All this data is subject to coherence laws~\cite[Chapter XI]{MacLane98:CatWorkingMath}.

\begin{example}
\label{expl:powerset-monoidal}
A \emph{preordered monoid}, that is, a monoid $(M, +, 0)$ equipped with a preorder~$\leq$ making the operation~$+$ monotone in each argument, can be seen as a monoidal category whose objects are the elements of $M$, morphism are given by the preorder, the tensor by the monoid operation and the unit object by the monoid unit~$0$.
This results even in a \emph{strict} monoidal monoidal category because the associators and unitors are identity morphisms, and it is symmetric if the monoid operation is commutative.
The poset~$\pow X$ from \cref{expl:powerset} admits the structure of a commutative monoidal preorder with set-theoretic union as monoid operation and the empty set as unit.
In fact, Pym uses this kind of structure as the basis for Kripke semantics of separation logic~\cite{Pym02:SemanticsProofTheory}.
\end{example}

For a symmetric monoidal category $(\CC, \tens, \tensU)$, \textcite{Day70} showed how $\PSh{\CC}$ canonically becomes a symmetric monoidal category.
First, define the \emph{external tensor} of~$F, G \in \PSh{\CC}$ to be the presheaf
\begin{equation*}
  F\etens G \from\op{(\CC\times\CC)}\to\Set
  \quad \text{with} \quad
  (F\etens G)(A, B) \coloneqq F(A)\times G(B).
\end{equation*}

\begin{definition}
\label{defn:Day}
\emph{Day convolution} $\Day$ is the left Kan extension of the external tensor along~$\tens$:
\begin{equation*}
  \begin{tikzcd}
    \op{\CC \times \CC} \ar[d, "F \etens G"{swap}, ""{name=E, anchor=center, pos=0.4}] \rar{\op{\tens}}
    & \op{\CC} \ar[dl, "F \Day G", ""{name=D, anchor=center, pos=0.35}]
    \\
    \Set
    \arrow["\kappa", shorten >=5pt, shorten <=5pt, Rightarrow, from=E, to=D]
  \end{tikzcd}
\end{equation*}
\end{definition}
This makes $(\PSh{\CC}, \Day, \YoP{I})$ the free symmetric monoidal co-completion of $\CC$, such that $\Yo$ is a strong monoidal functor~\cite{IK86:UniversalPropertyConvolution}.
Moreover, $\PSh{\CC}$ is monoidal closed, since for each presheaf $F$, the functor $- \Day F \from \PSh{\CC} \to \PSh{\CC}$ has a right adjoint $F \multimap - \from \PSh{\CC} \to \PSh{\CC}$.

\begin{example}
\label{expl:day}
For the monoidal structure given by union on $\caL$, see \cref{expl:powerset-monoidal}, the Day convolution of presheaves $F, G\in\PSh{\caL}$ reduces to the coproduct, as we have for all $U\in\caL$
    \begin{equation*}
        (F\Day G)(U) \cong \coprod_{U_1 \cup U_2 = U} F(U_1)\times G(U_2)
    \end{equation*}
\end{example}

The following definition specifies what a monoid in a monoidal category is, which will the allow us to  model separating conjunction as a monoid for Day convolution.
\begin{definition}
\label{defn:monoid}
Let~$\CC$ be a monoidal category with tensor~$\otimes$ and unit object~$I$.
A \emph{monoid object} in~$\CC$ is an object~$A\in\CC$ equipped with morphisms
\begin{equation*}
\mu \from A\tens A\to A
\qquad\text{and}\qquad
\eta \from I\to A,
\end{equation*}
called the \emph{multiplication} and the \emph{unit} of the monoid.
These morphisms must satisfy the \emph{associativity law}

\noindent $\mu\comp (\id_{A}\tens\mu)\comp\mu = \mu\comp (\mu\tens\id_{A})$, and the \emph{unit laws}
$\lambda = \mu\comp (\eta\tens\id_{A})$ and~$\rho = \mu\comp (\id_{A}\tens\eta)$.
\end{definition}

A monoid in $(\PSh{\CC}, \Day, \YoP{e})$ over a monoidal category $(\CC, \tens, \tensU)$ is a a presheaf $F \in \PSh{\CC}$ with natural transformations~$\mu: F \Day F \to F$ and~$\eta: \YoP{I} \to F$ satisfying:
\begin{equation*}
  \mu \circ (\mu \Day \id_F) = \mu \circ (\id_F \Day \mu) 
  \quad \text{and} \quad
  \mu \circ (\eta \Day \id_F) = \id_F = \mu \circ (\id_F \Day \eta) 
\end{equation*}





\begin{example}
  \label{ex:partial-memory-monoid}
    The partial memory sheaf $M_{p} \from \op{\caL}\to \Set$, see \cref{ex:partial-finitary-memory}, is a monoid for Day convolution on $(\Sh{\caL}, \cup, \emptyset)$, with monoid product
    \begin{equation*}
      \mu \from M_{p} \Day M_{p} \to M_{p}
      \quad \text{and} \quad
      \eta \from \YoP{\emptyset} \to M_{p}
    \end{equation*}
    given as follows.
    By the Yoneda lemma, $\eta$ corresponds uniquely to an element of $M_{p}\emptyset$, which has to be thus the empty map $\emptyset \to \Val + \set{\bot}$.
    The multiplication is given by first defining a family of maps $m_{U_{1}, U_{2}} \from M_{p}(U_1) \times M_{p}(U_2) \to M_{p}(U)$ for all $U$ and $U_{1} \cup U_{2} = U$ by
    \begin{equation*}
      m_{U_{1}, U_{2}}(\sigma_{1}, \sigma_{2}) =
      \begin{cases}
        \sigma_{1}(x), & \text{if } x \in U_{1} \setminus U_{2} \text{, or } x \in U_{1} \cap U_{2} \text{ and } \sigma_{1}(x) = \sigma_{2}(x) \\
        \sigma_{2}(x), & \text{if } x \in U_{2} \setminus U_{1} \\
        \bot, & \text{otherwise}
      \end{cases}
    \end{equation*}
    and the using the universal property of coproducts to lift this family to
    \begin{equation*}
      \mu_U  \from \coprod_{U_1 \cup U_2 = U} M_{p}(U_1) \times M_{p}(U_2) \to M_{p}(U),
    \end{equation*}
    which essentially implements amalgamation without the compatibility assumption.
    We note that this monoid operation requires the use of partial memory and cannot work for strict memory
\end{example}

Our ultimate aim is to establish monoid operations on predicates on memory and other resources to model separating conjunction.
We will establish in what follows the general picture of how predicates on resources, their logic and finally monoid operations on them can be generally constructed.
Once we have our framework in place, we come back to a detailed discussion of monoid operations on memory predicates in \cref{sec:intepreting-separating-conjunction} and on random variables in \cref{sec:probability}.

%% file: framework-internal-logic.tex
\section{Sheafeology: Categorical Logic Internal in Categories of Sheaves}
\label{sec:internalcategoriessheaves}

In this section, we introduce the main techniques to carry out our programme of resource-aware logic via sheafeology.
After a review of general internal category theory in \cref{sec:internal-categories}, we come to our main technical results.
The first step is construct an internal category $\Univ$ in the category of $1$-sheaves, which we call the universe and which allows us to internalise small $0$-sheaves as objects in this category.
This universe becomes the internal analogue of $\Set$, thus enabling resource-aware sets and logic.
Indeed, we show in \cref{sec:predicates} that there is a internal category $\Pred$ of predicates over $0$-sheaves and with an internal projection functor $\Pred \to \Univ$.
Finally, we show in \cref{sec:internal-fibration} that this functor is an internal fibration, models propositional logic and posses existential quantification, which we then use in \cref{sec:internal-propositional-logic} to interpret the propositional part of separation logic.

\subsection{Internal Categories}
\label{sec:internal-categories}
We proceed to review the necessary material on \emph{internal categories}~\cite[Ch.~8]{Borceux94} and fix notation. 
We illustrate on several examples that will be of interest in the subsequent.
Throughout this section, fix a category~$\EC$ with pullbacks; we call such a category a \emph{base category}.

A \emph{category internal to}~$\EC$ is a pair of objects~$\IntCat{A}=(A_0,A_1)$ in~$\EC$ equipped with morphisms
$s \colon A_1\to A_0$ and $t\colon A_1\to A_0$, the \emph{source morphism} and \emph{target morphism}, an \emph{identity morphism}
$e\colon A_0\to A_1$, and a \emph{composition morphism}

\noindent $c\colon\pb{A_1}{A_1}{A_0}\to A_1$,
where the pullback~$(\pb{A_1}{A_1}{A_0}, p_{\ell}, p_r)$ is formed along the source and target maps as depicted below:
\begin{equation*}
\begin{tikzcd}[column sep = 20mm]
\pb{A_1}{A_1}{A_0}\arrow[r, "p_r"]\arrow[d, swap, "p_{\ell}"]   &A_1\arrow[d, "t"] \\
A_1\arrow[r, "s"]							       &A_0
\end{tikzcd}
\end{equation*}
This data is subject to the following laws:
\begin{align}
  \label{item:id-law}
  & \text{Typing of identities:} && s\comp e = \id_{A_0} = t\comp e \\
  \label{item:comp-type}
  & \text{Typing of composition:} && s\comp p_{r} = s\comp c \text{ and } t\comp p_{\ell} = t\comp c \\
  \label{item:comp-unit}
  & \text{Unit laws:} && c\comp\pair{\id_{A_1}, e\comp s} = \id_{A_1} = c\comp\pair{e\comp t, \id_{A_1}} \\
  \label{item:comp-assoc}
  & \text{Associativity:} && c\comp(\pb{\id_{A_1}}{c}{A_0}) = c\comp(\pb{c}{\id_{A_1}}{A_0})
\end{align}


We see that a category~$\IntCat{A}$ internal to~$\Set$ is a small category by viewing~$A_0$ as the \emph{set of objects} and~$A_1$ as the \emph{set of morphisms}. 
In particular,~$A_1$ is the coproduct $\coprod\IntCat{A}(a, b)$ over all~$a,b\in A_0$.
The morphisms~$s,t\colon A_1\to A_0$ then provide the assignment of a source and target to each morphism, where we view~$f\in A_1$ as a morphism $s(f)\to t(f)$.
The pullback~$\pb{A_1}{A_1}{A_0}$ is isomorphic to
\begin{equation*}
\set{(g,f)\in A_1\times A_1\setsep s(g) = t(f)} \, ,
\end{equation*}
which consists of all \emph{composable pairs} of morphisms in~$A_1$.
The map $c\colon\pb{A_1}{A_1}{A_0}\to A_1$ is then the composition of~$\IntCat{A}$.
We typically write~$gf$ instead of~$c(g,f)$ for~$(g,f)\in\pb{A_1}{A_1}{A_0}$ if confusion is unlikely.
Axiom~(\ref{item:id-law}) states that the identity~$e(a)$ for~$a\in A_0$ is an endomorphism~$a\to a$,
and axiom~(\ref{item:comp-type}) asserts that the composition~$gf$ has type~$s(f)\to t(g)$.
The remaining axioms~(\ref{item:comp-unit}) and~(\ref{item:comp-assoc}) express that~$e$ is a left and right unit of composition, and that composition is associative.
Thus, internal categories generalise the concept of a small category to one relative to the category~$\EC$.

We next describe a suitable concept of morphism between internal categories.
\begin{definition}
\label{def:internal-functor}
Let~$\IntCat{A},\IntCat{B}$ be categories internal to~$\EC$. 
An \emph{internal functor} from~$\IntCat{A}$ to $\IntCat{B}$ is a pair of morphisms

\noindent $(f_0\colon A_0\to B_0, f_1\colon A_1\to B_1)$ such that the following diagram commutes:
\begin{equation*} 
\begin{tikzcd}[column sep = 20mm]
A_0 \arrow[r, "e" description]
	\arrow[d, swap, "f_0"] 
& A_1 \arrow[l, "t", shift left=2]
	  \arrow[l, "s"', shift right=2]
	  \arrow[d, "f_1"] 
&\pb{A_1}{A_1}{A_0} \arrow[l, "c_A"']
	 		  	  \arrow[d, "\pb{s}{t}{f_1}"] 
	   \\
 B_0 \arrow[r, "e" description]                  
&B_1 \arrow[l, "s"', shift right=2] 
			 \arrow[l, "t", shift left=2]                  
&\pb{B_1}{B_1}{B_0} \arrow[l, "c_B"]                  
\end{tikzcd}
\end{equation*}
\end{definition}

The last ingredient to internal category theory is the notion of internal natural transformations.
\begin{definition}
  Let $\IntCat{A}, \IntCat{B}$ be categories internal to $\Int{\EC}$. Given internal functors
  \begin{equation*}
    (f_0, f_1),\ (g_0, g_1) \colon (A_0, A_1) \to (B_0, B_1),
  \end{equation*}
  an \emph{internal natural transformation} from $(f_0, f_1)$ to $(g_0, g_1)$ is a morphism $\alpha\colon A_0 \to B_1$ in $\EC$ such that the following diagrams commute:
\begin{itemize}
\item 
\textbf{Component assignment:} The morphism $\alpha$ assigns to each object of $\IntCat{A}$ a morphism in $\IntCat{B}$ from $f_0$ to $g_0$, i.e.,~$s \circ \alpha = f_0$ and~$t \circ \alpha = g_0.$
This is expressed in the diagram:
    \begin{equation*}
      \begin{tikzcd}
        & A_0 \arrow[dl, "f_0"'] \arrow[dr, "g_0"] \arrow[d, "\alpha"] & \\
        B_0 & B_1 \arrow[l, "s"] \arrow[r, "t"'] & B_0
      \end{tikzcd}
    \end{equation*}

    \item \textbf{Internal naturality condition:} The following square must commute:
    \begin{equation*}
      \begin{tikzcd}[column sep=large, row sep=large]
        A_1 \arrow[r, "\pair{\alpha \comp s,\, g_1}"] \arrow[d, "\pair{f_1,\, \alpha \comp t}"'] & B_1 \times_{B_0} B_1 \arrow[d, "c"] \\
        B_1 \times_{B_0} B_1 \arrow[r, "c"'] & B_1
      \end{tikzcd}
    \end{equation*}
    where~$c\colon B_1 \times_{B_0} B_1 \to B_1$ is the internal composition in $\IntCat{B}$,
    and the pairings are defined by using the universal property of pullbacks as in the following diagrams.
    \begin{equation*}
      \begin{tikzcd}[column sep=large]
        A_1 \arrow[ddr, bend right=25, "g_1"'] \arrow[dr, dashed, "\pair{\alpha \comp s,\, g_1}"{description}] \arrow[r, "s"] & A_0 \arrow[dr, bend left=10, "\alpha"] \\
        & B_1 \times_{B_0} B_1 \arrow[r, "p_r"] \arrow[d, "p_\ell"'] & B_1 \arrow[d, "t"] \\
        & B_1 \arrow[r, "s"] & B_0
      \end{tikzcd}
      \qquad
      \begin{tikzcd}[column sep=large]
        A_1 \arrow[drr, bend left=15, "f_1", below] \arrow[dr, dashed, "\pair{f_1,\, \alpha \comp t}"{description}] \arrow[d, "t", below]\\
        A_ 0 \arrow[dr, bend right=10, "\alpha"] & B_1 \times_{B_0} B_1 \arrow[r, "p_r"] \arrow[d, "p_\ell"'] & B_1 \arrow[d, "t"] \\
        & B_1 \arrow[r, "s"] & B_0
      \end{tikzcd}
    \end{equation*}
\end{itemize}
\end{definition}

Vertical composition of internal natural transformations results in an internal natural transformation, 
leading to a 2-category~\cite{KS74}
$\Int{\EC}$
of categories internal to~$\EC$ as 0-cells, internal functors as 1-cells, and internal natural transformations~\cite{Miranda20}.
\subsection{The Universe Sheaf}
\label{sec:internal-sheaves}
Given a site $(\CC, J)$, we construct a sheaf~$\Univ\in\Sh[1]{\CC, J}$ such that every sheaf in~$\Sh[0]{\CC, J}$ can be internalised as a global element of~$\Univ$, that is, a natural transformation~$\TSheaf\to\Univ$ where~$\TSheaf$ is the terminal sheaf.
Intuitively, $\Univ$ will be our internal universe for resources.

The object part of the universe sheaf is defined by
\begin{equation}
\label{eq:univ-object-def}
\Univ(A) = \Sh[0]{\slice{\CC}{A}, \SliceCoverage{A}} \, 
\end{equation}
where $\SliceCoverage{A}$ is the coverage induced by $J$ on the slice category $\slice{\CC}{A}$ according to the following lemma.

\begin{lemmaE}[][category=internal-categories-sheaves]
\label{lem:slice-coverage}
Let $\Site$ be a Grothendieck site and $A \in \CC$ an object.
The mapping $\SliceCoverage{A}$ given by
\begin{equation*}
\SliceCoverage{A}(p) = \setDef{F \subseteq \Yo_{p}}{\sliceDom[A]F \in J(\sliceDom[A]p)} 
\end{equation*}
is a Grothendieck coverage and thus $\SliceSite{A}$ is a Grothendieck site.
Moreover, if $F(q) \subseteq \slice{\CC}{A}(q, p)$ is a subset of morphism, such that $\sliceDom[A]F \in J(A)$, then $F \in \SliceCoverage{A}(p)$.
\end{lemmaE}
	\begin{proofE}
First, we check pullback stability.
Let $p \from B \to A$, $F \in \SliceCoverage{A}(p)$ with $\sliceDom[A]F = F' \in J(B)$ and $g \from q \to p$.
We obtain a pullback cover $\pbCover{(\sliceDom[A]g)}{F'} \in J(C)$, where $q \from C \to A$.
For $u \from D \to C$ in $\pbCover{(\sliceDom[A]g)}{F'}$ with $g \comp u = f \comp v$ for some $f \from r \to p$ and $v \from C \to \dom f$, we have that
\begin{equation*}
r \comp v = q \comp u
\end{equation*}
since the following diagram commutes.
\begin{equation*}
\begin{tikzcd}
	D & {\dom f} \\
	C & B \\
	& A
	\arrow["g", from=2-1, to=2-2]
	\arrow["q"', from=2-1, to=3-2]
	\arrow["p"', from=2-2, to=3-2]
	\arrow["f"', from=1-2, to=2-2]
	\arrow["r", bend left=35, from=1-2, to=3-2]
	\arrow["v", from=1-1, to=1-2]
	\arrow["u"', from=1-1, to=2-1]
\end{tikzcd} 
\end{equation*}
Therefore, the following diagram is well-typed and commutes in $\slice{\CC}{A}$.
\begin{equation*}
\begin{tikzcd}
qu = rv \rar{u} \dar[swap]{v}
      & q \dar{g}
      \\
      r \rar{f}
      & p
\end{tikzcd}
\end{equation*}
Since this holds for all $u$, we have $\sliceDom[A](\pbCover{g}F) = \pbCover{(\sliceDom[A]g)}F'$ and thus $\pbCover{g}F \in \SliceCoverage{A}(q)$.

The maximal sieve $\Yo_{p}$ on $p \from B \to A$ is given by $\pbCover{p}\Yo_{A}$, which is in $\SliceCoverage{A}(p)$ since $\Yo_{A} \in J(A)$.

Next, we prove transitivity.
Let $F \in \SliceCoverage{A}(p)$ and $G \subseteq \Yo_{p}$, such that $\pbCover{h}G \in \SliceCoverage{A}(q)$ for all $h \from q \to p$.
We want to show that $G \in \SliceCoverage{A}(p)$,~i.e.~$\sliceDom[A]G \in J(A)$.
For all $h \from C \to B$, we have $h \from ph \to p$ in $\slice{\CC}{A}$.
Thus, by assumption, $\pbCover{h}G \in \SliceCoverage{A}(ph)$ and by definition $\sliceDom{A}(\pbCover{h}G) \in J(C)$.
It follows from transitivity of $J$ that~$\sliceDom[A]G \in J(A)$ and thus $G \in \SliceCoverage{A}(p)$ by definition.

  Finally, if $F(q) \subseteq \slice{\CC}{A}(q, p)$ is merely a subset of morphism with $\sliceDom[A]F \in J(A)$, then $F$ is actually a sieve because $\sliceDom[A]F$ is and because $\sliceDom[A]$ is a functor:
  Let $f \in F(q)$ and $g \from r \to q$.
  Then $\sliceDom[A](f \comp g) = \sliceDom[A]f \comp \sliceDom[A]g \in \sliceDom[A]F$ and thus $F \in \SliceCoverage{A}(p)$.
\end{proofE}

Let us briefly establish some structure of slice categories since those will feature in our construction of our universe sheaf. 
The \emph{codomain functor}~$\cod \from \ArrC{\CC} \to \CC$ on the arrow category~$\ArrC{\CC}$ is a strict opfibration with fibres $\slice{\CC}{A}$ presented by the slice functor $\slice{\CC}{-} \from \CC \to \CatC_{0}$~\cite{Jacobs99}.
The codomain functor has a right adjoint~$T \from \CC \to \ArrC{\CC}$ defined on objects by $TA = \id_{A}$ which fulfils $\cod \comp T = \id_{\CC}$.
Consequently,~$TA$ is terminal in $\slice{\CC}{A}$ so we obtain a unique morphism $!_{p} \from p \to TA$.

\begin{lemmaE}[][category=internal-categories-sheaves]
\label{lem:geometric-slice-functor}
Let $\Site$ be a site and $f \from A \to B$ a morphism in $\CC$.
The functor $\slice{\CC}{f} \from \slice{\CC}{A} \to \slice{\CC}{B}$ induces a geometric morphism $\overline{\slice{\CC}{f}} \from \Sh[k]{\SliceSite*{A}} \to \Sh[k]{\SliceSite*{B}}$, where the left adjoint $\bparens{\overline{\slice{\CC}{f}}}^{\ast}$ is given by $\bparens{\overline{\slice{\CC}{f}}}^{\ast}(S) = S \comp \op{(\slice{\CC}{f})}$.
\end{lemmaE}
\begin{proofE}
Let $\Site$ be a site and $f \from A \to B$ a morphism in $\CC$.
Given a sheaf $S \in \Sh[k]{\SliceSite*{B}}$, we first show that $S \comp \op{(\slice{\CC}{f})}$ is a sheaf.
Suppose that we are given $p \in \slice{\CC}{A}$, $G \in \SliceCoverage{A}(p)$ and a compatible family $\set{x_{g} \in S((\slice{\CC}{f})(\dom g))}_{g \in G}$.
By definition, we have $\slice{\CC}{f}(p) = f \comp p$, thus also $G \in \SliceCoverage{B}(f \comp p)$ and $x_{g} \in S(f \comp (\dom g))$.
Compatibility then ensures that there is a unique amalgamation $x \in S(f \comp p) = S((\slice{\CC}{f})p)$ with $S(g)(x) = x_{g}$.
Therefore, $S \comp \op{(\slice{\CC}{f})}$ is a sheaf.

Next, we show that $\slice{\CC}{f}$ has the covering lifting property.
This means that we need to turn a cover $G \in \SliceCoverage{B}((\slice{\CC}{f})p)$ for $p \from C \to A$ into a cover $G' \in \SliceCoverage{A}(p)$ with $\img{(\slice{\CC}{f})}G' = G$.
Every $g \in G$ fits into the following diagram, which commutes since $g \from {\dom g} \to (\slice{\CC}{f})p$ and by definition of $\slice{\CC}{f}$.
\begin{equation*}
\begin{tikzcd}
      D \ar[rr, "g"] \ar[ddr, bend right, "\dom g"{swap}] \drar[swap]{pg}
      & & C \ar[dl, "p"] \ar[ddl, bend left, "(\slice{\CC}{f})p"] \\
      & A \dar{f} & \\
      & B &
\end{tikzcd}
\end{equation*}
Since $pg$ is in $\slice{\CC}{A}$, we may define
\begin{equation*}
    G' = \setDef{g \from pg \to p}{g \from \dom g \to (\slice{\CC}{f})p \in G} \, ,
\end{equation*}
which fulfils $\img{(\slice{\CC}{f})}G' = G$.
By definition of slice coverages, we get $G' \in \SliceCoverage{A}(p)$ and thus $\slice{\CC}{f}$ has the covering lifting property. Thus, we may apply \cite[Thm. VII.10.5]{MM94:SheavesGeometryLogic} to obtain that $\bparens{\overline{\slice{\CC}{f}}}^{\ast}$ is part of a geometric morphism, as desired.
\end{proofE}
With these preliminary results set up, we can prove that $\Univ$ is a sheaf.
\begin{theoremE}[][category=internal-categories-sheaves]
\label{thm:univ-sheaf}
The mapping $\Univ$, as defined in \cref{eq:univ-object-def}, extends to a sheaf $\Univ \in \Sh[1]{\CC, J}$.
\end{theoremE}
\begin{proofE}
To make $\Univ$ a presheaf $\op{\CC} \to \Set_{1}$ we define  
\begin{equation*}
\Univ(f)(S) = S \comp \op{(\slice{\CC}{f})} 
\end{equation*}
for all~$f \from A \to B$ in $\CC$.
By \cref{lem:geometric-slice-functor}, we have that $\Univ(f) \from \Univ(B) \to \Univ(A)$.
Moreover, $\Univ$ is a functor because $\slice{\CC}{-}$ is.
It remains to be seen that it is also a sheaf.

To this end, let $F \in J(A)$ be a cover of $A$ and $\set{S_{f} \in \Univ(\coverDom[A]{f})}_{f \in F}$ be a compatible family of sheaves in $\Sh[0]{\SliceSite*{\coverDom[A]{f}}}$.
Compatibility amounts to saying that for all $g \from B \to \coverDom[A]{f_{1}}$ and $h \from B \to \coverDom[A]{f_{2}}$ with $f_{1}g = f_{2}h$ it holds that 
\begin{equation*}	
S_{f_{1}} \comp \op{\slice{\CC}{g}} = S_{f_{2}} \comp \op{\slice{\CC}{h}}.
\end{equation*}
In what follows, we show that there is a sheaf $S \in \Univ(A) = \Sh[0]{\SliceSite*{A}}$ with $S \comp \op{\slice{\CC}{f}} = S_{f}$.

We begin by defining the category $\DC$ of factorisations of the cover as full subcategory of the arrow category $\ArrC{\CC}$, where the set of objects of $\CC$ is just $F$ and the morphisms are given as follows.
\begin{equation*}
\DC(f_{1}, f_{2}) = \setDef{g \in F}{(g, \id_{A}) \in \ArrC{\CC}(f_{1}, f_{2})}
\end{equation*}
Since $\CC$ is small, we can form the diagram $I \from \DC \to \CatC_{0}$ by taking $I(f) = \slice{\CC}{\coverDom{f}}$ and $Ig = \slice{\CC}{g}$, that is, $I$ is the restriction of the slice functor $\slice{\CC}{-} \from \CC \to \CatC_{0}$.
Forming the colimit of $I$ in $\CatC_{0}$, we obtain a small category $\KC = \colim I$.
  As the opposite category functor $\opF \from \CatC_{0} \to \CatC_{0}$ is an involution and thus its strictly left-adjoint to itself, we obtain that $\op{\KC} = \colim (\opF \comp I)$.

For $g \in \DC(f_{1}, f_{2})$, we have
\begin{equation*}
    \slice{\CC}{f_{2}} \comp \slice{\CC}{g} = \slice{\CC}{f_{2}g} = \slice{\CC}{f_{1}}
\end{equation*}
and thus there exists a unique functor $\tilde{F}$ that renders the following diagram commutative for all $f$.
\begin{equation*}
    \begin{tikzcd}
      \KC \rar[dashed]{\tilde{F}}
      & \slice{\CC}{A} \\
      \slice{\CC}{\coverDom[A]{f}} \uar{\kappa_{f}} \urar[swap]{\slice{\CC}{f}}
    \end{tikzcd}
\end{equation*}

By compatibility, we have for all $g \in \DC(f_{1}, f_{2})$ that $S_{f_{2}} \comp \op{(\slice{\CC}{g})} = S_{f_{1}} \comp \op{(\slice{\CC}{\id})} = S_{f_{1}}$.
Hence there is a unique functor $\tilde{S} \from \op{\KC} \to \Set[0]$, such that the following diagram commutes for all $f$.
\begin{equation*}
    \begin{tikzcd}
      \colim(\opF \comp D) = \op{\KC} \rar[dashed]{\tilde{S}}
      & \Set[0] \\
      \op{(\slice{\CC}{\coverDom[A]{f}})} \uar{\op{\kappa_{f}}} \urar[swap]{S_{f}}
    \end{tikzcd}
\end{equation*}

Since $\Set[0]$ is complete and $\slice{\CC}{A}$ is small, there is a right Kan extension of $\tilde{S}$ along $\op{\tilde{F}}$ as in the following diagram.
\begin{equation*}
    \begin{tikzcd}
      \op{\KC}
      \rar{\op{\tilde{F}}}
      \ar[d, "\tilde{S}"{swap}, ""{name=S, pos=0.4, anchor=center, inner sep=0}]
      & \slice{\CC}{A}
      \ar[dl, "R", ""{name=R, pos=0.35, anchor=center, inner sep=0}]
      \\
      \Set[0]
      \arrow["\rho"{swap}, shorten <=5pt, shorten >=3pt, Rightarrow, from=R, to=S]
    \end{tikzcd}
\end{equation*}
Whiskering with~$\kappa_{f} \from \slice{\CC}{\coverDom[A]{f}} \to \KC$ yields a natural transformation $\rho^{f}$ as follows.
\begin{equation*}
\rho^{f} = \rho \op{\kappa_{f}} \from R \comp \slice{\CC}{f} 
	    = R \comp \op{\tilde{F}} \comp \op{\kappa_{f}} \longrightarrow \tilde{S} \comp \op{\kappa_{f}} 
	    = S_{f}
\end{equation*}

The functor $R$ is a good candidate for the amalgamation of the sheaves $S_{f}$.
Unfortunately we only have the projections $\rho^{f}$, but $S$ does not fulfill the strict equality 
$R \comp \Univ{f} = R \comp \slice{\CC}{f} = S_{f}$.
However, this can be rectified.
We define a functor $S \from \op{\slice{\CC}{A}} \to \Set[0]$ as follows.
On objects we set
\begin{equation*}
S(p) =
    \begin{cases}
      S_{f}(p'), & p = f p' \\
      R(p), & \text{otherwise}
    \end{cases}
    \, ,
\end{equation*}
which is well-defined because if $p = f_{2}p'' = f_{1}p'$, then $S_{f_{1}}(p') = S_{f_{2}}(p'')$ by compatibility.
We define $S$ on a morphism $g \from q \to p$ by
\begin{equation*}
S(g) =
    \begin{cases}
      S_{f}(g), & p = fp' \\
      \rho^{f}_{q'} \comp R(g), & q = fq' \text{ and } p \text{ cannot be factorised} \\
      R(g), & \text{otherwise}
    \end{cases}
\end{equation*}
Note that the case in which both $p$ and $q$ factorise is subsumed by the first clause:
$p = fp'$ implies that $q = pg = fp'g$.
In the second case we use that
\begin{equation*}
    S(p) = R(p) \xrightarrow{R(g)} R(q) = R(fq') = (R \comp \slice{\CC}{f})(q') \xrightarrow{\rho^{f}_{q'}} S_{f}(q') = S(q)
\end{equation*}
Functoriality of $S$ is clear, except in the second case, where one uses naturality of $\rho$.\todo{Give details}
We also obtain a natural transformation $\pi \from S \comp \op{F} \to \tilde{S}$, which makes $S$ a right Kan extension.\todo{Give details}
By definition, we have $\Univ(f)(S) = S \comp \op{\slice{\CC}{f}} = S_{f}$ and, by the universal property of Kan extension, $S$ is unique with this property.

It remains to prove that $S$ is a sheaf in $\Sh[0]{\SliceSite*{A}}$.
Given a cover $G = \set{g \from p_{g} \to p} \in \SliceCoverage{A}(p)$ of $p \in \slice{\CC}{A}$ and a compatible family $x_{g} \in S(p_{g})$, we need to construct $x \in S(p)$ with $S(g)(x) = x_{g}$.
By the Yoneda lemma, $S(p) \cong \PSh[0]{\CC}(\YoP{p}, S)$ and we can instead construct a natural transformation $\alpha \from \YoP{p} \to S$.
Since $S$ is a right Kan extension, such a natural transformation is uniquely given by a natural transformation $\beta \from \YoP{p} \comp \op{\tilde{F}} \to \tilde{S}$.
Since $KC$ is a colimit, such a $\beta$ then arises uniquely from a family of natural transformations $\beta^{f} \from \YoP{p} \comp \op{\slice{\CC}{f}} \to S_{f}$ with $\beta^{f} \slice{\CC}{g} = \beta^{fg}$ for all $g$ in $\DC$.
Finally, for $q \in \slice{\CC}{A_{f}}$ and $h \from f \comp q \to p$ in $\slice{\CC}{A}$, we construct $\beta^{f}_{q}(h) \in S_{f}(q)$ using that $S_{f}$ is a sheaf.
To this end, we define for $u \in \pbCover{h}G$ in the pullback cover with $hu = gv$ or some $g \in G$ and $v \from \dom u \to p_{g}$ elements $y^{q,h}_{u} \in S_{f}(qu)$ by $y^{q,h}_{u} = (\pi \kappa_{f})_{qu}(S(v)(x_{g})$.
One can show that $y^{q,h}_{u}$ is independent of the choice of $v$ by using compatibility of $x_{f}$ and that these are compatible elements for $\pbCover{h}G$.\todo{Give details}
This yields uniquely $\beta^{f}_{q}(h)$ as their amalgamation and thus makes $\beta^{f}_{q}$ immediately a map.
Naturality of $\beta^{f}$ follows from uniqueness of amalgamations and thus we obtain $\alpha \from \YoP{p} \to S$ as discussed above.\todo{Give details}
By the Yoneda lemma, our candidate for the amalgamation of $x_{f}$ is given by $x = \alpha_{p}(\id_{p})$.
That this yields indeed a unique amalgamation follows by applying the Yoneda lemma to each $x_{f}$ and the universal property $S$ as Kan extension.\todo{Give details}

With this, $S$ is a sheaf and we have shown that $\Univ$ is a well-defined sheaf.
\end{proofE}

We can then internalise small sheaves into this universe.
\begin{theoremE}[][category=internal-categories-sheaves]
\label{thm:sheaf-internalisation}
There is an isomorphism $i \from \Sh[0]{\CC, J} \to \Sh[1]{\CC, J}(\TSheaf, \Univ)$, given as follows.
\begin{equation*}
  i(S)_{A}(\ast) = S \comp \op{\bparens{\sliceDom[A]}}
\end{equation*}
\end{theoremE}
\begin{proofE}
We define maps $i \from \Sh[0]{\CC, J} \to \Sh[1]{\CC, J}(\TSheaf, \Univ)$ and its inverse $j \from \Sh[1]{\CC, J}(\TSheaf, \Univ) \to \Sh[0]{\CC, J}$ by
\begin{equation*}
  i(S)_{A}(\ast) = S \comp \op{(\sliceDom[A])}
    \quad \text{and} \quad
    \begin{aligned}
      j(\gamma) A & = \gamma_{A}(\ast)(\id_{A}) \\
      j(\gamma)(f \from A \to B) & = \gamma_{B}(\ast)(!_{f})
    \end{aligned} \, ,
\end{equation*}
where $!_{f} \from f \to \id_{B}$ is $f$ regarded as the unique morphism $f \to \id_{B}$ in $\slice{\CC}{B}$ with $\id_{B} \comp f = f$.
That $j(\gamma)$ is well-defined on morphisms follows because, when viewing $f$ as object in $\slice{\CC}{B}$, we have
  \begin{align*}
    \gamma_{B}(\ast)(f)
    & = (\gamma_{B}(\ast) \comp \op{(\slice{\CC}{f})})(\id_{A})
      \tag*{by def. $\slice{\CC}{f}$} \\
    & = (\Univ{f} (\gamma_{B}(\ast))(\id_{A})
      \tag*{by def. $\Univ$} \\
    & = \gamma_{A}(\ast)(\id_{A})
      \tag*{by naturality of $\gamma$} \\
    & = j(\gamma) A
      \tag*{by def.}
  \end{align*}
and thus $j(\gamma)(f) \from j(\gamma)(B) \to j(\gamma)(A)$.
It remains to show that $j(\gamma)$ is a sheaf.

  Let $\set{f \from A_{f} \to A} \in J(A)$ be a cover and $x_{f} \in j(\gamma)(A_{f})$ a compatible family.
  Note that the cover can be seen as a cover $!_{f} \from f \to \id_{A}$ of $\id_{A} \in \slice{\CC}{A}$.
  We will use that $\gamma_{A}(\ast) \from \slice{\CC}{A} \to \Set[0]$ is a sheaf as follows.
  First, we note that $f \in \slice{\CC}{A}$ and thus $x_{f} \in \gamma_{A}(\ast)(f)$.
  In order to check compatibility, suppose that $h \from p \to f_{1}$ and $g \from p \to f_{2}$ are morphisms in $\slice{\CC}{A}$ with $f_{1}h = f_{2}g$.
  This means in that for $!_{g} \from g \to \id_{A}$, we have $\slice{\CC}{f_{2}}(!_{g}) = g \from p \to f_{2}$.
  Similarly, $\slice{\CC}{f_{1}}(!_{h}) = h \from p \to f_{1}$.
  \begin{align*}
    \gamma_{A}(\ast)(h)(x_{f_{1}})
    & = \gamma_{A}(\ast)(\slice{\CC}{f_{1}}(!_{h}))(x_{f_{1}}) \tag*{see above} \\
    & = (\Univ(f_{1}) \comp \gamma_{A})(\ast)(!_{h})(x_{f_{1}}) \tag*{def. $\Univ$} \\
    & = \gamma_{A_{f_{1}}}(\ast)(!_{h})(x_{f_{1}}) \tag*{by naturality} \\
    & = j(\gamma)(h)(x_{f_{1}}) \tag*{def. j} \\
    & = j(\gamma)(g)(x_{f_{2}}) \tag*{compatibility wrt. $j(\gamma)$} \\
    & = \gamma_{A}(\ast)(g)(x_{f_{2}}) \tag*{symmetric steps for $g$}
  \end{align*}
  Hence, the family $x_{f}$ is compatible for $\gamma_{A}(\ast)$ and we get a unique $x \in \gamma_{A}(ast)$ with $\gamma_{A}(\ast)(!_{f})(x) = x_{f}$.
  Since $j(\gamma)(f) = \gamma_{A}(\ast)(!_{f})$, this is equivalent to $x$ being unique with $j(\gamma)(f)(x) = x_{f}$ and $x$ is the unique amalgamation for $j(\gamma)(A)$.
  Thus, $j(\gamma)$ is a sheaf.

We next check that also $i$ is well-defined.
Note that $i(S)_{A}(\ast) = \op{(\slice{\CC}{A})} \xrightarrow{\op{(\sliceDom[A])}} \op{\CC} \xrightarrow{S} \Set[0]$.
Given a cover $G \in \SliceCoverage{A}(p)$, we get a cover $\sliceDom[A]G \in J(a)$ and thus $i(S)_{A}(\ast)$ is a sheaf because $S$ is.
Hence, $i(S)_{A}$ is well-defined.

Moreover, $i(S)$ is natural, since we have for all $f \from A \to B$ and by unfolding the definitions that
  \begin{equation*}
    (\Univ(f) \comp i(S)_{B})(\ast)
    = S \comp \op{(\sliceDom[B])} \comp \op{(\slice{\CC}{f})} = S \comp \op{(\sliceDom[A])} = i(S)_{A}(\ast) \, .
  \end{equation*}

  It remains to prove that $i$ and $j$ are inverses.
  In one direction, we have for all $A \in \CC$ and $f \from A \to B$ that
  \begin{equation*}
    j(i(S)) A
    = i(S)_{A}(\ast)(\id_{A})
    = (S \comp \op{(\sliceDom[A])})(\id_{A})
    = S(A)
  \end{equation*}
  and
  \begin{equation*}
    j(i(S)) f
    = i(S)_{B}(\ast)(!_{f})
    = (S \comp \op{(\sliceDom[B])})(!_{f})
    = S(f) \, ,
  \end{equation*}
  and thus $j \comp i = \id$.

  Finally, we show $i \comp j = \id$.
  On objects $p \from B \to A$ in $\slice{\CC}{A}$, we have that
  \begin{align*}
    i(j(\gamma))_{A}(\ast)(p)
    & = j(\gamma)(\op{(\sliceDom[A])}(p)) \tag*{def. $i$} \\
    & = j(\gamma)(B) \tag*{def. $\sliceDom[A]$} \\
    & = \gamma_{B}(\ast)(\id_{B}) \tag*{def $j$} \\
    & = (\Univ(p) \comp \gamma_{A})(\ast)(\id_{B}) \tag*{naturality $\gamma$} \\
    & = (\gamma_{A}(\ast) \comp \op{\slice{\CC}{p}})(\id_{B}) \tag*{def. $\Univ$} \\
    & = \gamma_{A}(\ast)(p) \tag*{def. slice functor}
  \end{align*}
  and on morphisms $f \from p \to q$ we have
  \begin{align*}
    i(j(\gamma))_{A}(\ast)(f)
    & = j(\gamma)(\op{(\sliceDom[A])}(f)) \tag*{def. $i$} \\
    & = j(\gamma)(f) \tag*{def. $\sliceDom[A]$} \\
    & = \gamma_{\dom q}(\ast)(!_{f}) \tag*{def. $j$} \\
    & = (\Univ(q) \comp \gamma_{A})(\ast)(!_{f}) \tag*{naturality $\gamma$} \\
    & = \gamma_{A}(\ast)(\slice{\CC}{q}(!_{f})) \tag*{def. $\Univ$} \\
    & = \gamma_{A}(\ast)(f) \tag*{def. slice functor}
  \end{align*}
Thus, $i \comp j = \id$ and we obtain the desired isomorphism.
\end{proofE}

Next, we complete $\Univ_0$ to an internal category.
For this, we first construct a sheaf $\Univ_1$, so that morphisms of small sheaves can be presented as maps ~$\TSheaf\to \Univ_1$.
The object part of ~$\Univ_1$ is given by
\begin{equation}
  \label{eq:morphism-universe}
  \Univ_1(A) = \coprod_{F, G \in \Univ_0} G^F
\end{equation}


\begin{theoremE}[][category=internal-categories-sheaves]
\label{thm:morphism-universe}
The mapping $\Univ_1$, as defined in~\cref{eq:morphism-universe}, extends to a sheaf ~$\Univ_1\in \Int{\Sh{\CC, J}}$.
\end{theoremE}
\begin{proofE}
Given a morphism ~$g\colon B\to A$ in ~$\CC$, let $\bparens{\overline{\slice{\CC}{g}}}^{\ast}$ denote the left-adjoint of the geometric morphism $\overline{\slice{\CC}{g}} \from \Sh[0]{\SliceSite*{B}} \to \Sh[0]{\SliceSite*{A}}$.

Now, for all ~$(F \in \Univ_0(A), G\in \Univ_0(A), \alpha \in \mor{\Sh[0]{\SliceSite*{A}}}(F, G))$, we define
\begin{equation*}
  \Univ_1(g)(F, G, \alpha) := (\Univ_0(g)(F), \Univ_0(g)(G), \bparens{\overline{\slice{\CC}{g}}}^{\ast}(\alpha))
\end{equation*}

Functoriality holds by functoriality of ~$\bparens{\overline{\slice{\CC}{g}}}^{\ast}$. 
Similarly, ~$\Univ_1(g)$ is well defined, as, again by functoriality of $\bparens{\overline{\slice{\CC}{g}}}^{\ast}$, $\bparens{\overline{\slice{\CC}{g}}}^{\ast}\alpha \in\mor{\Sh[0]{\SliceSite*{B}}}(F, G))$.

It remains to show that $\Univ_1$ satisfies the sheaf conditions for all ~$A\in \CC$ and all covers of~$A$.
Let~$S = \{f\colon \sliceDom{f}\to A\} \in J(A)$ be a cover of ~$A$ and ~$M = \{(F_{\sliceDom{f}}, G_{\sliceDom{f}}, \alpha_{\sliceDom{f}}) \in \Univ_1(\sliceDom{f})\}$ be a matching family for~$S$.

From this, we extract matching families and corresponding unique amalgamations (by virtue of ~$\Univ_0$ being a sheaf)
\begin{itemize}
  \item $\{ F_{\sliceDom{f}} \in \Univ_0(\sliceDom{f})\}$, $F \in \Univ_0(A)$
  \item $\{ G_{\sliceDom{f}} \in \Univ_0(\sliceDom{f})\}$, $G \in \Univ_0(A)$
\end{itemize}
and a collection of morphisms
\begin{itemize}
  \item $\{\alpha_{\sliceDom{f}}\colon F_{\sliceDom{f}}\to G_{\sliceDom{f}}\}$
\end{itemize}

Note that for any ~$F, G\in \Univ_0(A)$, the internal hom sheaf ~$[F, G] \in \Sh[0]{\SliceSite*{A}}$ is defined as
\begin{equation*}
  [F, G](g\colon B\to A) := \mor{\Sh[0]{\SliceSite*{A}}}(F\circ \op{\slice{\CC}{g}}, G\circ \op{\slice{\CC}{g}})
\end{equation*}

Since, for amalgamations~$F, G$, we hav that
\begin{align*}
  F = \Univ_0(f)(F_{\sliceDom{f}}) = F\circ \op{\slice{\CC}{f}} = F_{\sliceDom{f}}\\
  G = \Univ_0(f)(G_{\sliceDom{f}}) = G\circ \op{\slice{\CC}{f}} = G_{\sliceDom{f}}
\end{align*}
it follows that
\begin{equation*}
  \alpha_{\sliceDom{f}}\colon F_{\sliceDom{f}}\to G_{\sliceDom{f}} \in \mor{\Sh[0]{\SliceSite*{\sliceDom{f}}}}(F\circ \op{\slice{\CC}{f}}, G\circ \op{\slice{\CC}{f}}).
\end{equation*}
That is,~$ \alpha_{\sliceDom{f}} \in [F, G](f)$.

Moreover, the cover ~$S \in J(A)$ is also a cover for $\id_A$, and the compatibility condition of $M$ ensuring that
\begin{equation*}
  \bparens{\overline{\slice{\CC}{f_i}}}^{\ast}(\alpha_{\sliceDom{f_i}}) = \bparens{\overline{\slice{\CC}{f_j}}}^{\ast}(\alpha_{\sliceDom{f_j}})
\end{equation*}
for ~$f_i, f_j$ in $S$, shows that ~$\{\alpha_{\sliceDom{f}}\}$ is a matching family for $S$.

Therefore, because~$[F, G]$ is a sheaf, there exists a unique amalgamation
\begin{equation*}
  \alpha\colon F\to G \in [F, G](\id_A)
\end{equation*}
restricting to all~$\alpha_{\sliceDom{f}}$ in $\{\alpha_{\sliceDom{f}}\}$.

We let~$(F, G, \alpha)$ be the unique amalgamation for $M$, showing that~$\Univ_1$ satisfies the sheaf condition for all $A$ in~$\CC$ and all covers of~$A$, and thereby concluding the proof.
\end{proofE}

Having defined the sheaves $\Univ_0$ and $\Univ_1$, representing small sheaves and their morphisms, we organise this data into an internal category in ~$\Sh{\CC, J}$.

\begin{theoremE}[][category=internal-categories-sheaves]
  Given the sheaves ~$\Univ_0$ and ~$\Univ_1$ as defined in~\cref{thm:univ-sheaf,thm:morphism-universe}, together with the maps ~$s^U, t^U, e^U, c^U$ defined for all $A$ in $\CC$ by lifting the structure maps of ~$\Sh[0]{\SliceSite*{A}}$ as follows:
\begin{itemize}
\item 
\textbf{source and target:} For each ~$(F, G, \alpha) \in \Univ_1$,
 \begin{equation*}
 s^U_A(F, G, \alpha) = F \quad\text{and}\quad t^U_A(F, G, \alpha) = G.
 \end{equation*}
 \item 
 \textbf{identity assignment: } For each~$F \in \Univ_0(A)$,
$e^U_A(F) = (F, F, \id_F).$
 \item 
 \textbf{composition:} For each~$P = \left( (F, G, \alpha), (G, H, \beta)\right) \in \pb{\Univ_1}{\Univ_1}{\Univ_0}$,
 \begin{equation*}
      e^U_A(P) = (F, H, \beta\circ \alpha).
\end{equation*}
\end{itemize}
The tuple ~$ \Univ := (\Univ_0, \Univ_1, s^U, t^U, e^U, c^U)$ forms an internal category in~$\Sh{\CC, J}$.
\end{theoremE}
\begin{proofE}
Because we use the structure of the ambient category, all maps are well-defined. 
It remains to check that all maps are natural and that the axioms of an internal category are satisfied.
\begin{enumerate}
 \item 
 \textbf{Naturality:} Let~$f\colon B\to A$ be a morphism.
 
     \textbf{Source and target:} 
     We first show that~$s^U_B\circ \Univ_1(f) = \Univ_0(f)\circ s^U_A$.
      Given~$(F, G, \alpha)\in \Univ_1(A)$, we compute as follows:
      \begin{align*}
       s^U_B(\Univ_1(f)(F, G, \alpha)) = S^U_B(\Univ_0(f)(F), \Univ_0(f)(G), \bparens{\overline{\slice{\CC}{f}}}^{\ast}(\alpha)) 						       = \Univ_0(f)(F) = s^U_A(\Univ_0(f)(F, G, \alpha)).
      \end{align*}
      The corresponding identity~$t^U_B\circ \Univ_1(f) = \Univ_0(f)\circ t^U_A$ is obtained similarly:
      \begin{align*}
        t^U_B(\Univ_1(f)(F, G, \alpha)) = t^U_B(\Univ_0(f)(F), \Univ_0(f)(G), \bparens{\overline{\slice{\CC}{f}}}^{\ast}(\alpha)) 							       = \Univ_0(f)(G) = t^U_A(\Univ_0(f)(F, G, \alpha)).
      \end{align*}
      
      \textbf{identity assignment:} We next verify that~$e^U_B\circ \Univ_0(f) = \Univ_1(f)\circ e^U_A$.
      Given~$F \in \Univ_0(A)$, we have
      \begin{align*}
       e^U_B(\Univ_0(f)(F)) = (\Univ_0(f)(F), \Univ_0(f)(F), \bparens{\overline{\slice{\CC}{f}}}^{\ast}(\id_F)) 
       					= \Univ_1(f)(F, F, \id_F) 
					= \Univ_1(f)(e^U_A(F))
      \end{align*}
      
      \textbf{composition:} 
      Finally, we show that~$c^U_B\circ \pb{\Univ_1(f)}{\Univ_1(f)}{\Univ_0(f)} = \Univ_1(f)\circ c^U_A$.
      We compute for~$ P := \left( (F, G, \alpha), (G, H, \beta)\right) \in \pb{\Univ_1{A}}{\Univ_1(A)}{\Univ_0(A)}$
      as follows:
      \begin{align*}
        c^U_B(\pb{\Univ_1(f)}{\Univ_1(f)}{\Univ_0(f)}(P))\\
        = c^U_B(\left( (\Univ_0(f)(F), \Univ_0(f)(G), \bparens{\overline{\slice{\CC}{f}}}^{\ast}(\alpha)), (\Univ_0(f)(G), 			   \Univ_0(f)(H), \bparens{\overline{\slice{\CC}{f}}}^{\ast}(\beta))\right)) \\
        = (\Univ_0(f)(F), \Univ_0(f)(H), \bparens{\overline{\slice{\CC}{f}}}^{\ast}(\beta\comp\alpha))\\
        = \Univ_1(f)(F, H, \beta\comp\alpha) = \Univ_1(f)(c^U_A(P)).
    \end{align*}
\item 
\textbf{Internal category axioms:}
All internal category axioms are satisfied point-wise because
 \begin{itemize}
  \item 
  composition and identity in ~$\Sh[0]{\SliceSite*{A}}$ are associative and unital;
  \item 
  $\bparens{\overline{\slice{\CC}{f}}}^{\ast}$ is functorial and therefore preserves identities and composition.
  \end{itemize}
Thus, the relevant diagrams commute for all~$A\in\CC$ hence also globally.
\end{enumerate}
\end{proofE}

We show that the internal category $\Univ$ indeed captures small sheaves and their morphisms.



\begin{theoremE}[][category=internal-categories-sheaves]
\label{thm:internalising-small-sheaves}
A small sheaf~$F\colon \op{\CC}\to \Set_0$ can be internalised as global section $\sheafInt{F}\colon \TSheaf\to \Univ_0$ in $\Sh{\CC, J}$ defined by $\sheafInt{F} = i(F)$ using the isomorphism from \cref{thm:sheaf-internalisation}.
A natural transformations~$\alpha\colon F\to G$ of small sheaves can be internalised in the universe by defining the internal natural transformation~$\sheafInt{\alpha}\colon \TSheaf\to \Univ_1$ in ~$\Int{\Sh{\CC, J}}$
defined for all~$A\in \CC$ as
\begin{equation*}
\sheafInt{\alpha}_A(\ast) \coloneqq (\sheafInt{F}_{A}(\ast), \sheafInt{G}_{A}(\ast), \alpha \sliceDom[A] \from F \comp \op{\bparens{\sliceDom[A]}} \to G\circ \op{\bparens{\sliceDom[A]}})
\end{equation*}
\end{theoremE}
\begin{proofE}
The first part holds by~\cref{thm:sheaf-internalisation}. 
The map ~$\sheafInt{\alpha}$ is well defined: for all ~$A\in \CC$, elements in ~$\Univ_1(A)$ are tuples ~$(F^A\in \Univ_0(A), G^A\in \Univ_0(A), \alpha^A\in [F^A, G^A])$.
We transport the tuple ~$(F, G, \alpha)$, where $\alpha$ is the morphism of small sheaves to be transported, along the isomorphism
\begin{equation*}
i\colon\Sh[0]{\CC, J}\xra{\cong}\mor{\Sh[1]{\CC, J}}(\TSheaf, \Univ_0)
\end{equation*}
to show that ~$(F\circ \op{(\sliceDom[A])}), (G\circ \op{(\sliceDom[A])}) \in \Univ_0(A)$ and ~$\alpha^A \in [(F\circ \op{(\sliceDom[A])}), (G\circ \op{(\sliceDom[A])})]$, and therefore ~$\sheafInt{\alpha}_A(*)\in \Univ_1(A)$.
  
It remains to show that~$\sheafInt{\alpha}$ is a valid internal natural transformation. 
This follows from naturality of~$\alpha^A$ and the fact that the structure maps of~$\Univ_0$ and~$\Univ_1$, viewed as an internal category, are defined using the structure of the ambient category ~$\Sh[0]{\SliceSite*{A}}$.
\end{proofE}

Having shown how small sets arise as global sections of the universe, we now extend this perspective to recover small sheaves as internally indexed families varying over objects in ~$\CC$.

\subsection{Predicates over the Universe}
\label{sec:predicates}
Our next goal is to describe a concept of \emph{predicate} internal to~$\Sh{\CC, J}$ which forms the basis of what we have called \emph{sheafeology}.
Externally, predicates over a set are modeled typically captured by subsets.
We abstract from this view by exploiting the topos structure of sheaf categories: we construe predicates of sheaves as \emph{subsheaves}.
This will afford the construction of a pair of sheaves~($\Pred_0, \Pred_1)$ which carry the structure of a category~$\Pred$ of \emph{predicates} in~$\Int{\Sh{\CC, J}}$.

We first describe a sheaf~$\Pred_0$ which collects all internal predicates over small sheaves construed as \emph{subsheaves}:~natural transformations between sheaves with each components a monomorphism.
Importantly, as is the case for any topos~\cite{Johnstone02}, 
$\Univ_0(A) = \Sh[0]{\SliceSite*{A}}$
has a \emph{subobject classifier} for every~$A\in\CC$.
That is, there is a sheaf~$\Omega$ with the property that for any sheaf~$F$, the set of (equivalence classes of) presheaves of $F$ correspond bijectively with natural transformations from~$F$ to~$\Omega$.
That is, any subsheaf~$G\hra F$ can be identified with its \emph{characteristic morphism}~$\chi_G\colon F\to\Omega$. 
We employ this in defining~$\Pred_0(A)$ to range over all characteristic morphisms:
\begin{equation}
\label{eq:predicate-sheaf}
\Pred_0(A) = \coprod_{F \in \Univ_0(A)}\Sh[0]{\SliceSite*{A}}(F, \Omega)
\end{equation}

\begin{theoremE}[][category=internal-categories-sheaves]
\label{thm:predicate-sheaf}
The assignment~$A\mapsto\Pred_0(A)$ of~\cref{eq:predicate-sheaf} extends to a sheaf ~$\Pred_0\in\Sh{\CC, J}$.
\end{theoremE}
\begin{proofE}
	Given a morphism ~$g\colon B\to A$ in ~$\CC$, let $\bparens{\overline{\slice{\CC}{g}}}^{\ast}$ denote the left-adjoint of the geometric morphism $\overline{\slice{\CC}{g}} \from \Sh[0]{\SliceSite*{B}} \to \Sh[0]{\SliceSite*{A}}$. Now, given the subobject classifier ~$\Omega$ in ~$\Sh[0]{\SliceSite*{A}}$ and the pair ~$(F, \alpha\colon F\to\Omega) \in \Pred_0(A)$, ~$\alpha$ characterises a subobject ~$P\subseteq F$ through the isomorphism ~$\mathrm{Sub}_A(F)\cong [F, \Omega]$~\cite[Prop. I.3.1]{MM94:SheavesGeometryLogic}, where ~$\mathrm{Sub}_A(F)$ is the set of subobjects of ~$F$ in ~$\slice{\CC}{A}$. 
By left exactness of ~$\bparens{\overline{\slice{\CC}{g}}}^{\ast}$, we have that ~$\bparens{\overline{\slice{\CC}{g}}}^{\ast}P$ is a subobject of ~$\bparens{\overline{\slice{\CC}{g}}}^{\ast}F$ in ~$\slice{\CC}{B}$. 
Therefore, we can define~$\Pred_0$ on morphisms ~$g\colon B\to A$ in ~$\CC$ for a pair ~$(F, \alpha)\in\Pred_0(A)$ as
\begin{equation*}
\Pred_0(f)(F, \alpha) = 
(\bparens{\overline{\slice{\CC}{g}}}^{\ast}(F), \bparens{\overline{\slice{\CC}{g}}}^{\ast}(\alpha)) \cong P\op{\slice{\CC}{g}} \subseteq F\op{\slice{\CC}{g}}
\end{equation*}
	
Functoriality of ~$\Pred_0$ holds by functoriality of ~$\bparens{\overline{\slice{\CC}{g}}}^{\ast}$. 
It remains to show that ~$\Pred_0$ satisfies the sheaf condition for all ~$A\in \CC$ and all covers of ~$A$. 
Let ~$\{f\colon \dom f\to A\}$ be a cover of~$A$ and ~$ M_f\in \Pred_0(A)$ be a compatible family represented by tuples ~$(F_f\in \Univ_0(\dom f), P_f \subseteq F_f \in \Univ_0(\dom f))$. 
Again, compatible means that for all ~$g\colon B\to \dom f_1$ and~$h\colon B\to \dom f_2$ we have that
\begin{equation*}
F_{f_1}\circ \op{\slice{\CC}{g}} = 
F_{f_2}\circ \op{\slice{\CC}{h}} \quad P_{f_1}\circ \op{\slice{\CC}{g}} = P_{f_2}\circ \op{\slice{\CC}{h}}
\end{equation*}
	
We show that there is a sheaf ~$F\in \Univ_0(A)$ with a subobject ~$P\subseteq F$ such that ~$F \circ \op{\slice{\CC}{f}} = F_f$ and ~$P\circ \op{\slice{\CC}{f}}$ for all ~$f\in S$.
	
From the matching family ~$M_S$ we obtain
\begin{itemize}
\item 
a compatible family ~$P_f \in \Univ_0(\dom f)$
\item 
a compatible family ~$F_f \in \Univ_0(\dom f)$
\item 
a family of monomorphisms ~$P_f\subseteq F_f)$
\end{itemize}
	
Because ~$\Univ_0$ is a sheaf, we obtain unique amalgamations ~$P, F\in \Univ_0(A)$ restricting to the pairs in ~$M_S$.
It remains to show that there exists a unique ~$P\subseteq F$. Given that ~$\mathrm{Sub}_A(F) \cong [F, \Omega]$ is a sheaf, we can use the same argument as in the proof of\cref{thm:morphism-universe} to show that the family ~$P_f\subseteq F_F$ induces this monomorphism.
We let ~$(F, P\subseteq F) \in \Pred_0(A)$ be the unique amalgamation for ~$M_S$, showing that ~$\Pred_0$ is a sheaf, concluding the proof.
\end{proofE}

We next construct a sheaf~$\Pred_1$ which will serve as the sheaf of morphisms between internal predicates.
In the following, we use~$[F, G]$ to denote the set of natural transformations~$F\to G$ for~$F,G\in\Univ_0(A)$.
The object part of~$\Pred_1$ is then given by the following equalizer diagram:
\begin{equation*}
\label{eq:morphism-predicate-sheaf}
\begin{tikzcd}
\Pred_1(A) \arrow[r, "\mathrm{eq}"] & {\displaystyle\coprod[F, G]} \arrow[r, yshift=0.7ex, "\pi_1"{name=p_1, above}] \arrow[r, yshift=-0.7ex, "\pi_2"{name=p_2, below}] & \Pred_0(A),
\end{tikzcd}
\end{equation*}
where the coproduct ranges over all pairs of predicates $(F, \alpha),(G, \beta)\in \Pred_0(A)$ and~$\pi_1,\pi_2$ are defined by the assignments
\begin{equation*}
\pi_1\colon((F, \alpha), (G, \beta), \gamma) \mapsto \alpha 
\quad\text{and}\quad
\pi_2\colon ((F, \alpha), (G, \beta), \gamma) \mapsto \beta\circ \gamma.
\end{equation*}
Thus,~$\Pred_1(A)$ consists of triples~$((F, \alpha), (G, \beta), \gamma)$ such that~$\alpha = \beta\comp\gamma$. 
Intuitively, we understand this as expressing that~\emph{$\alpha$ entails~$\beta$ along~$\gamma$}.

\begin{theoremE}[][category=internal-categories-sheaves]
\label{thm:morphism-predicate-sheaf}
The mapping ~$\Pred_1$ as defined in~\cref{eq:morphism-universe} extends to a sheaf ~$\Pred_1 \in \Sh{\CC, J}$.
\end{theoremE}
\begin{proofE}
Let $f\colon B \to A$ in $\CC$. 
Recall that $\Pred_1(A)$ and~$\Pred_1(B)$ are given as equalizers. 
Since $\Univ_0$ and $\Univ_1$ are also defined point-wise as coproducts over sheaves 
and morphisms on $\slice{\CC}{A}$, the maps induced by precomposition with $f$ commute with $\pi_1$ and $\pi_2$. Thus, there are induced morphisms
\begin{equation*}
\coprod [F, G](A) \cong \coprod [F, G](B)
\quad\text{and}\quad
\Pred_0(A) \cong \Pred_0(B)
\end{equation*}
that commute with $\pi_1$ and $\pi_2$. 
By the universal property of equalizers, there is a unique morphism
\begin{equation*}
\Pred_1(f)\colon \Pred_1(A) \to \Pred_1(B)
\end{equation*}
such that $\mathrm{eq}_B \circ \Pred_1(f) = (\coprod [F, G](f)) \circ \mathrm{eq}_A$.
	
Functoriality of~$\Pred_1$ follows by uniqueness of this map. 
It remains to show that for all ~$A\in\CC$ and all covers of ~$A$ the sheaf condition holds.
	
Let ~$\{f\colon \dom f\to A\}$ be a cover of ~$A$ and 
\begin{equation*}
\left((F_f, \alpha_f\colon F_f\to\Omega_f), (G_f, \beta_f\colon G_f\to \Omega_f), \gamma_f\colon F_f\to G_f\right) \in \Pred_1(\dom f)
\end{equation*}
be a compatible family. 
As ~$(F_f, \alpha_f), (G_f, \beta_f)\in \Pred_0(\dom f)$ and ~$\gamma_f\in \Univ_1(\dom f)$, by virtue of ~$\Pred_0$ and ~$\Univ_0$ being sheaves, the compatible family glues to a unique amalgamation
\begin{equation*}
\left( (F,\alpha)\in\Pred_0(A), (G, \beta)\in \Pred_0(A), \gamma\in \Univ_1(A) \right)
\end{equation*}
	
To see that~$\left((F, \alpha), (G, \alpha), \gamma\right)$ in ~$\Pred_1(A)$ we want to show that ~$\alpha = \beta\circ \gamma$. Note that ~$\Omega^G\times G^F$ and ~$\Omega^F$ are sheaves in ~$\Sh{\SliceSite*{A}}$, and that ~$(\beta, \gamma)\in (\Omega^G\times G^F)(\id_A)$ and ~$\alpha\in \Omega^F(\id_A)$ are amalgamations for the cover $S$ of $\id_A$. Composition in ~$\Sh{\SliceSite*{A}}$ is given by the map of sheaves 
\begin{equation*}
\circ^A\colon \Omega^G\times G^F\to \Omega^F
\end{equation*}
	
Through the matching condition, for all ~$f\colon \dom f\to A$ we have that ~$\circ^A_f\colon (\beta_f, \gamma_f)\mapsto \beta_f\circ \gamma_f$ with ~$\beta_f\circ\gamma_f = \alpha_f$.
	
For ~$\circ^A$ to be a valid presheaf morphism we therefore also need to have ~$\circ_{\id_A}^A\colon (\beta, \gamma)\mapsto \gamma\circ \beta$ with ~$\gamma\circ\beta = \alpha$.
Putting everything together, we have that ~$\left((F, \alpha), (B, \beta), \gamma\right) \in \Pred_1(A)$, proving that ~$\Pred_1$ is a sheaf.
\end{proofE}

As promised, the pair~$(\Pred_0, \Pred_1)$ provide the underlying sheaves of a category internal to $\Sh{\CC, J}$.
The remaining structure is obtained from the corresponding structure of~$\Sh[0]{\SliceSite*{A}}$.
In detail, we define the \emph{source} and \emph{target morphisms} for each
$\left((F, \alpha), (G, \beta), \gamma\right) \in \Pred_1(A)$ by
\begin{equation*}
s^P_A\left((F, \alpha), (G, \beta), \gamma\right) = (F, \alpha) 
\quad\text{and}\quad
t^P_A\left((F, \alpha), (G, \beta), \gamma\right) = (G, \beta);
\end{equation*}
the \emph{identity morphism} is defined for each~$(F, \alpha) \in \Pred_0(A)$ by
\begin{equation*}
e^P_A(F, \alpha) = \left((F, \alpha), (F, \alpha), id_F\right);
\end{equation*}
\emph{composition} is defined on~$\left(((F, \alpha), (G, \beta), \gamma), ((G, \beta), (H, \delta), \epsilon)\right)$ in~$\pb{\Pred_1(A)}{\Pred_1(A)}{\Pred_0(A)}$ by
\begin{equation*}
c^P_A(P) = \left((F, \alpha), (H, \delta), \epsilon\circ \gamma\right).
\end{equation*}

\begin{theoremE}[][category=internal-categories-sheaves]
The tuple~$\Pred = (\Pred_0, \Pred_1, s^P, t^P, e^P, c^P)$ is an internal category in~$\Sh{\CC, J}$.
\end{theoremE}
\begin{proofE}
	Again, since we use the structure maps of the ambient category, the internal category axioms are satisfied. We check the structure maps are well-defined presheaf morphisms in ~$\Sh{\CC, J}$. For all ~$A\in\CC$, ~$s^P_A$ and ~$t^P_A$ are well typed by construction.
	Given ~$(F, \alpha) \in \Pred_0(A)$ we have that ~$e^P_A(F, \alpha) \in Pred_1(A)$, as ~$\alpha\circ\id_F = \alpha$.
	For
	\begin{equation*}
		P \coloneqq \left(((F, \alpha), (G, \beta), \gamma), ((G, \beta), (H, \delta), \epsilon)\right) \in \pb{\Pred_1(A)}{\Pred_1(A)}{\Pred_0(A)}
	\end{equation*}
	
	\noindent we have that ~$c_A^P(P) \in \Pred_1(A)$ as ~$\alpha = \beta\circ\gamma = \delta\circ\epsilon\circ\gamma$.
	
	Because ~$\Pred_0$ is defined as a coproduct ranging over sheaves in ~$\Sh[0]{\SliceSite*{A}}$ and ~$\Pred_1$ is defined as an equaliser over maps in ~$\Sh[0]{\SliceSite*{A}}$, both varying functorially in ~$A$, and reindexing along ~$f$ (through $\bparens{\overline{\slice{\CC}{g}}}^{\ast}$) respects these (co)limits, the structure maps are valid presheaf morphisms, natural in ~$\CC$.
\end{proofE}

As ~$\Pred_0$ and ~$\Pred_1$ are defined using (co)limits indexed by ~$\Univ_0$ and ~$\Univ_1$, there is a straightforward way to project the predicate sheaves to the universe sheaves through the mappings 
\begin{equation*}
  (p_0, p_1)\colon (\Pred_0, \Pred_1)\to (\Univ_0, \Univ_1)
\end{equation*} 

\noindent defined pointwise for ~$A\in \CC$ as
\begin{align}
	\label{eq:internal-predicate-functor}
	p_{0,A}\colon (F, \alpha) \in \Pred_0(A) \mapsto F \in \Univ_0(A)\\
	p_{1,A}\colon \left( (F, \alpha), (G, \beta), \gamma \right) \in \Pred_1(A) \mapsto \gamma \in \Univ_1(A)
\end{align}

We now verify that these projections assemble into an internal functor.

\begin{theoremE}[][category=internal-categories-sheaves]
	The mappings ~$p_0$ and ~$p_1$ defined in~\cref{eq:internal-predicate-functor} form the internal functor ~$p\colon \Pred\to \Univ$.
\end{theoremE}
\begin{proofE}
	We show ~$p_0, p_1$ are well-defined and satisfy the axioms of an internal functor:
	\begin{itemize}
		\item $p_0$ and $p_1$ are valid presheaf morphisms: Let ~$A, B,f\colon B\to A$ in ~$\CC$
		\begin{itemize}
			\item \textbf{object component:} For any ~$(F, \alpha) \in \Pred_0(A)$:
			\begin{align*}
				p_0(\Pred_0(f)(F,\alpha)) = p_0(\bparens{\overline{\slice{\CC}{g}}}^{\ast}(F), \bparens{\overline{\slice{\CC}{g}}}^{\ast}(\alpha))\\
				= \bparens{\overline{\slice{\CC}{g}}}^{\ast}(F) = \Univ_0(f)(\bparens{\overline{\slice{\CC}{g}}}^{\ast}F) = \Univ_1(f)(p_1(F,\alpha))
			\end{align*}
			\item \textbf{morphism component:} For any ~$\left( (F, \alpha), (G, \beta), \gamma \right) \in \Pred_1(A)$
			\begin{align*}
				p_1(\Pred_1(f)(\left( (F, \alpha), (G, \beta), \gamma \right) \in \Pred_1(A)))\\
				= p_1(\left( (\bparens{\overline{\slice{\CC}{g}}}^{\ast}F, \bparens{\overline{\slice{\CC}{g}}}^{\ast}\alpha), (\bparens{\overline{\slice{\CC}{g}}}^{\ast}G, \bparens{\overline{\slice{\CC}{g}}}^{\ast}\beta), \bparens{\overline{\slice{\CC}{g}}}^{\ast}\gamma \right) \in \Pred_1(A)) \\
				= \bparens{\overline{\slice{\CC}{g}}}^{\ast}(\gamma) = \Univ_1(\gamma) = \Univ_1(f)(p_1\left( (F, \alpha), (G, \beta), \gamma \right))
			\end{align*}
		\end{itemize}
		\item ~$p$ satisfies the axioms of an internal functor:
		\begin{itemize}
			\item \textbf{respect for source and target:} Let ~$\left( (F, \alpha), (G, \beta), \gamma \right) \in \Pred_1(A)$
			\begin{align*}
				p_0(s^P_A(\left( (F, \alpha), (G, \beta), \gamma \right))) = p_0(F, \alpha) = F = s^U_A(\gamma)\\
				= s^U_A(p_1(\left( (F, \alpha), (G, \beta), \gamma \right)))
			\end{align*}
			similarily for target
			\begin{align*}
				p_0(t^P_A(\left( (F, \alpha), (G, \beta), \gamma \right))) = p_0(G, \beta) = G = s^U_A(\gamma)\\
				= t^U_A(p_1(\left( (F, \alpha), (G, \beta), \gamma \right)))
			\end{align*}
			\item \textbf{respect for identities:} Let ~$(F, \alpha)\in \Pred_0(A)$
			\begin{align*}
				p_1(e^P_A(F, \alpha)) = p_1\left((F, \alpha), (F, \alpha) \id_F\right) = \id_F = e^U_F(F) = e^U_A(p_0(F, \alpha))
			\end{align*}
			\item \textbf{respect for composition:} for any ~
      \begin{equation*}
        P \coloneqq \left(((F, \alpha), (G, \beta), \gamma), ((G, \beta), (H, \delta), \epsilon)\right) \in \pb{\Pred_1(A)}{\Pred_1(A)}{\Pred_0(A)}
      \end{equation*}
			\begin{align*}
				p_1(c^P_A(P)) = p_1\left((F, \alpha), (H, \delta), \epsilon\circ \gamma\right) = \epsilon\circ \gamma = c^U_A(\gamma,\epsilon)\\
				= c^U_A(\pb{p_1}{p_1}{p_0}(P))
			\end{align*}
		\end{itemize}
	\end{itemize}
\end{proofE}

This completes the internal construction of predicates over small sheaves and ties it structurally to the underlying universe.
To bring the construction back to the ground, we briefly unpack the internal construction in the degenerate case where the sheaf topos reduces to the category of sets.

\bsnote{cite?}
	For the case where ~$\Sh{\CC, J} \cong \Set_1$, the internal functor ~$p$ corresponds with the external predicate fibration 
	
	\noindent $p_{\mathrm{ext}}\colon \Pred_{\mathrm{ext}}\to \Set_0$, where ~$\Pred_{\mathrm{ext}}$ is the category with
	\begin{itemize}
		\item \textbf{objects:} pairs ~$(X \in \Set_0, P\subseteq X)$
		\item \textbf{morphisms:} given ~$f\colon X\to Y$ in ~$\Set_0$, there is a morphism from ~$(X, P\subseteq X)$ to ~$(Y, Q\subseteq Y)$, if ~$\img{f}(P)\subseteq Q$.
	\end{itemize}
	
	Then ~$p_{\mathrm{ext}}$ is the forgetful functor projecting ~$\Pred_{\mathrm{ext}}$ to ~$\Set_0$.

\subsection{Internal Predicate Fibration}
\label{sec:internal-fibration}

We now show that the internal functor~$p \colon \Pred \to \Univ$ can play a role analogous to the usual predicate fibration of sets in categorical logic~\cite{Jacobs99}.
In particular, we prove that $p$ is an internal fibration~\cite{Street74} in the~2-category~$\Int{\Sh{\CC, J}}$.
Without going into the details, this 2-categorical concept of fibration carries the same intuition as the one from ordinary category theory: a fibration provides Cartesian liftings which formalise reindexing predicates along maps in the base category.
This reindexing is analogous to taking pre-images, as discussed in \cref{sec:introduction}.

\begin{theoremE}[][category=internal-categories-sheaves]
\label{thm:internal-fibration}
The internal functor~$p\colon \Pred\to \Univ$ is an internal fibration in~$\Int{\Sh{\CC, J}}$.
\end{theoremE}
\begin{proofE}
	For every 2-cell $\beta$
	\begin{equation*}
		\begin{tikzcd}
			\IntCat{X} \arrow[r, "e"] \arrow[dr, "b"{name=U, below}, pos=0.2] & \Pred \arrow[d, "p"{name=D}]\\
			& \Univ
			\arrow[Rightarrow, from=U, to=D, shorten <=5pt, shorten >=3pt, "\beta"]
		\end{tikzcd}
	\end{equation*}

	\noindent we need to show there exists a $p$-cartesian 2-cell
	\begin{equation*}
		\begin{tikzcd}
			\IntCat{X} 
			\arrow[r, bend left=50, "e'"{name=U}]
			\arrow[r, bend right=50, "e"{name=D, below}]
			& \Pred
			\arrow[Rightarrow, from=U, to=D, shorten <=7pt, shorten >=7pt, "\alpha"]
		\end{tikzcd}
	\end{equation*}

	\noindent such that ~ $p\alpha = \beta$.
	
	We start by analysing what exactly a generalised object ~$ e \in [\IntCat{X}, \Pred]$ is.
	It is an internal functor with components
	\begin{equation*}
		e_0\colon \IntCat{X}_0\to \Pred_0 \quad e_1\colon \IntCat{X}_1\to \Pred_1
	\end{equation*}
	
	We can see these components for any ~$A\in \CC$ as families indexed by $x\in \IntCat{X}_0(A)$:
	\begin{align*}
		e_{0, A} = \{(x \in \IntCat{X}_0(A), (F_x, \phi_x) \in \Pred_0(A))\}
		e_{1, A} = \{(f\colon x\to x^\prime \in \IntCat{X}_1(A), (\phi_x, \psi_x, \gamma_x) \in \Pred_1(A))\}
	\end{align*}

	\noindent and similarily
	\begin{align*}
		b_{0, A} = \{(x \in \IntCat{X}_0(A), F_x \in \Univ_0(A))\}
		b_{1, A} = \{(f\colon x\to x^\prime \in \IntCat{X}_1(A), \gamma_x \in \Univ_1(A))\}
	\end{align*}

	Then ~$\beta\colon \IntCat{X}_0\to \Univ_1$ can also be seen as a family
	\begin{equation*}
		\beta_A = \{(x\in \IntCat{X}_0(A), f_x\colon F_x \to p_0(G_x, \beta_x))\}
	\end{equation*}

	\noindent with ~$f_x = \beta_A(x), F_x=b_{0,A}(x)$, and $(G_x, \beta_x) = e_a(x)$.
	
	We now construct the $p$-cartesian lift 
	\begin{equation*}
		\alpha\colon e^\prime \Rightarrow e \quad \alpha\colon \IntCat{X}_0\to \Pred_1
	\end{equation*}

	\noindent as follows: ~$e_{0,A}(x)$ characterises a subobject ~$P_x\subseteq G_x$. We can form the pullback of this subobject along ~$\beta_A(x)$ to get the subobject
	\begin{equation*}
		Q_x \coloneqq \beta_a(x)^* P_x\subseteq F_x
	\end{equation*}

	\noindent characterised by ~$e^\prime_{0,A}(x)$.

	We then define $\alpha$ componentwise as
	\begin{align*}
		\alpha_A\colon \IntCat{X}_0(A)\to \Pred_1(A)\\
		\alpha_A \colon x\mapsto \left((b_{0, A}(x), e^\prime_{0, A}(x)), (p_0(E_{A, 0}(x)), e_{A, 0}(x)), \beta_A(x)\right)
	\end{align*}
	
	By construction we have for all $A$ in $\CC$ and all ~$x\in \IntCat{X}_0(A)$ that 
	\begin{equation}
		\label{eq:firstfib}
		p_{0, A}(\alpha_A(x)) = \beta_A(x)
	\end{equation}
	
	We also have that $\alpha_A(x)\in \Pred_1(A)$: we can see ~$\alpha_A(x)$ as a triple
	\begin{equation*}
		\alpha_A(x) = ((F_x, \phi_x), (G_x, \psi_x), \beta_A(x))
	\end{equation*}

	\noindent where $\phi_x$ characterises a subobject ~$P_x\subseteq G_x$, and $\phi_x$ characterises the subobject ~$Q_X\subseteq F_x$ obtained by pulling $\psi_X$ back along $\gamma_A(x)$
	\begin{equation*}
		Q_x \coloneqq \beta_A(x)^* P_x
	\end{equation*}
	
  \noindent so by definition ~$\phi_x = \psi_x\circ \beta_A(x)$.
	
	Now it remains to check the lifting conditions hold: we start with the following data:
	\begin{itemize}
		\item A 1-cell ~$h\colon \IntCat{Y}\to \IntCat{X}$
		\item A 1-cell ~$e^{\prime\prime}\colon \IntCat{Y}\to\Pred$
		\item A 2-cell ~$\xi\colon \IntCat{Y}_0\to\Pred_1$ between ~$e^{\prime\prime}\Rightarrow \alpha\circ h$
		\item A 2-cell ~$\theta\colon \IntCat{Y}_0\to \Pred_1$ between ~$p\circ e^{\prime\prime}\Rightarrow p\circ e^\prime \circ h$
	\end{itemize}

	\noindent satisfying
	\begin{equation*}
		p\xi = p\alpha h\circ \theta
	\end{equation*}

	We will show there exists a unique 2-cell ~$\zeta\colon \IntCat{Y}_0\to \Pred_1$ between ~$e^{\prime\prime}\Rightarrow e^\prime\circ h$ satisfying 
	\begin{equation*}
		p\zeta = \theta \quad \xi = \alpha h\circ \zeta
	\end{equation*}
	
	The idea again is to work pointwise over objects $A\in \CC$ and elements $y\in \IntCat{Y}_0(A)$.
	Then we can explicitly state
	\begin{itemize}
		\item \begin{equation}
			\label{eq:secondfib}
			\xi_A(y) \coloneqq \left((H_y, \epsilon_y), (G_{h(y)}, \psi_y), \delta_y\right)
		\end{equation}

		\noindent where $\psi_{h(y)}$ characterises a subobject ~$ P_{h(y)}\subseteq G_{h(y)}$, and $\epsilon_y$ a subobject ~$R_y\subseteq H_y$ and
		\begin{equation}
			\label{eq:thirdfib}
			\epsilon_y = \psi_{h(y)}\circ \delta_y \Rightarrow R_y\subseteq P_{h(y)}
		\end{equation}
		\item
		\begin{equation}
			\label{eq:fourthfib}
			\theta_A(y) \coloneqq \kappa_y\colon H_y\to F_{h(y)}
		\end{equation}

		\noindent with ~$F_{h(y)} = p_{0,A}(e^\prime_A(h(y)))$ 
	\end{itemize}
	such that 
	\begin{equation}
		\label{eq:fifthfib}
		p_A\xi_A = p_a\alpha_A h_A\circ \theta_A
	\end{equation}
	
	We now construct the 2-cell	!$\zeta\colon e^{\prime\prime}\Rightarrow e^\prime \circ h$ with
	\begin{equation*}
		\zeta_A(y) \coloneqq \left((H_y, \epsilon_y), (F_{h(y)}, \phi_{h(y)}), \kappa_y\right)
	\end{equation*}

	\noindent where 
	\begin{equation*}
		e^{\prime\prime}_{0, A}(y) = (H_y, \epsilon_y) \quad e^\prime_{0,A}(h_{0,A}(y)) = (F_{h(y)},  \phi_{h(y)}),
	\end{equation*}
	
	To check that $\zeta_A(y)$ defines a valid internal natural transformation we need to show that $\zeta_A(y)\in \Pred_1(A)$.
	For this we show ~$\epsilon_y = \phi_{h(y)} \circ \kappa_y$.
	\begin{align}
		\epsilon_y = \psi_{h(y)}\circ \delta_y \tag{By~\cref{eq:secondfib}}\\
		= \psi_{h(y)}\circ p\xi_A(y) \tag{By definition of $p$}\\
		= \psi_{h(y)}\circ (p_A\alpha_A h_A)(y)\circ \theta_A(y) \tag{By assumption of~\cref{eq:fifthfib}} \\
		= \psi_{h(y)}\circ \beta_A(h(y))\circ \theta_A(y) \tag{By~\cref{eq:firstfib}}\\
		= \phi_{h(y)}\circ \beta_A(h(y))\circ \theta_A(y) \nonumber\\
		= \phi_{h(y)}\circ \kappa_y \tag{By definition~\cref{eq:fourthfib}}
	\end{align}
	
	That ~$\phi_{h(y)} = \psi_{h(y)}\circ\beta_A(h_{0,A}(y))$ holds is because ~$e^\prime_{0, A}(h_{0, A}(y))$ characterises a subobject ~$Q_{h(y)}\subseteq F_{h(y)}$ defined by pulling back $G_{h(y)}$ along ~$\beta_A(h_{0, A}(y))$.
	
	We verify the triangle identities:
	\begin{itemize}
		\item We have that ~$\alpha h\circ \zeta = \xi$:
		\begin{align*}
			(\alpha h_A\circ \zeta_A)(y) = \left( (H_y, \epsilon_y), (G_{h(y)}, \psi_{h(y)}), \beta_A(h(y))\circ \kappa_y\right)\\
			\xi_A(y) = \left( (H_y, \epsilon_y), (G_{h(y)}, \psi_{h(y)}), \delta_y\right)
		\end{align*}

		\noindent and that ~$\beta_A(h(y))\circ \kappa_y = \delta_y$ follows from the assumption ~$p\alpha h\circ \theta = p\xi$.
		\item We have that ~$p\zeta = \theta$ by definition of $\zeta$, as it's third component is ~$\kappa_y$ and we define ~$\theta_A(y) = \kappa_y$.
	\end{itemize}
	
	Uniqueness holds: assume there exists ~$\zeta^\prime\colon \IntCat{Y}_0\to \Pred_1$ between ~$ e^{\prime\prime}\Rightarrow e^\prime\circ h$ such that ~$p\zeta^\prime = \theta$ and ~$ \alpha h\circ \zeta^\prime = \xi$. In order for these equalities to hold and as the second component in ~$\zeta_A(y)$ is given by the pullback property, only the first component in ~$\zeta_A^\prime(y)$ can vary from ~$\zeta_A(y)$. Therefore we define
	\begin{equation*}
		\zeta_A^\prime(y) = \left((H_y, \epsilon^\prime_y), (F_{h(y)}, \phi_{h(y)}), \kappa_y\right)
	\end{equation*}
	
	But now as ~$\zeta^\prime_A, \zeta_A\in \Pred_1(A)$, we must have that
	\begin{equation*}
		\epsilon^\prime = \phi_{h(y)}\circ \kappa_y = \epsilon_y
	\end{equation*}
	
  \noindent showing the first component also must be the same, and therefore ~$\zeta^\prime = \zeta$ for all ~$A\in \CC$ and ~$y\in \IntCat{Y}_0$.
\end{proofE}

We will not go into the details of internal fibrations and instead make use of the consequences of in order to develop resource-aware logic.
First of all, we can obtain fibre categories with  internal reindexing functors that arise from the fibration structure.

\begin{definition}
  \label{def:internal-fibre}
  Given a global section $X \from \IntCat{1}\rightarrow \Univ$, the \emph{fibre} of the fibration $p \from \Pred \to \Univ$ is defined by the following pullback in the category of internal categories $\Int{\Sh{\CC, J}}$.
    \begin{equation*}
        \begin{tikzcd}
            \IntCat{\Pred_X} \arrow[r, "\iota_X"] \arrow[d, "!"{swap}] & \Pred \arrow[d, "p"] \\
            \IntCat{1} \arrow[r, "X", below] & \Univ
        \end{tikzcd}
    \end{equation*}

    This defines an internal category ~$\Pred_{X}$ in ~$\Sh{\CC, J}$ with object and morphism components given by the pullbacks in ~$\Sh{\CC, J}$:
    \begin{equation*}
      \begin{tikzcd}
        \Pred_{X, 0} \arrow[r, "\iota_{X, 0}"] \arrow[d, "!"] & \Pred_0 \arrow[d, "p_0"] & \Pred_{X, 1} \arrow[rr, "\iota_{X, 1}"] \arrow[d, "!"] & & \Pred_1 \arrow[d, "p_1"] \\
        \TSheaf \arrow[r, "X"] & \Univ_0 & \TSheaf \arrow[r, "X"] & \Univ_0 \arrow[r, "e^U"] & \Univ_1
      \end{tikzcd}
    \end{equation*}
    Source, target, identity, and composition morphisms come from the universal property of pullbacks.
  \end{definition}

\begin{example}[Predicates over Memory]
	Given the poset category of memory locations, $\caL$, and the memory sheaf $M$ as defined in~\cref{expl:powerset}, the internal predicate fibration ~$p\colon \Pred\to \Univ$ in $\Cat{\Sh{\caL}}$ and the global section $~\bar{M}$ in ~$\Sh{\caL}$ internalising the small sheaf $M$ as in~\cref{thm:internalising-small-sheaves}, induces the internal fibre ~$\Pred_M$.
  Recall that ~$\overline{M}$ is defined for $U\in \caL$ as ~$\overline{M}_U(*) = M\circ \op{(\dom^U)}\colon \op{(\caL/U)}\to \Set$.
  
  The components ~$\Pred_{M, 0}$ and ~$\Pred_{M, 1}$ are explicitly for $U \in \caL$ as
  \begin{align*}
		\Pred_{M, 0}(U) &= \Omega_U^{\overline{M}_U(*)}=  \{\alpha_U\colon F_U\to \Omega_U \in \Pred_0{U} \mid \overline{M}_U(*) = F_U\}  \\
		\Pred_{M, 1}(U) &= \{\left(\alpha_U\colon F_U\to \Omega_U, \beta_U\colon G_U\to \Omega_U, \gamma_U\colon F_U\to G_U\right) \mid \alpha_U = \beta_U\circ \gamma_U \wedge \id_{\gamma_U} = \id_{\overline{M}_U(*)} \}
	\end{align*}
 
	For all ~$U^\prime\subseteq U$ we have that $\overline{M}_U(*)(U^\prime\subseteq U) = [U^\prime, \mathbf{Val}]$, showing that $\overline{M}_U(*)$ is equal to $M$ restricted to $U$. From this we obtain
	\begin{itemize}
		\item $\Pred_{M, 0}(U)$ contains all subobjects $P_U\subseteq M_U$, with for all $U^\prime\subseteq U$, $P_U(U^\prime\subseteq U) \subseteq [U, \mathbf{Val}]$.
		\item $\Pred_{M, 1}(U)$ contains all subobjects $P_U, Q_U \subseteq M_U$, such that $P_U \subseteq Q_U$. For all $U^\prime\subseteq U$, this is an inclusion $P_U(U^\prime\subseteq U) \subseteq Q_U(U^\prime\subseteq U) \subseteq [U, \mathbf{Val}]$.
	\end{itemize}
We have that $P_U(\id_U) \subseteq [U, \mathbf{Val}]$, so we write $P_U = P_U(\id_U)$ for predicates over $[U, \mathbf{Val}]$.

The sheaf condition allows for gluing consistent predicates: 
for predicates ~$P_{U_1}\subseteq [U_1, \mathbf{Val}]$ and~$P_{U_2}\subseteq [U_2, \mathbf{Val}] $
such that for all $\sigma_{U_1} \in P_{U_1}$ and $\sigma_{U_2}\colon U_2\to\mathbf{Val} \in P_{U_2}$ we have that
	\begin{equation*}
		\sigma_{U_1|U_1\cap U_2} = \sigma_{U_2|U_1\cap U_2},
	\end{equation*}
then there exists a unique predicate ~$P\subseteq [U_1\cup U_2, \mathbf{Val}]$ which restricts to $P^{U_1}$ and $P^{U_2}$.
\end{example}

In \cref{thm:internal-fibration} we established that $\Pred$ is an internal fibration.
This enables, in principle, a form of substitution or pre-image but that is not immediately visible from the definition of 2-fibrations.
The following theorem uses the 2-fibration structure to provide internal reindexing functors between fibres that expose the substitution operation.
Moreover, following standard categorical logic~\cite{Jacobs99}, we show that there are internal functors that implement existential quantification.
We will use these functors in \cref{sec:intepreting-separating-conjunction} to provide internal semantics for separating conjunction.
A full interpretation of substitution or first-order quantification would require a significantly more elaborate setup, including 2-dimensional Day convolution and fibrations, and is left for future work.

\begin{theoremE}[][category=internal-categories-sheaves]
  \label{thm:internal-reindexing}
  For all $X, Y\colon \TSheaf\to \Univ_0$ and morphisms $f \colon \TSheaf\to \Univ_1$ between $X$ and $Y$, there are internal functors
  \begin{equation*}
    \reidx{f} \colon \IntCat{E}_Y\to \IntCat{E}_X \quad \text{and} \quad \exists^f\colon \IntCat{E}_X\to \IntCat{E}_Y \, .
  \end{equation*}
\end{theoremE}
\begin{proofE}
  Sketch of Construction:
	For the proof and an explicit definition of the construction, 
  
  \noindent see~\cref{prop:internal-reindexing}.

    \begin{itemize}
        \item The pullback defining ~$\IntCat{E}_Y$ provides a $p$-cartesian 2-cell ~$\bar{f}_{\mathrm{cart}}$ above ~$\bar{f}$ with domain ~$\IntCat{E}_{Y, 0}$, mapping to ~$\IntCat{E}_1$.
        \item The pullback defining ~$\IntCat{E}_{X, 0}$ together with the cartesian 2-cell ~$\bar{f}_{\mathrm{cart}}$ enables the definition of the reindexing map 
        \begin{equation*}
            f^*_0\colon \IntCat{E}_{Y, 0}\to \IntCat{E}_{X, 0}
        \end{equation*}
        on objects.
        \item Construct 2-cells ~$\xi$ and ~$\gamma$ defining coherence between morphism spaces and internal reindexing.
        \item These induce a unique 2-cell ~$\bar{f}_{\mathrm{lift}}$, enabling the definition of the reindexing map 
        \begin{equation*}
            f^*_1\colon \IntCat{E}_{Y, 1}\to \IntCat{E}_{X, 1}
        \end{equation*}
        on morphisms.
        \item Dually, in ~$\co{\Int{\Sh{\CC, J}}}$, the 2-category obtained by reversing the 2-cells in ~$\Int{\Sh{\CC, J}}$, construct the opcartesian lift ~$\op{\bar{f}}_{\mathrm{cart}}$ and lift ~$\op{\bar{f}}_{\mathrm{lift}}$ to obtain the maps
        \begin{equation*}
            \exists^f_0\colon \IntCat{E}_{X, 0}\to \IntCat{E}_{Y, 0} \quad \text{and} \quad \exists^f_1\colon \IntCat{E}_{X,1}\to\IntCat{E}_{X, 1}
        \end{equation*}
        \item The maps ~$f^*_0$ and ~$f^*_1$ together form the internal functor
        \begin{equation*}
            f^*\colon \IntCat{E}_Y\to \IntCat{E}_X
        \end{equation*}
        The maps ~$\exists^f_0$ and ~$\exists^f_1$ together form the internal functor 
        \begin{equation*}
            \exists^f\colon \IntCat{E}_X\to \IntCat{E}_Y
        \end{equation*}
    \end{itemize}
\end{proofE}

    
We now illustrate the internal existential quantification by applying it to the Day convolution monoid structure on the memory sheaf.
This illustrates a step towards internalising separating conjunction as a logical connective in our framework.

\begin{example}
  \label{ex:existential-separating}
	Let ~$\caL$ be the poset of memory locations and ~$M$ the memory sheaf from~\cref{ex:partial-finitary-memory}, with monoid product 
	\begin{equation*}
		\mu\colon M\DaySH M \to M
	\end{equation*}
	from~\cref{ex:partial-memory-monoid}.
	
	This internalises to global sections ~$\overline{M\DaySH M}, \overline{M}\colon \TSheaf \to \Univ_0$ and a 2-cell 
  
  \noindent ~$\overline{\mu}\colon \TSheaf \to \Univ_1$ in ~$\Int{\Sh{\caL}}$ by~\cref{thm:internalising-small-sheaves}.
	
	Given the internal predicate fibration ~$p\colon \Pred \to \Univ$, the internal fibres ~$\Pred_{M\DaySH M}, \Pred_M$ and transformation ~$\overline{\odot}$ induce
	\begin{equation*}
		\exists^\mu\colon \Pred_{M\DaySH M} \to \Pred_M.
	\end{equation*}
	
	At stage ~$U \in \caL$, this maps
	\begin{equation*}
		P^U \subseteq \coprod_{U = U_1 \cup U_2} M(U_1) \times M(U_2)
	\end{equation*}
	to the amalgamated predicate
	\begin{equation*}
		\exists^\odot_U(P^U) = \{ \sigma \in M(U) \mid \exists (\sigma_1, \sigma_2) \in P_U .\, \sigma = m_{U_1, U_2}(\sigma_1, \sigma_2) \}
	\end{equation*}
	with ~$(\sigma_1 \cup \sigma_2)(x) = \sigma_1(x)$ if ~$x \in U_1$, and ~$\sigma_2(x)$ if ~$x \in U_2$.
  
	Intuitively, a memory $\sigma \in M(U)$ satisfies $P * Q$ if it can be built by gluing together two compatible submemories $\sigma_1$ and $\sigma_2$ that satisfy $P$ and $Q$ respectively.
	Though not defining separating conjunction as a connective 
	
	\noindent $\Pred_M \otimes \Pred_M \to \Pred_M$, this quantification captures its semantics via predicate amalgamation.
\end{example}

\subsection{Internal Logic of Resource-Aware Predicates}
\label{sec:internal-propositional-logic}

Using that the subobject classifier admits an internal Heyting algebra structure in the topos $\Sh{\CC, J}$, we can equip the fibres of the internal fibration $p \colon \Pred \to \Univ$ also with the connectives of intuitionistic propositional logic.
The twist is of course that this logic becomes resource-aware.
We demonstrate this on the following standard syntax for propositional logic.
\begin{equation*}
	\phi, \psi \Coloneqq \top \mid \bot \mid \phi \wedge \psi \mid \phi \vee \psi \mid \phi \to \psi
\end{equation*}
This logic and its interpretation can be extended with quantifier, but we leave these out to avoid difficulties with names.
We can then interpret formulas $\phi$ as predicates on a given resource sheaf $F$, which are global sections of the object part of the fibre $\Pred_{F}$ of the predicate fibration over $F$.
That is, the semantics of $\phi$ is natural transformation
\begin{equation*}
  \sem{\phi} \colon \TSheaf \to \Pred_{F,0} \, .
\end{equation*}
Using the Heyting algebra structure of subobjects
\begin{equation*}
  \intTop, \intBot \from \TSheaf \to \Pred_{F,0}
  \quad \text{and} \quad
	\intAnd, \intOr, \intTo \from \Pred_{F,0} \times \Pred_{F,0} \to \Pred_{F,0} \, ,
\end{equation*}
the semantics of formulas is defined iteratively and for $A \in \CC$ as follows.
Here, we write the binary operators in infix notation as abbreviation for their composition with the pairing, that is, $\sem{\phi} \intAnd \sem{\psi}$ stands for $\intAnd \comp \pair{\sem{\phi}, \sem{\psi}}$.
\begin{align*}
	\sem{\top} &= \intTop \\
	\sem{\bot} &= \intBot \\
	\sem{\phi \land \psi} &= \sem{\phi} \intAnd \sem{\psi} \\
	\sem{\phi \lor \psi} &= \sem{\phi} \intOr \sem{\psi} \\
	\sem{\phi \to \psi} &= \sem{\phi}\intTo \sem{\psi}
\end{align*}

This interpretation of logical connectives is resource-aware by or construction of the predicate fibration.
For example, the semantics of implication can be unfolded to
\begin{equation*}
  \sem{\phi \to \psi}_{A}(\ast)
  = \setDef{s \in F(A)}{\all{f\colon B \to A} F(f)(s) \in \sem{\phi}_{B}(\ast) \Rightarrow F(f)(s) \in \sem{\psi}_{B}(\ast)} \, ,
\end{equation*}
which combines implication with the possible restrictions on views and is akin to Kripke semantics.
Thus, we obtain the negation-free fragment of intutitionistic propositional logic, cf.~\cite[Sec.~4.1]{Pym02:SemanticsProofTheory}.

%% file: interpreting-logic-generalised.tex
\section{Sheaves for Separation Logic}
\label{sec:internallogic}
In this section, we assume that the~$(\CC, J)$ has the structure of a (symmetric) monoidal category with tensor~$\otimes$ and unit object~$I$. 
We identify conditions on this structure which enable us to extend of the internal logic of~$\Sh[1]{\CC, J}$ at a resource sheaf~$F$ by separating connectives.
Informally, we require the following:
\begin{itemize}
\item 
the monoidal structure of~$(\CC, J)$ \emph{lifts} to Day convolution on sheaves, and
\item 
$F$ carries the structure of a monoid object for Day convolution in~$\Sh[1]{\CC, J}$.
\end{itemize}
Under these conditions, separating connectives arise a uniform manner independent of the resource sheaf under consideration (i.e.~the model structure of the ensuing separation logic). 
We begin by making the necessary conditions on~$(\CC, J)$ precise in the following.

\begin{definition}
\label{def:day-stable}
A monoidal site~$(\CC, J)$ is \emph{Day-stable} if the following conditions are obtained.
\begin{enumerate}
\item\label{it:day} 
Day convolution is closed under sheaves, i.e.~$F\Day G\in\Sh[1]{\CC, J}$ for any~$F,G\in\Sh[1]{\CC, J}$.
That is,~$\Day$ restricts to tensor~$\DaySH$ on~$\Sh[1]{\CC, J}$ along~$\Sh[1]{\CC, J}\hra\PSh{\CC}$.
\item 
$\DaySH$ preserves regular monomorphisms (i.e. subobjects) in both arguments.
\item 
For all $A \in \CC$, the domain functor $\sliceDom[A] \colon \slice{\CC}{A} \to \CC$ is lax monoidal.
\end{enumerate}
\end{definition}

To the best of our knowledge, it is an open problem to determine exact conditions on~$\tens$ which ensure that condition~\cref{it:day} is attained.
We only note that~$\Day$ is closed under sheaves if the sheafification functor~$a$ is \emph{strong} monoidal,
i.e.~$a(F\Day G)\cong a(F)\DaySH a(G)$.
In this case, Day's Reflection Theorem~\cite{Day72} ensures that~$\Sh[1]{\CC, J}$ is moreover closed.
Preservation of monomorphisms (i.e. subsheaves) is not automatic, but simplifies in case sheafification is strong monoidal:
it is then enough to check that~$\Day$ is preserves regular monos since sheafification is left exact.
The final condition is a technical condition that ensures we can `lift' convolution to predicates over a resource sheaf.
We leave it to future work to identify relaxations on these conditions. 

Hereafter, we assume that~$(\CC, J)$ is Day-stable and~$F\in\Sh[1]{\CC, J}$ carries the structure of a monoid object for~$\DaySH$ with multiplication~$\mu$ and unit~$\eta$.
This assumption enables us to apply convolution to predicates on a resource sheaf~$F$ and, in particular,  
we will construct a morphism
\begin{equation*}
\alpha\colon \Pred_{F, 0}\DaySH \Pred_{F, 0}\to \Pred_{F\DaySH F, 0}
\end{equation*}
expressing how predicates over individual resources can be combined into a predicate over their joint composition.
To this end, we note that our universe hierarchy is closed under Day convolution:

\begin{lemmaE}[][category=internal-logic]
For any sheaf~$F\in\Sh[1]{\CC, J}$, we have ~$F\DaySH F\in\Sh[1]{\CC, J}$.
\end{lemmaE}
\begin{proofE}
Let ~$\tens$ be the monoidal product on ~$\CC$.
We show that ~$(F\DaySH F)(A) \in \Set_1$ and ~$(f, s, t)\in \Set_0$ for all ~$A\in\CC$ and ~$(f, s, t) \in (F\DaySH F)(A)$.
	
We have that ~$(F \DaySH F)(A)$ is defined as the colimit
\begin{equation*}
		(F\DaySH F)(A) := \int^{B, C \in \CC} \mor{\CC}(A, B\cdot C) \times F(B)\times F(C)
\end{equation*}
As ~$\CC$ is small, we have that ~$\mor{\CC}(A, B\cdot C) \in \Set_0$. 
	
By assumption ~$F(B), F(C) \in \Set_1$, so ~$\mor{\CC}(A, B\cdot C)\times F(B)\times F(C)\in \Set_1$.
As the index of the colimit is in ~$\Set_1$ and universes are closed under (co)limits indexed by objects in the same universe, we have that ~$(F \DaySH F)(A) \in \Set_1$.
	
Now, elements in ~$(F\DaySH F)(A)$ are equivalence classes represented as tuples 
\begin{equation*}
	(f\in \CC(A, B\cdot C), s\in F(B), t\in F(C))
\end{equation*}
All indices are in ~$\Set_0$ so the tuple is in ~$\Set_0$.
\end{proofE}

We now turn to the categorical constructions.
The (external) monoid structure on~$F$ is not sufficient to interpret separating conjunction internally. 
In particular, certain resource sheaves (including the sheaf of strict memory) do not carry a monoid product for Day convolution.
This observation motivates the development of a monoidal structure on predicates which is independent from the monoidal structure on~$(\CC, J)$. 
To construct this, we must first formalise the notion of gluing predicate data across covers; a role played by matching objects and their amalgamation.
\begin{definition}
    Given a site ~$(\CC, J)$, let $F$ be a sheaf on $\Cat{C}$.
    For all ~$A \in \Cat{C}$, all covers ~$ S = \{f\colon \dom f\to A\} \in J(A)$, and all pullbacks
    \begin{equation*}
        \begin{tikzcd}
            \dom f \times_A \dom g \arrow[r, "p_f"] \arrow[d, "p_g"] & \dom f \arrow[d, "f"]\\
            \dom g \arrow[r, "g"] & c
        \end{tikzcd}
    \end{equation*}
    with ~$f, f \in S$, we define the \emph{matching object} ~$\MatchObject{F}{A}{\dom f}{\dom g}$ as the pullback object

    \begin{equation*}
        \begin{tikzcd}
            \MatchObject{F}{A}{\dom f}{\dom g} \arrow[r, "\pi_f"] \arrow[d, "\pi_g"] & F(\dom f) \arrow[d, "F(p_f)"]\\
            F(\dom g) \arrow[r, "F(p_g)"] & F(\dom f\times_A \dom g)
        \end{tikzcd}
    \end{equation*}
    explicitly defined as
    \begin{equation*}
        \MatchObject{F}{A}{\dom f}{\dom g} = \{(s_f, s_g) \in F(\dom f)\times F(\dom g) \ | \ F(p_f)(s_f) = F(p_g)(s_g)\}
    \end{equation*}
\end{definition}
The above definition extracts the subset of section pairs that agree on their overlap.
This forms the building block for assembling matching families across the entire cover.

\begin{example}
    Let ~$M\colon \op{\caL}\to \Set$ be the memory sheaf. 
    
    \noindent For the memory region ~$U := \{x_1, x_2, x_3\}$ and the cover ~$\{U_1 := \{x_1, x_2\}, U_2 := \{x_2, x_3\}\}$, given the memory states 
    \begin{equation*}
    	\sigma_1 := \{ x_1\mapsto 7, x_2\mapsto 3\} \in M(U_1) \quad \text{and} \quad \sigma_2 := \{ x_2\mapsto 3, x_3 \mapsto 9\}
    \end{equation*}
    
    \noindent the pair ~$(\sigma_1, \sigma_2)$ is an element in the matching object ~$\MatchObject{M}{U}{U_1}{U_2}$.

    Given the memory state ~$\sigma_2^\prime := \{x_2\mapsto -1, x_3\mapsto 9\}$, the pair ~$(\sigma_1, \sigma_2^\prime)$ is not an element in the matching object.
\end{example}
To glue together all compatible section pairs across a cover, we must consider matching families for all finite covers and identify those that agree under refinement. 
This leads us to define the presheaf of matching objects.

\begin{definition}[Matching Object Presheaf]
\label{def:MOpsh}
    Given a site ~$(\CC, J)$, let ~$F$ be a sheaf on ~$\CC$.
    The \emph{matching object presheaf} is the functor
    \begin{equation*}
        \MatchSh{F}\colon \op{\CC}\to\Set
    \end{equation*}
    defined for ~$A\in \Cat{C}$ as the filtered colimit
    \begin{equation*}
        \MatchSh{F} := \lim_{\substack{\rightarrow\\ S\in J(A)}} \coprod_{\dom f, \dom g\in S} \MatchObject{F}{A}{\dom f}{\dom g}
    \end{equation*}
This colimit is taken over all covering sieves ~$S \in J(A)$, with representatives identified up to refinement, ensuring that matching families over different covers represent the same element whenever they agree on a common refinement.

    An element of ~$\MatchSh{F}(A)$ is then an equivalence class with representatives

    \begin{equation*}
        [S, \{(s_f, s_g)\in \MatchObject{F}{A}{\dom f}{\dom g}\}_{f, g\in S} ]
    \end{equation*}
    where ~$ S = \{f\colon \dom f\to A\}$ is a covering sieve in ~$J(A)$. 
    
    Two representatives ~$[S, \{(s_f, s_g)\}_{f, g\in S}], [S^\prime, \{(s_f^\prime, s_g^\prime)\}_{f, g\in S^\prime}]$ are equivalent as 
    \begin{equation*}
    	[S, \{(s_f, s_g)\}] \sim [S^\prime, \{s_f^\prime, s_g^\prime\}] 
    \end{equation*}
    if-and-only-if there exists a covering sieve ~$T\subseteq S\cap S^\prime$, such that the matching families agree on $T$ (the common refinement) after pulling back.

    Then, for any ~$h\colon B\to A$ in ~$\Cat{C}$
    \begin{equation*}
        \MatchSh{F}(h)\colon \MatchSh{F}(A)\to \MatchSh{F}(B)
    \end{equation*}
    is defined for a representative ~$[S, \{(s_f, s_g)\}_{f, g\in S}]$ with ~$S = \{f\colon \dom f\to A\}$ as follows:
    \begin{enumerate}
        \item pull back ~$S$ along ~$h$:
        \begin{equation*}
            h^*S = \{\pi_f\colon \dom f\times_A B \to \dom f\} \in J(B)
        \end{equation*}
        \item for each ~$(f, g)\in S$ we can pull back each pair ~$(s_f, s_g)\in \MatchObject{F}{A}{\dom f}{\dom g}$ to get a pair
        \begin{equation*}
            (s_{f|B}, s_{g|B}) \in \MatchObject{F}{B}{\dom f\times_A B}{\dom g\times_A B}
        \end{equation*}
        with ~$s_{k|B} = F(\pi_{\dom k}\colon \dom k\times_A B\to \dom k)(s_k)$ with ~$k=\{f, g\}$.
        \item Define ~$\MatchSh{F}(h)([S, \{(s_f, s_g)\}_{f, g\in S}]) := [h^*S, \{(s_{f|B}, s_{g|B})\}_{f, g\in S}] \in \MatchSh{F}(B)$
    \end{enumerate}
\end{definition}

Note that we quotient over common refinements to ensure that matching families representing the same global object are identified, regardless of the cover they are described on. 
In poset sites, this means that the colimit selects the unique amalgamating section defined on the least upper bound ~$U_1\cup U_2$, rather than an arbitrary larger domain. 
Without this quotienting, the uniqueness of amalgamation would be violated.

\begin{propositionE}[][category=internal-logic]
    The matching object presheaf as defined in Definition~\ref{def:MOpsh} is well-defined.
\end{propositionE}
\begin{proofE}
    Given a site ~$(\CC, J)$, let ~$F$ be a sheaf on ~$\Cat{C}$. We will show that ~$\MatchSh{F}$
    \begin{enumerate}
        \item \textbf{is well-defined:} for all morphisms ~$h\colon B\to A$ in ~$\Cat{C}$ and all representatives 
        \begin{equation*}
            t := [S, \{(s_f, s_g \in \MatchObject{F}{A}{\dom f}{\dom g})\}_{f,g\in S}] \in \MatchSh{F}(A)
        \end{equation*}
        for the covering sieve ~$ S = \{f\colon \dom f\to A\}\in J(A)$ we have to show that 
        \begin{equation*}
            \MatchSh{F}(h)([S, \{(s_f, s_g)\}_{f, g\in S}]) \in \MatchSh{F}(B)
        \end{equation*}
        By definition of ~$J$ being a cartesian Grothendieck topology, we have that ~$f^*S \in J(d)$. 
        Now, denote 
        \begin{equation*}
            P := (\dom f\times_A B)\times_B (\dom g\times_A B) \cong (\dom f\times_A \dom g)\times_B B \quad \text{for any} \ f, g \in S
        \end{equation*}
        We have to show that for ~$(s_f, s_g) \in \MatchObject{F}{A}{\dom f}{\dom g}$ that 
        \begin{align*}
        	&\big( 
        	s_{f|B} = F(\pi_{\dom f} \colon \dom f \times_A B \to \dom f)(s_f), \\
        	&\quad s_{g|B} = F(\pi_{\dom g} \colon \dom g \times_A B \to \dom g)(s_g) 
        	\big) 
        	\in \MatchObject{F}{B}{\dom f \times_A B}{\dom g \times_A B}
        \end{align*}
        which amounts to showing that
        \begin{equation*}
            F(\pi_{\dom f}\circ \pi_1^\prime\colon P\to \dom f\times_A B)(s_f) = F(\pi_{\dom g}\circ \pi_2^\prime\colon P\to \dom g\times_A B)(s_g)
        \end{equation*}
        Because of the isomorphism in ~$P$, we have that
        \begin{align*}
            \pi_{\dom f}\circ \pi_1^\prime = p_{\dom f}\colon \dom f\times_A \dom g\to \dom f\circ \pi_1\colon P\to \dom f\times_A \dom g \\
            \pi_{\dom g}\circ \pi_2^\prime = p_{\dom g}\colon \dom f\times_A \dom g\to \dom g\circ \pi_1
        \end{align*}
        and as we assume that 
        \begin{equation*}
            F(p_{\dom f})(s_f) = F(p_{\dom g})(\dom g)
        \end{equation*}
        we have that 
        \begin{equation*}
            F(p_{\dom f}\circ \pi_1)(s_f) = F(p_{\dom g}\circ \pi_1)(s_g)
        \end{equation*}
        and therefore
        \begin{equation*}
            F(\pi_{\dom f}\circ \pi_1^\prime)(s_f) = F(\pi_{\dom g}\circ \pi_2^\prime)(s_g)
        \end{equation*}
        showing that for any ~$(s_f, s_g)\in \MatchObject{F}{A}{\dom f}{\dom g}$ that ~$(s_{f|B}, s_{g|B}) \in \MatchObject{F}{B}{\dom f\times_A B}{\dom g\times_A B}$, and therefore ~$\MatchSh{F}(h)(t) \in \MatchSh{F}(B)$,
        \item \textbf{is functorial:}
        \begin{itemize}
            \item For all ~$A \in \Cat{C}$, ~$\MatchSh{F}(\id_A) = \id_{\MatchSh{F}(A)}$ as pulling back along the identity does nothing.
            \item For all ~$h\colon B\to A, j\colon C\to B$ in ~$\Cat{C}$ we have ~$\MatchSh{F}(j\circ h) = \MatchSh{F}(j)\circ \MatchSh{F}(h)$ because pulling back covers is functorial: for any ~$S\in J(A)$ we have that
            \begin{equation*}
                (h\circ j)^*S = j^*(h^*S)
            \end{equation*}
            and pulling back sections is functorial: by functoriality of ~$F$, we have that
            \begin{equation*}
                F(\pi_j\colon (\dom f\times_A B)\times_B C\to \dom f\times_A B)\circ F(\pi_h\colon \dom f\times_A B\to B) = F(\pi_h\circ \pi_j)
            \end{equation*}
        \end{itemize},
        \item \textbf{respects equivalence classes:} for all ~$A\in\Cat{C}$ and all 
        \begin{equation*}
            x_S = [S, \{(s_f, s_g)\}_{f,g\in S}], x_R = [R, \{(r_f, r_g)\}_{f, g\in R}]\in \MatchSh{F}(A)
        \end{equation*} 
        such that ~$x_S\sim x_R$ we need to show that for all ~$k\colon B\to A$ in $\Cat{C}$ we have that
        \begin{equation*}
            \MatchSh{F}(k)(x_S) \sim \MatchSh{F}(k)(x_R)
        \end{equation*}
        Let ~$T$ be the common refinement of ~$R$ and ~$S$:
        \begin{equation*}
            T = \{h\colon \dom h\to A\} \subseteq R\cap S \ \text{such that} \ \forall h,j\in T \ . \ (s_h, s_j) = (r_h, r_j)
        \end{equation*}
        Pulling back along ~$k$, for each ~$h\in T$ we get
        \begin{itemize}
            \item a section ~$s_h^\prime \in F(\dom h\times_A B\to \dom h)(s_h)$,
            \item a section ~$r_h^\prime \in F(\dom h\times_A B\to \dom h)(r_h)$.
        \end{itemize}
        By functoriality and ~$s_h = r_h$ it holds that 
        \begin{equation*}
            F(\dom h\times_A B\to \dom h)(s_h) = F(\dom h\times_A B\to \dom h)(r_h) \Rightarrow s^\prime_h = r^\prime_h
        \end{equation*}
        so for
        \begin{equation*}
        k^*T = \{h^\prime\colon \dom_h\times_A B\to \dom h\}_{h\in k^*T}
        \end{equation*}
        we get
        \begin{equation*}
            \forall h, j \in k^*T \ . \ (s_h^\prime, s_j^\prime) = (r_h^\prime, r_j^\prime)
        \end{equation*}
        showing that ~$\MatchSh{F}(k)(x_S) = \MatchSh{F}(k)(x_R)$.
    \end{enumerate}
\end{proofE}
Having shown that the matching object construction yields a well-defined presheaf, we now define the map from matching families to global elements. 
This operator formalises how a compatible family is glued into a single section.

\begin{definition}[Amalgamation Operator]
\label{def:amalgamation}
    Given a site ~$(\CC, J)$, for any sheaf ~$F$ on ~$\Cat{C}$ and the presheaf ~$\MatchSh{F}$ as defined in Definition~\ref{def:MOpsh}, there exists a morphism of presheaves
    \begin{equation*}
        \amalgmap{F}\colon \MatchSh{F}\Rightarrow F
    \end{equation*}
    with for all ~$A\in\Cat{C}$ components
    \begin{equation*}
        \amalgmap{F}_A\colon \MatchSh{F}(A)\to F(A)
    \end{equation*}
    sending a representative
    \begin{equation*}
        [S = \{f\colon \dom f\to A\}\in J(A), \{(s_f, s_g)\}_{f, g\in S}]\in \MatchSh{F}(A)
    \end{equation*}
    to the unique amalgamation
    \begin{equation*}
        \amalgmap{F}_A([S, \{(s_f, s_g)\}_{f, g\in S}]) := s\in F(A)
    \end{equation*}
    such that 
    \begin{equation*}
        F(f)(A) = s_f \quad \text{for all} \ f\in S
    \end{equation*}
\end{definition}

\begin{propositionE}[][category=internal-logic]
For a site ~$(\Cat{C}, J)$ and a sheaf ~$F$ on $\Cat{C}$ the map ~$\amalgmap{F}$ as defined in Definition~\ref{def:amalgamation} is well-defined.
\end{propositionE}
\begin{proofE}
    For ~$A\in \Cat{C}$, Let ~$t\in \MatchSh{F}(A)$ be an equivalence class represented as
    \begin{equation*}
        t := [S = \{f\colon \dom f\to A\}\in J(A), \{(s_f, s_g) \in \MatchObject{F}{A}{\dom f}{\dom g}\}_{f, g\in S}]
    \end{equation*}
    The pairwise compatibility of all ~$(s_f, s_g)$ implies the collection ~$\{s_f\}_{f\in S}$ is a compatible family for the cover ~$S$.
    Because ~$F$ is a sheaf, there exists a unique amalgamation ~$s\in F(A)$ of ~$\{s_f\}_{f\in S}$.
    Suppose ~$t$ can also be represented as

    \begin{equation*}
        t := [S^\prime = \{h\colon \dom h\to A\}_{h\in S^\prime}\in J(A), \{(s_h^\prime, s_j^\prime) \in \MatchObject{F}{A}{\dom h}{\dom j}\}_{h,j\in S^\prime}]
    \end{equation*}
    yielding a matching family ~$\{s_h^\prime\}_{h\in S^\prime}$ with amalgamation ~$s^\prime \in F(A)$.
    If 
    \begin{equation*}
        [S, \{(s_f, s_g)\}_{f, g\in S}] \sim [S^\prime, \{(s_h^\prime, s_j^\prime)\}_{h, j\in S^\prime}]
    \end{equation*}
    there exists a common refinement ~$T \subseteq S^\prime\cap S$ such that the matching family ~$\{s_h^\prime\}_{h\in S^\prime}$ and ~$\{s_f\}_{f\in S}$ restrict to the same family over ~$T$.
    Then, by uniqueness of the amalgamation, we have that
    \begin{equation*}
        s = s^\prime \in F(A)
    \end{equation*}
    showing that ~$\amalgmap{F}_A$ is well defined for all $t\in\MatchSh{F}(A)$.
\end{proofE}
The amalgamation operator respects equivalence classes of matching families and yields a unique global element due to the sheaf condition. 
In fact, this map realises an isomorphism between the matching object presheaf and the sheaf itself.

\begin{theoremE}[][category=internal-logic]
\label{thm:amalg-iso}
    Given a site ~$(\CC, J)$, let ~$F$ be a sheaf on ~$\Cat{C}$. There is a natural isomorphism
    \begin{equation*}
        \amalgmap{F}\colon \MatchSh{F}\xrightarrow{\sim} F
    \end{equation*}
    where ~$\amalgmap{F}$ is defined in Definition~\ref{def:amalgamation}.
\end{theoremE}
\begin{proofE}
    We define an inverse
    \begin{equation*}
        \theta^F: F\Rightarrow \MatchSh{F}
    \end{equation*}
    with components
    \begin{equation*}
        \theta_A^F\colon F(A)\to \MatchSh{F}(A)
    \end{equation*}
    for every ~$A\in \Cat{C}$.
    Given ~$s\in F(A)$, we can choose any covering sieve ~$S = \{f\colon \dom f\to A\}\in J(A)$ and restrict ~$s$ along each ~$f$ to get
    \begin{equation*}
        \{s_f\in F(f)(s)\}_{f\in S}
    \end{equation*}

    Now, for any ~$f, g\in I$, we consider the pullback ~$\dom f\times_A \dom g$ with projections
    \begin{equation*}
        p_f\colon \dom f\times_A \dom g\to \dom f \quad p_g\colon \dom f\times_A \dom g\to \dom g
    \end{equation*}

    Because we have that
    \begin{equation*}
        F(p_f)(s_f) = F(p_f\circ \dom f)(s) = F(p_g\circ \dom g)(s) = F(p_g)(s_g)
    \end{equation*}
    we get
    \begin{equation*}
        \{(s_f, s_g) \in \MatchObject{F}{A}{\dom f}{\dom g}\}_{f, g\in S}
    \end{equation*}
    so we define
    \begin{equation*}
        \theta^F_A := [S, \{(s_f, s_g)\}_{f, g\in S}] \in \MatchSh{F}(A)
    \end{equation*}

    By construction ~$\theta^F_A$ maps ~$s$ to a matching family and ~$\amalgmap{F}$ maps this to the unique amalgamation, which must be ~$s$:
    \begin{equation*}
        \amalgmap{F}_A(\theta_A^F(s)) = s
    \end{equation*}

    Now, suppose ~$t = [S, \{(s_f, s_g)\}_{f, g\in S}]$ such that
    \begin{equation*}
        \theta_A^F(\amalgmap{F}_A(t)) = [S^\prime, \{(s_h^\prime, s_j^\prime)\}_{h, j\in S^\prime}]
    \end{equation*}

    Because both matching families ~$\{(s_f, s_g)\}_{f, g\in S}, \{(s_h^\prime, s_j^\prime)\}_{h, j\in S^\prime}$ have the same amalgamation, there must be a common refinement ~$T\subseteq S^\prime\cap S$ over which both families agree, meaning that
    \begin{equation*}
        [S, \{(s_f, s_g)\}_{f, g\in S}]\sim [S^\prime, \{(s_h^\prime, s_j^\prime)\}_{h, j\in S^\prime}]
    \end{equation*}
    and therefore that
    \begin{equation*}
        \theta^F_A(\amalgmap{F}_A(t)) = t
    \end{equation*}
    proving that ~$\MatchSh{F}\cong F$.

    The isomorphism is natural: Let ~$k\colon B\to A$ in $\Cat{C}$. Given ~$[S, \{(s_f, s_g)\}_{f, g\in S}] \in \MatchSh{F}(A)$ the unique amalgamation of the family ~$k^*S$ with sections ~$F(\dom f\times_A B\to \dom f)(s_f)$ is the same as restricting the amalgamation of ~$[S, \{(s_f, s_g)\}]$ along ~$f$:
    \begin{equation*}
        F(k)(\amalgmap{F}_A([S, \{(s_f, s_g)\}_{f, g\in S}])) = \amalgmap{F}_B(\MatchSh{F}(k)([S, \{(s_f, s_g)\}_{f, g}]))
    \end{equation*}
\end{proofE}

\begin{corollary}
    Let ~$F\colon \op{\Cat{C}}\to \Set$ be a presheaf on a site ~$(\Cat{C}, J)$. 
    
    \noindent The matching object presheaf ~$\MatchSh{F}$ as defined in Definition~\ref{def:MOpsh} is a sheaf if-and-only-if ~$F$ is a sheaf.
\end{corollary}

\begin{corollary}
    For a sheaf ~$F$ on a site ~$(\Cat{C}, J)$ and the matching object sheaf ~$\MatchSh{F}$ as defined in Definition~\ref{def:MOpsh}, we have that for all ~$A\in \Cat{C}$ we have ~$F(A)\in \Set_i$ if-and-only-if ~$\MatchSh{F}(A)\in \Set_i$.
\end{corollary}

Recall~\cref{ex:existential-separating} that for now we have only been able to interpret the separating conjunction over 

\noindent ~$\Pred_{F\DaySH F, 0}$, for a suitable site and a resource sheaf $F$.
We would like to obtain a connective that extends our internal logic~\cref{sec:internal-propositional-logic}, which amounts to interpretting the separating conjunction over ~$\Pred_{F, 0}\DaySH \Pred_{F, 0}$.
We use the left kan extension representation to obtain the map
\begin{equation}
    \label{eq:monoidal-fibre}
    \alpha\colon \Pred_{F, 0}\DaySH \Pred_{F,0} \to \Pred_{F\DaySH F, 0}
\end{equation}

\noindent as summarised in the diagram below with $P := \Pred_{F, 0}$ and ~$P_{\DaySH} := \Pred_{F\DaySH F, 0}$.
\begin{equation*}
    \begin{tikzcd}
        & \op{\mathcal{C}} 
          \arrow[dr, bend left=30, "P_{\DaySH}"{name=D, above}] 
          \arrow[dr, bend right=30, "P \DaySH P"{name=U, description}] \\
        \op{(\mathcal{C} \times \mathcal{C})} 
          \arrow[ur, "\op{(\DaySH)}"{name=L, above}] 
          \arrow[rr, swap, "P \boxtimes P"{name=R, below}] 
        & & \Set
        \arrow[Rightarrow, from=U, to=D, "\alpha", shorten <=1pt, shorten >=1pt] 
  \arrow[Rightarrow, from=R, to=L, "\lambda", shorten <=6pt, shorten >=6pt]
      \end{tikzcd}
\end{equation*}

The details are explained in~\cref{apdx:map}.

Now, we can put everything together.

\begin{theorem}
	\label{thm:separation-monoid}
	Let ~$(\CC, J)$ be a site, and let ~$p\colon \Pred \to \Univ$ be the internal predicate fibration in ~$\Int{\Sh{\Cat{C}}}$. Given a monoid ~$(F, \odot, \eta)$ for Day convolution on ~$\Sh{\Cat{C}}$ with ~$F$ a small sheaf, there is an induced map
	\begin{equation*}
		\exists^\odot_0\colon \Pred_{F, 0} \DaySH \Pred_{F, 0} \to \Pred_{\MatchSh{F}, 0}.
	\end{equation*}
	It is constructed as follows:
	\begin{enumerate}
		\item By Theorem~\ref{thm:amalg-iso}, the canonical map ~$F \to \MatchSh{F}$ is an isomorphism, yielding
		\begin{equation*}
			\odot' \colon F \DaySH F \to \MatchSh{F}.
		\end{equation*}
		\item Internalising the sheaves and the map gives
		\begin{equation*}
			\overline{F \DaySH F},\, \overline{\MatchSh{F}}\colon \fOb \to U_0, \quad \overline{\odot'}\colon \fOb \to U_1,
		\end{equation*}
		as in Theorem~\ref{thm:internalising-small-sheaves}.
		\item Then by Theorem~\ref{thm:internal-reindexing}, this induces
		\begin{equation*}
			\exists^{\odot'} \colon \Pred_{F \DaySH F} \to \Pred_{\MatchSh{F}},
		\end{equation*}
		and composing with the map from~\cref{eq:monoidal-fibre} yields ~$\exists^\odot_0$.
	\end{enumerate}
\end{theorem}

\subsection{Interpreting the Separating Conjunction}
\label{sec:intepreting-separating-conjunction}

Let ~$(\CC, J)$ be a Day-stable site. Given the small sheaf ~$F$ on ~$\CC$, let ~$(F, \odot, \mathrm{emp})$ be a resource monoid for Day convolution on ~$\Sh{\CC, J}$. Given the internal bifibration ~$p\colon \Pred\to\Univ$, the internal logic as given in~\cref{sec:internal-propositional-logic} can be extended with a separating conjunction connective as follows:

\begin{equation*}
	P, Q ::= \dots \mid P * Q,
\end{equation*}
with semantics defined pointwise: for ~$A \in C$,
\begin{align*}
	\llbracket P * Q \rrbracket_0(A) &= \llbracket P \rrbracket_0(A) \star \llbracket Q \rrbracket_0(A)
\end{align*}

\noindent where 
\begin{equation*}
	\star\colon \Pred_{F, 0}\DaySH \Pred_{F, 0}\to \Pred_{F, 0}
\end{equation*}

\noindent is defined as the composition ~$\star := \exists^{\amalgmap{M}}_0 \circ \exists^\odot_0$, given the map ~$\exists^\odot_0$ as defined in~\cref{thm:separation-monoid} and the map ~$\exists^{\amalgmap{F}}_0$ induces by internal existential quantification along the isomorphism ~$\amalgmap{F}\colon \MatchSh{F}\cong F$.

Now, we will give an example how two models of the separating conjunction can be intepreted in the same framework.

\begin{example}
    Let ~$M$ be the partial memory sheaf from~\cref{ex:partial-finitary-memory}, and suppose ~$M$ carries a resource monoid structure ~$(M, \odot, \mathrm{emp})$ for Day convolution, as in~\cref{ex:partial-memory-monoid}. 
    Given predicates ~$P_1, P_2 \in \Pred_{M, 0}(U)$ over a region ~$U \in \caL$, the separating conjunction ~$P_1 * P_2$ is defined abstractly as the composite
    \begin{equation*}
    P_1 * P_2 := \exists_\odot\big(\alpha(P_1 \DaySH P_2)\big),
    \end{equation*}
    where:
    \begin{itemize}
        \item $P_1 \DaySH P_2 \in (\Pred_{M,0} \DaySH \Pred_{M,0})(U)$ describes pairs of local predicates $(m_1, m_2)$ over some cover $U_1, U_2$ of $U$;
        \item $\alpha$ transports such pairs to a predicate over $M \DaySH M$, evaluated at the composed resource $m = \odot(m_1, m_2)$;
        \item $\exists_\odot$ reindexes this predicate along the monoid multiplication $\odot \colon M \DaySH M \to M$, yielding a predicate over $M(U)$.
    \end{itemize}
    
    Unfolding this definition, we obtain:
    \begin{equation*}
    \llbracket P_1 * P_2 \rrbracket_0(U) = \left\{ m \in M(U) \,\middle|\,
    \begin{aligned}
        &\exists\, U_1, U_2 \text{ with } U = U_1 \cup U_2,\ \exists\, m_1 \in P_1(U_1),\ m_2 \in P_2(U_2), \\
        &\text{such that } \odot(m_1, m_2) = m
    \end{aligned}
    \right\}.
    \end{equation*}
    
    If we take $\odot$ to be defined only when $m_1$ and $m_2$ agree on overlap (weak conjunction), this yields:
    \begin{equation*}
    \llbracket P_1 * P_2 \rrbracket_0(U) = \left\{ m \in M(U) \,\middle|\,
    \begin{aligned}
        &\exists\, U_1, U_2 \text{ with } U = U_1 \cup U_2,\ \exists\, m_1 \in P_1(U_1),\ m_2 \in P_2(U_2), \\
        & m_1|_{U_1 \cap U_2} = m_2|_{U_1 \cap U_2},\ \text{and } m = m_1 \cup m_2
    \end{aligned}
    \right\},
    \end{equation*}
    corresponding to partial map semantics.
    
    If instead we take $\odot$ to be defined only when $U_1 \cap U_2 = \emptyset$ (strong conjunction), we obtain:
    \begin{equation*}
    \llbracket P_1 * P_2 \rrbracket_0(U) = \left\{ m \in M(U) \,\middle|\,
    \begin{aligned}
        &\exists\, U_1, U_2 \text{ with } U = U_1 \uplus U_2,\ \exists\, m_1 \in P_1(U_1),\ m_2 \in P_2(U_2), \\
        &\text{and } m = m_1 \cup m_2
    \end{aligned}
    \right\},
    \end{equation*}
    which captures the semantics of disjoint heap separation.
    \end{example}
    
Both forms of separating conjunction—weak and strong—arise from the same resource sheaf $M$ over the same base ~$\caL$. 
The difference lies in the choice of monoid structure $\odot$ on $M$, showing that the connective's semantics is parametrised by how resources are allowed to compose, not by changing the underlying logical universe.

\subsection{Model Specific Connectives}
We conclude by illustrating how various model specific atom are handled in our framework. 

\begin{example}
\label{ex:strict-pointsto}
In the internal logic of~$\Sh{\caL}$, the \emph{strict} points-to formula ~$t_1 \mapsto t_2$ is interpreted as a predicate over the memory sheaf ~$M$, i.e., as a global section
\begin{equation*}
\llbracket t_1 \mapsto t_2 \rrbracket \colon \TSheaf \to \Pred_{M, 0}
\end{equation*}
with components at stage ~$U \in \caL$ defined by
\begin{equation*}
\llbracket t_1 \mapsto t_2 \rrbracket_U(*) := 
\left\{ \sigma \in M(U) \ \middle|\ t_1 \in U \Rightarrow \sigma(t_1) = t_2 \right\}.
\end{equation*}
This encodes the standard semantics of the \emph{strict} points-to connectives: 
if the address~\mbox{$t_1$} is in scope, its value must match ~\mbox{$t_2$} and, otherwise, nothing is required.
\end{example}

\begin{example}
\label{ex:nonstrict-pointsto}
Let ~$x \in \Loc$ and ~$v \in \Val$. 
In the internal logic of~$\Sh{\caL}$, the \emph{non-strict} points-to predicate is interpreted as a morphism
\begin{equation*}
\llbracket x \hookrightarrow v \rrbracket\colon \TSheaf \to \Pred_{M', 0}
\end{equation*}
for the sheaf~$M'(U) := [U, \Val + \bot]$. Its component at stage~$U$ is defined by
	\begin{equation*}
	\llbracket x \hookrightarrow v \rrbracket_U(*) := 
	\left\{ \sigma \in M'(U) \ \middle| \ 
	x \in U \Rightarrow \left( x \in \mathrm{dom}(\sigma) \wedge \sigma(x) = v \right)
	\right\}.
	\end{equation*}
	This captures the \emph{non-strict} semantics: outside the domain, no constraint is imposed.
\end{example}

\noindent

%% file: probability-sheaves.tex
\section{A probabilistic separation logic}
\label{sec:probability}
We are going to demonstrate how a variant of probabilistic separation logic can be obtained as an instance of our framework.
In particular, we recover the model of separation described by Li~et~al.~\cite{LiEA24} in which separation is construed as products of sample spaces and resource combination is governed by probabilistic independence between random variables. 
We review the necessary details in the following.

\todo[inline]{Make the semantics precise if time permits.}
\begin{example}
\label{ex:PSL-model1}
The \emph{formulas} of probabilistic separation logic (PSL) are inductively generated by the following grammar:
\begin{equation*}
 P, Q ::= \top \mid 
 	     X \sim \mu \mid 
	      P * Q,
 \end{equation*}
Formulas of PSL are interpreted over probability spaces on a non-empty countable sample space $\bbS$.
That is,~$\bbS$ is equipped with a pair~$(\F, \mu)$ consisting of a $\sigma$-algebra~$\F$ 
and a distribution~$\mu\colon\F\to[0,1]$.
A \emph{random variable} on~$\bbS$ is a map~$\bbS\to\Z$.
The formula $X \sim \mu$ holds if $X$ is $\F$-measurable and its law agrees with~$\mu$.
The formula $P_1 * P_2$ holds on~$(\bbS, \mu)$ 
if there exist spaces $(\bbS_1, \mu_1)$ and $(\bbS_2, \mu_2)$ and a surjection $p\colon \bbS \to \bbS_1 \times \bbS_2$ such that~$(\bbS_i, \mu_1)$ satisfies~$P_i$ and~$\mu$ is the pullback of $\mu_1 \otimes \mu_2$ along~$p$.
\end{example}

We proceed to illustrate how PSL, 
slightly generalised to additionally include the connectives of~\cref{sec:internal-propositional-logic},
can be recovered as an instance of our framework.
In this direction, we first give a description of a suitable site.

\paragraph*{Ambient setting.}
We write~$\Surj$ for the small category of countable sets and surjections, and equip it with the atomic coverage~$J^{\at}$.
That is, a sieve~$(f\colon A_F\to A)$ is considered covering if and only if it is non-empty.
The resulting site $(\Surj, J^{\mathrm{at}})$ carries the structure of a symmetric monoidal category with the tensor being given by the cartesian product.  
Moreover, the corresponding sheaf category $\Sh{\Surj, J^{\at}}$ is Day-stable~\cite{LiEA24}.

\paragraph*{Resources and predicates}
Next, we define a resource sheaf that enables indexing of predicates over probability space.
Let~$\bbP(\bbS)$ denote the set of all probability spaces on sample space~$\bbS$.
That is,~$\bbP(\bbS)$ consists of pairs~$(\F, \mu)$ of a $\sigma$-algebra~$\F$ on~$\bbS$ and a probability distribution~$\mu$ for~$\F$.
The assignment~$\bbS\mapsto\bbP(\bbS)$ is the object-part of a sheaf
$\mathbb{P}\colon \op{\Surj} \to \Set$
with the action on a surjective map~$f\colon\bbS\to\bbT$ defined by the map
$\bbP(f)\colon\bbP(\bbT)\to\bbP(\bbS)$	
which sends~$\mathcal{P}\in\bbP(\bbT)$ to its \emph{pullback probability space}~$\bbP(f)(\mathcal{P})$ on~$\bbS$.

A pullback probability space along a surjection ~$f\colon \bbS^\prime\to\bbS$ maps a probability space ~$(\mathcal{G}, \nu)$ on $\bbS$ to the space
\begin{equation*}
    (\{f^{-1}(G) \mid G\in \mathcal{G}\}, \mu) \quad \text{on} \ \bbS^\prime
\end{equation*}

\noindent where $\mu(f^{-1}(G)) := \nu(G)$ for all $G\in \mathcal{G}$.

Recall that~$\Pred_{\mathbb{P}, 0}$ is defined for~$\bbS\in\Surj$ via the pullback in~\cref{def:internal-fibre}. 
This means that  
\begin{equation*}
 \Pred_{\bbP, 0}(\bbS) := \left\{ \alpha \colon \mathbb{P}_{\bbS} := \mathbb{P} \circ \op{\dom_{\bbS}} \to \Omega \right\}.
\end{equation*}
Since maps into the subobject classifier correspond to a unique subobject, this defines the collection of subobjects of the presheaf~$\bbP_{\bbS} \colon (\Surj/\bbS)^{op} \to \Set$.
In detail, the functor~$\bbP_{\bbS}$ assigns the set~$\bbP(\bbT)$ of probability spaces on~$\bbT$ to each surjection~$p\colon \bbT \to \bbS$ 
A predicate over~$\bbP$ is thus a family of subsets of probability spaces on all such refinements~$\bbT \to \bbS$, subject to naturality under pullback.
This enforces the extension invariant: propositions that are true in one sample space should remain true when pulled back along a surjective refinement.

Intuitively, suppose that we assert a property such as~$X \sim \mu$ over a probability space on~$\bbS$. 
Then we must be able to pull back that space along any surjection~$\bbT \to \bbS$ and still make sense of 
whether~$X \sim \mu$ holds: \emph{predicates must be stable under surjective refinement of the sample space}.

\paragraph*{Resources and predicates}
Next, we define the resource sheaf that will index probabilistic predicates and describe what predicates over this sheaf represent.

\textit{Probability Sheaf~$\mathbb{P}$} Define a presheaf (in fact, a sheaf)
\begin{equation*}
    \mathbb{P}\colon \op{\Surj} \to \Set
\end{equation*}
by setting, for each sample space~$\bbS$,
\begin{equation*}
    \mathbb{P}(\bbS) := \left\{ \text{probability spaces on the underlying set } \bbS \right\}.
\end{equation*}
In other words, an element of~$\mathbb{P}(\bbS)$ is a probability measure~$\mathcal{P}$ on~$\bbS$, equipped with an appropriate~$\sigma$-algebra. We use the symbol for probability measure for probability space interchangably.
For a morphism~$f\colon \bbS' \twoheadrightarrow \bbS$, the map~$\mathbb{P}(f)$ sends a probability space~$\mathcal{P} \in \mathbb{P}(\bbS)$ to the pullback probability space~$f^{-1}\mathcal{P}$ on~$\bbS'$.

We now recall that~$\Pred_{\mathbb{P}, 0}$ is explicitly defined, for~$\bbS \in \Surj$, via the pullback in~\cref{def:internal-fibre} as
\begin{equation*}
    \Pred_{\mathbb{P}, 0}(\bbS) := \left\{ \alpha \colon \mathbb{P}_{\bbS} := \mathbb{P} \circ \op{\dom_\bbS} \to \Omega \right\}.
\end{equation*}
Since each such map into the subobject classifier corresponds uniquely to a subobject, this defines the collection of subobjects of the presheaf~$\mathbb{P}_{\bbS} \colon (\Surj/\bbS)^{op} \to \Set$.

Concretely, the functor~$\mathbb{P}_{\bbS}$ assigns to each surjection~$p\colon \bbS' \to \bbS$ the set~$\mathbb{P}(\bbS')$ of probability spaces on~$\bbS'$. 
A predicate over~$\mathbb{P}$ is thus a family of subsets of probability spaces on all such refinements~$\bbS' \to \bbS$, subject to naturality under pullback.

More intuitively, suppose we assert a property such as~$X \sim \mu$ over a probability space on~$\bbS$. Then we must be able to pull back that space along any surjection~$\bbS' \to \bbS$ and still make sense of whether~$X \sim \mu$ holds. That is, predicates must be stable under surjective refinement of the sample space.

\paragraph*{Abstract semantic operation.}
The sheaf~$\bbP$ carries the structure of a monoid for Day convolution on the site~$(\Surj, J^{\at})$. 
Concretely, the convolution~$\mathbb{P} \DaySH \mathbb{P}$ is isomorphic to a sheaf~$\mathbb{P}^2_\perp$, where
\begin{equation*}
\big((\bbS_1, \mathcal{P}_1), (\bbS_2, \mathcal{P}_2)\big) \in \bbP^2_\perp(\bbS)
\end{equation*}
consists of spaces~$\mathcal{P}_1 \in \mathbb{P}(\bbS_1)$ and~$\mathcal{P}_2 \in \mathbb{P}(\bbS_2)$, pulled back along a surjection~$p\colon \bbS \twoheadrightarrow \bbS_1 \times \bbS_2$.

The monoid multiplication~$\odot \colon \mathbb{P}^2_\perp \to \mathbb{P}$ maps such a pair to the pullback of the product measure~$\mathcal{P}_1 \otimes \mathcal{P}_2$ along~$p$, combining distributions through a shared refinement~$\bbS$.
We internalise~$\mathbb{P} \DaySH \mathbb{P}$ and~$\mathbb{P}$ as global sections~$\overline{\mathbb{P} \DaySH \mathbb{P}}, \overline{\mathbb{P}} \in \Univ_0$ of the internal universe, and the external multiplication~$\odot$ becomes an internal morphism~$\overline{\odot} \colon \TSheaf \to \Univ_1$ by~\cref{thm:internalising-small-sheaves}.

This induces an internal existential quantification~$\exists^\odot$ along~$\overline{\odot}$ via opcartesian lifting in the internal bifibration

\noindent $p\colon \Pred \to U$, as guaranteed by~\cref{thm:internal-reindexing}.
For~$P \in \Pred_{\mathbb{P} \DaySH \mathbb{P}, 0}(\bbS)$ and~$p\colon \bbS' \twoheadrightarrow \bbS$,
\begin{equation*}
\exists^\odot(P)(p) := \left\{ \mathcal{P} \in \mathbb{P}(\bbS') \;\middle|\; 
  \exists (\mathcal{P}_1, \mathcal{P}_2) \in P(p)\ \text{with}\ \mathcal{P} = \odot(\mathcal{P}_1, \mathcal{P}_2)
\right\}.
\end{equation*}
This expresses the universal property of~$\exists^\odot$: a probability space~$\mathcal{P}$ lies in~$\exists^\odot(P)$ if it arises as the image of a pair~$(\mathcal{P}_1, \mathcal{P}_2)$ under convolution, where the pair satisfies~$P$.
That is,~$\exists^\odot(P)$ holds on a distribution if it can be decomposed into independent components that satisfy~$P$ and then recombined, yielding an internal semantics of probabilistic separating conjunction.
\paragraph*{Logical connective.}
To interpret the separating conjunction within the internal logic, we compose:
\begin{equation*}
* := \Pred_{\mathbb{P}, 0} \DaySH \Pred_{\mathbb{P}, 0} \to \Pred_{\mathbb{P} \DaySH \mathbb{P}, 0} \xrightarrow{\exists^\odot} \Pred_{\mathbb{P}, 0},
\end{equation*}
where the first map is defined via Day convolution on predicates, as in Section~\ref{sec:internallogic}. This defines the separating conjunction as a derived logical connective in the internal logic.


%% file: conclusion.tex
\section{Concluding and Future Work}
\label{sec:conclusion}
We have developed elements of sheafeology, that is, categorical logic internal in sheaf categories.
The locality and compatibility axioms of sheaves enable reasoning about the combination and decomposition of resources in a modular way.
Different flavours of resources, and different treatments thereof, can be expressed within the same unifying framework.

In~\cref{sec:internalcategoriessheaves}, we generalised fibrational logic to a 2-categorical setting, recovering the predicate fibration of~\cite{Jacobs99} internally in a sheaf topos.
In~\cref{sec:internallogic}, we combined this internal logic with a generalised separating conjunction, allowing us to interpret various models of separation logic within the same categorical infrastructure.
In~\cref{sec:probability}, we recovered a known instance of probabilistic separation logic and extended it to accommodate intuitionistic connectives.
Together, these results define an expressive, compositional, and internally-defined framework for resource-aware logics.

\paragraph*{Future Work}
The internal universe is expressive: we have seen how ordinary sets can be made resource-aware via sheaf semantics.
The unified view on logics sets the stage for any potential research directions.

A promising direction is to build internal models of computational effects, following classical categorical structures such as monads and Kleisli categories, but interpreted internally in the sheaf topos.
For instance, the power object provides a natural setting for internalising state-indexed monads, as in the work of \textcite{MM15}.
Internally-defined Kleisli morphisms may then serve as a basis for reasoning about concurrent or effectful processes, via sheaf structures over state spaces.

A particularly interesting line of investigation is to study how such internal Kleisli structures interact with the internal predicate fibration.
This could lead to an extension of our framework with resource-aware modalities, enabling a comparison with probabilistic modalities as proposed by Li et al.~\cite{LiEA23}, or ownership modalities as presented by Jung et al.~\cite{JungEA18}.
We would like to explore this direction to test the expressive boundaries of our approach.

At present, the frame rule has not been treated. However, inspired by the work of Aguirre and Katsumata~\cite{AK20}, we anticipate that internal sheaves of predicate transformers provide semantics for the frame rule across different flavours of separation logic.
In a similar line, categorical logic leads fairly directly to proof systems, something that we wish to explore in the future to possibly recover existing proof systems~\cite{WickersonEA13}

Finally, the current framework models locality or spatiality through a single covering structure at a time. In future work, we aim to explore whether 2-sheaf categories could allow for modelling orthogonal resources simultaneously.
This can enable treating concurrent processes via sheaf structure on states~\cite{Goguen92:SheafSemanticsConcurrent, Lilius93:SheafSemanticsPetri, Malcolm09:SheavesObjectsDistributed, MP86:SheafTheoreticModelConcurrency, Pratt86:ModelingConcurrencyPartial, SS09, TO15, WG92:SheafSemanticsFOOPS} and combine this with resources shared by processes.


%% file: notation.tex
\section{Notation}
\label{sec:notation}

\begin{tabular}{l|l}
  Notation & Meaning \\\hline
  $\img{f}$ & image of $f$ \\
  $\preimg{f}$ & preimage of $f$ \\
  $\Set$ & Category of small sets and maps (same as $\Set[0]$) \\
  $\Set[k]$ & kth category of set and maps \\
  $\CC, \DC, \EC, \dots$ & General categories \\
  $\ArrC{\CC}$ & Arrow category on $\CC$ \\
  $\slice{\CC}{A}$ & Slice category of $\CC$ over $A$ \\
  $\Univ$ & Universe sheaf \\
  $\PSh{\CC}$ & Category of presheaves on a category  (same as $\PSh[1]{\CC}$) \\
  $\PSh[k]{\CC}$ & Category of presheaves with values in $\Set[k]$ \\
  $\Yo$ & Yoneda embedding $\CC \to \PSh{\CC}$ \\
  $\Sh{\CC, J}$ & Category of sheaves on a site (same as $\Sh[1]{\CC, J}$) \\
  $\Sh[k]{\CC, J}$ & Category of sheaves with values in $\Set[k]$ \\
  $\timesU$ & Unit for $\times$/terminal object in $\Set$ \\
  $\TSheaf$ & Terminal sheaf \\
  $\slice{\CC}{A}$ & Slice category over $A$ \\
  $\sliceDom[A]$ & domain functor $\slice{\CC}{A} \to \CC$ \\
  $\SliceCoverage[J]{A}$ & Induced coverage on a slice category \\
  $\SliceSite[\CC][J]{A}$ & Induced site on a slice category \\
  $\pbCover{h}G$ & pullback of a cover/sieve $G$ \\
  $\sheafInt{F}$ & internalisation of small sheaf $F$ as global section of universe \\
  $\sheafInt{\alpha}$ & internalisation of morphism of small sheaves
\end{tabular}


%% file: fibration.tex
\section{Internal Fibrations}
Before we start the proof, we introduce known concepts that provide tools to make the proofs more manageable.
\begin{definition}[Discrete Internal Category]
	For a base category ~$\EC$, every object $X \in \EC$ determines a discrete internal category in $\EC$, denoted $\DiscCat{X}$, whose object of objects is $X$, and all structure maps are identities on $X$.
	
	This construction extends to a functor
	\begin{equation*}
	\DiscCat{-} \colon \EC \to \Int{\EC}
	\end{equation*}
	from $\EC$ to the category of internal categories in $\EC$.
\end{definition}

\begin{definition}[Externalisation]
	\label{def:ext}
	For a base category ~$\EC$, given $C$ an internal category $\mathcal{C}$ in $\mathrm{Cat}(\mathcal{E})$, the externalisation \[P_C: \underline{C}\rightarrow \mathcal{E}\] is a split fibration over $\mathcal{E}$, where $\underline{C}$ is the category with
	\begin{itemize}
		\item objects: pairs $(I, X: I\rightarrow C_0)$ with $I$ in $\mathcal{E}$
		\item morphisms: pairs $(u: I\rightarrow J, X: I\rightarrow C_0)$, with $u, f$ morphisms in $\mathcal{E}$, such that the diagram below commutes
		\[\begin{tikzcd}
			& I & J \\
			{C_0} & {C_1} & {C_0}
			\arrow["u", from=1-2, to=1-3]
			\arrow["X"', from=1-2, to=2-1]
			\arrow["f", from=1-2, to=2-2]
			\arrow["Y", from=1-3, to=2-3]
			\arrow["s", from=2-2, to=2-1]
			\arrow["t"', from=2-2, to=2-3]
		\end{tikzcd}\]
		given source and target maps $s, t: C_1\rightarrow C_0 $ respectively.
	\end{itemize}
	Identity morphisms $\mathrm{id}_{I, X}$ are defined as $\mathrm{id}_{I, X} := (\mathrm{id}_I, i_C\circ X)$, where $i_C$ is the identity assigning morphism of the internal category $C$. \newline
	Given the morphisms $(u:I\rightarrow J, f: I\rightarrow C_1)$ and $(v: J\rightarrow K, g: I\rightarrow C_1)$ in $\underline{C}$, the composition $(v,g)\circ (u,f)$ is defined as \[(v\circ u, c\circ \langle f, g\circ u\rangle)\] given the composition morphism $c: C_1\times_{C_0} C_1\rightarrow C_1$.\newline
	Together this gives us the following unit laws and associative law.
	\begin{itemize}
		\item Given the morphism $(u: I\rightarrow J, f: I\rightarrow C_1)$ between $(I, A)$ and $(J, B)$ in $\underline{C}$ we have the following unit laws.
		\[ (u, c\circ  \langle f, i_C\circ B\circ u\rangle) = (u, f) = (u, c\circ \langle i_C\circ A, f\rangle)\]
		\item Given the morphisms $(u: I\rightarrow J, f: I\rightarrow C_1), (v: J\rightarrow K, g: J\rightarrow C_1)$, and\newline $(w: K\rightarrow M , h: K\rightarrow C_1)$ between $(I, A)$, $(J, B)$, $(K, C)$, $(M, D)$ in $\underline{C}$ we have the following associative law.
		\[((w\circ v)\circ u, c\circ \left\langle f, c\circ \langle g, h\circ v\rangle\circ u\right\rangle) = (w\circ (v\circ u), c\circ \left\langle c\circ \langle f, g\circ u\rangle, h\circ v\right\rangle)\]
	\end{itemize}
	We write $f: X\circlearrow Y$ for the vertical morphism $(u: I\rightarrow J, f: I\rightarrow C_1)$.
\end{definition}

\begin{proposition}
	For a base category ~$\EC$, there is an equivalence of categories \[\mathrm{Cat}(\mathcal{E}) \simeq \mathrm{SpFib}(\mathcal{E})\] with $\underline{(-)}: \mathrm{SpFib(\mathcal{E})}\rightarrow \mathrm{Cat}(\mathcal{E})$ and $i: \mathrm{SpFib}(\mathcal{E})\rightarrow \mathrm{Cat}(\mathcal{E})$.
\end{proposition}

\begin{definition}[The 2-Category $\mathrm{SpFib}(\mathcal{E})$]
	For a base category ~$\EC$, the category of split fibrations over $\mathcal{E}$, $\mathrm{SpFib}(\mathcal{E})$ has as
	\begin{itemize}
		\item $0$-cells: internalisations/ split fibrations $P_C: C\rightarrow \mathcal{E}$
		\item $1$-cells: A functor between split fibrations $P_C$ and $P_D$ is given by the split fibred functor $\underline{F}: \underline{C}\rightarrow \underline{D}$, defined on
		\begin{itemize}
			\item objects $(I, X:I\rightarrow C_0)$ in $\underline{C}$ as $\underline{F}(I, X) = (I, F_0X)$
			\item morphisms $(u: I\rightarrow J, f: I\rightarrow C_1)$ as $\underline{F}(u, f) = (u, F_1f)$
		\end{itemize}
		given the internal functor $F: C\rightarrow D $ in $\mathrm{Cat}(\mathcal{E})$.
		\item $2$-cells: a natural transformation $\underline{\alpha}: \underline{F}\Rightarrow \underline{G}$ between split fibred functors $\underline{F},\underline{G}: \underline{C}\rightarrow \underline{D}$, such that $P_D\underline{\alpha} = P_C$. \newline
		For all $(I, X:I\rightarrow C_0)$ in $\underline{C}$, the $2$-cell $\underline{\alpha}$ has components $\underline{\alpha}_X$, corresponding to morphisms 
		
		\noindent $(\mathrm{id}_I, \alpha\circ X): F_0X\circlearrow G_0X$ in $\underline{D}$, given the internal natural transformation $\alpha: C_0\rightarrow D_1 $.
		The naturality condition is expressed as the commuting diagram below for any\newline $(u:I\rightarrow J, f: I\rightarrow C_1)$ in $\underline{C}$.
		\[\begin{tikzcd}
			{\underline{F}(I,X)} & {\underline{G}(I, X)} \\
			{\underline{F}(J, Y)} & {\underline{G}(J, Y)}
			\arrow["{\underline{\alpha}_X}", from=1-1, to=1-2]
			\arrow["{\underline{F}(u, f)}"', from=1-1, to=2-1]
			\arrow["{\underline{G}(u, f)}", from=1-2, to=2-2]
			\arrow["{\underline{\alpha}_Y}"', from=2-1, to=2-2]
		\end{tikzcd}\]
		Unwrapped, the naturality condition expresses the equality
		\[\underline{G}f\circ \underline{\alpha}_X = (u, c\circ \langle \alpha\circ X, G_1f\rangle ) = (u, c\circ \langle F_1f, \alpha\circ Y\circ u\rangle) = \underline{\alpha}_Y\circ F_1f\] 
	\end{itemize}
\end{definition}

The equivalence of categories $\mathrm{Cat}(\mathcal{E})\simeq \mathrm{SpFib}(\mathcal{E})$ allows us to work with internal functors and internal natural transformations externally. For example, proving naturality for an externalised natural transformation $\underline{\alpha}$ proves $\alpha$ is a valid internal natural transformation. Proving functoriality for an externalised functor $\underline{F}$ proves that $F$ is a valid internal functor.

\begin{proposition}
	\label{prop:internal-reindexing}
	The construction in Definition~\ref{thm:internal-reindexing} defines an internal functor $u^*: E_X\rightarrow E_Y$, and this functor is uniquely determined by the cartesian property of the map $\bar{u}_0$.
\end{proposition}
\begin{proof}
	It remains to show that
	\begin{enumerate}
		\item the universal property of the pullback $E_{Y, 0}$ induces the unique map $u^*_0$. \newline
		We have to show that $Y\circ \pi_{X, 2} = p_0\circ s_E\circ \bar{u}_o$.
		\begin{align*}
			Y\circ \pi_{X, 2} = Y \tag{$\pi_{X, 2}$ is a terminal map}\\
			= s_B\circ \beta \tag{component of internal natural transformation}\\
			= s_B\circ p\bar{u}_o \tag{$\bar{u}_o$ is $p$-cartesian}\\
			= s_B\circ p_1\circ \bar{u}_o \tag{internal whiskering}\\
			= p_0\circ s_E\circ \bar{u}_o \tag{$p$ is an internal functor}
		\end{align*}
		
		\item the $p$-cartesian lift $\bar{u}_o$ and the 2-cells $\xi$ and $\gamma$ induce the unique 2-cell $\bar{u}_m$.
		
		We have to define the 2-cells $\gamma$ and $\xi$ and the 1-cells $x$ and $\tilde{u}_m$, such that they satisfy the cartesian lifting property. \newline
		\begin{itemize}
			\item $x := I\circ \mathrm{disc}(t_X): \mathrm{disc}(E_{X, 1})\rightarrow E_X$
			\item $\tilde{u}_m := \tilde{u}_o\circ I\circ \mathrm{disc}(s_X): \mathrm{disc}(E_{X, 1})\rightarrow E$
			\item $\xi := c_E\circ \langle \bar{u}_os_X, \iota_X\rangle: \mathrm{disc}(E_{X, 1})_0 = E_{X, 1}\rightarrow E_1$
			\item $\gamma := i_B\circ Y :\mathrm{disc}(E_{X, 1})_0 = E_{X, 1}\rightarrow B_1$
		\end{itemize}
		
		Given source and target morphisms $s_X: E_{X, 1}\rightarrow E_{X, 1}$ and $t_X: E_{X, 1}\rightarrow E_{X, 0}$ respectively.
		Given a discrete internal category $\mathrm{disc}(X_0)$, $I$ is the inclusion functor to the original internal category $X$, with $I_0:= \mathrm{id}$ and $I_1 := i_X$.\newline
		The map $\langle \bar{u}_os_X, \iota_X\rangle$ is uniquely defined in the pullback below.
		
		Both $x$ and $\tilde{u}_m$ are valid 1-cells in $\mathrm{Cat}(\mathcal{E})$ as they are composites of internal functors.
		
		To check that $\xi$ is a valid 2-cell, we only have to verify the source and target correspond with $\tilde{u}_m$ and $\iota_X\circ x$ respectively, as the domain of $\xi$ is a discrete internal category, so the internal naturality condition holds trivially. \newline
		
		The source of $\xi$ corresponds with $\tilde{u}_{m, 0}$:
		\begin{align*}
			s_E\circ \xi = s_E\circ c_E\langle \bar{u}_os_X,  \iota_X\rangle \tag{unwrapping $\xi$}\\
			= s_E\circ \pi_1\circ \langle\bar{u}_os_X, \iota_X\rangle \tag{source of composite}\\
			= s_E\circ \bar{u}_os_X \tag{pullback}\\
			= \tilde{u}_{o, 0} \circ s_X \tag{source of internal natural transformation}
		\end{align*}
		The other way:
		\begin{align*}
			\tilde{u}_{m, 0} = \tilde{u}_{o, 0}\circ I_0\circ \mathrm{disc}(s_X)_0 \tag{unwrapping $\tilde{u}_{m, 0}$}\\
			= \tilde{u}_{o, 0}\circ s_X \tag{$I_0 = \mathrm{id}$ and $\mathrm{disc}(s_X)_0 = s_X$}
		\end{align*}
		The target of $\xi$ corresponds with $\pi_{X, 1}\circ x_0$ (the object component of $\iota_X\circ x)$):
		\begin{align*}
			t_E \circ \xi = t_E\circ c_E \circ \langle \bar{u}_os_X, \iota_X\rangle \tag{unwrapping $\xi$}\\
			= t_E \circ \pi_2\circ \langle \bar{u}_os_X, \iota_X \rangle \tag{target of composite}\\
			= t_E\circ \iota_X \tag{pullback}\\
			= \pi_{X, 1}\circ t_X \tag{functoriality of $\iota_X$}
		\end{align*}
		The other way:
		\begin{align*}
			\pi_{X, 1}\circ x_0 = \pi_{X, 1}\circ I_0\circ \mathrm{disc}(t_X)_0 \tag{unwrapping $x_0$}\\
			= \pi_{X, 1}\circ t_X \tag{$I_0 = \mathrm{id}$ and $\mathrm{disc}(t_X)_0 = t_X$}
		\end{align*}
		
		The definition of $\gamma$ yields a valid 2-cell from $p\circ \tilde{u}_m$ to $p\circ \tilde{u}_o$ by construction.
		
		To give intuition why the maps are defined this way, we look at how reindexing functors are defined on morhpisms externally. Given a morphism $f: P\rightarrow Q$ above $X$ in $\underline{E_X}$, the morphism $u: Y\rightarrow X$ induces cartesian lifts $\bar{u}(P): u^*P\rightarrow P$ and $\bar{u}(Q): u^*Q\rightarrow Q$ with $u^*P$ and $u^*Q$ above $Y$. We can compose $f$ with $\bar{u}(P)$ to get a morphism above $u$. The lifting property of $\bar{u}(Q)$ induces the unique morphism $u^*f: u^*P\rightarrow u^*Q$ above $\mathrm{id}_Y$. This is summarised in the diagram below.
		\[\begin{tikzcd}
			{E:} & {u^*P} & P \\
			& {u^*Q} & Q \\
			{B:} & Y & X
			\arrow["p", from=1-1, to=3-1]
			\arrow["{\bar{u}(P)}", from=1-2, to=1-3]
			\arrow["{u^*f}"', dashed, from=1-2, to=2-2]
			\arrow["f", from=1-3, to=2-3]
			\arrow["{\bar{u}(Q)}"', from=2-2, to=2-3]
			\arrow["u"', from=3-2, to=3-3]
		\end{tikzcd}\]
		
		Internally, this construction is given as 1-cells that assign to every $f$ in $E_{X, 1}$ a morphism $\bar{u}(s_X(f))$ and a morphism $\bar{u}(t_X(f))$ in $E_1$ together with a 2-cell that assigns to every $f$ in $E_{X, 1}$ the composition $f\circ \bar{u}(s_X(f))$.
		
		Because 2-cells in $\mathrm{Cat}(\mathcal{E})$ go from the object of objects of an internal category to the object of morphisms of an internal category, we use the discrete category $\mathrm{Cat}(\mathcal{E})$ to obtain a 2-cell with the object of morphisms of the fibred category as its domain, as $\mathrm{disc}(E_{X,1})_0 = E_{X, 1}$. \newline \newline 
		Finally, to induce $\bar{u}_m$ we show the equality $p\xi = p\bar{u}_ox\circ \gamma$. We externalise and show the equality for all $ (I, A)$ in $\underline{\mathrm{disc}(E_{X, 1})}$.
		\begin{align*}
			\underline{p\bar{u}_ox\circ\gamma}_A = (\mathrm{id}_I, c_B\circ \langle\gamma\circ A, p\bar{u}_ox\circ A\rangle) \tag{definition vertical comp.}\\
			= (\mathrm{id}_I, c_B\circ \langle\gamma, p\bar{u}_ox A\rangle\circ A) \tag{pullback}\\
			= (\mathrm{id}_I, c_B\circ \langle i_B\circ Y, p\bar{u}_ox\rangle\circ A) \tag{unwrapping $\gamma$}\\
			= (\mathrm{id}_I, c_B\circ \langle i_B\circ Y, ux\rangle\circ A) \tag{cartesianness of $\bar{u}_o$}\\
			= (\mathrm{id}_I, c_B\circ \langle i_B\circ Y, u\circ I_0\circ \mathrm{disc}(t_X)_0\rangle\circ A) \tag{unwrapping $x$}\\
			= (\mathrm{id}_I, c_B\circ \langle i_B\circ Y, u\circ t_X\rangle\circ A) \tag{$I_0 = \mathrm{id}$ and $\mathrm{disc}(t_X)_0 = t_X$}
		\end{align*}
		The other way:
		\begin{align*}
			\underline{p\xi}_A = (\mathrm{id}_I, p\xi\circ A) = (\mathrm{id}_I, p_1\circ c_E\circ \langle \bar{u}_os, \iota_X\rangle\circ A) \tag{unwrapping $\xi$}\\
			= (\mathrm{id}_I, c_B\circ \langle p\bar{u}_os, p_1\circ\iota_X\rangle\circ A) \tag{$p$ respects composition}\\
			= (\mathrm{id}_I, c_B\circ \langle u\circ s_X, p_1\circ\iota_X\rangle\circ A) \tag{cartesianness of $\bar{u}_o$}\\
			= (\mathrm{id}_I, c_B\circ \langle u\circ s_X, i_B\circ X\rangle\circ A) \tag{pullback of $E_{X, 1}$}\\
		\end{align*}
		That $c_E\circ \langle u\circ s_X, i_B\circ X\rangle \circ A = c_E\circ \langle i_B\circ Y, u\circ t_X\rangle\circ A$ holds by the unit law of composition and because $A\circ s_X$ and $A\circ t_X$ map $I$ to the same fibre.
		\item To induce the unique arrow $u^*_1: E_{X, 1}\rightarrow E_{Y,1}$ we show the outer arrows in the pullback below commute.
		\begin{align*}
			p_1\circ \bar{u}_m = p\bar{u}_m \tag{internal whiskering}\\
			= \gamma \tag{unique lifting property}\\
			= i_B\circ Y \tag{unwrapping $\gamma$}\\
			= i_B\circ Y\circ \pi_{X, 2} \tag{$\pi_{X, 2}$ is terminal}
		\end{align*}
		\item Finally, we show that $u^*: E_X\rightarrow E_Y$, with the object component $u^*_0$ and the morphism component $u_1^*$ is an internal functor by showing the diagrams below commute.
		\[\begin{tikzcd}
			{E_{X, 1}} && {E_{Y, 1}} && {E_1} \\
			{E_{X, 0}} && {E_{Y, 0}} && {E_0}
			\arrow["{u^*_1}", from=1-1, to=1-3]
			\arrow["{t_X}", shift left=3, from=1-1, to=2-1]
			\arrow["{s_X}"', shift right=3, from=1-1, to=2-1]
			\arrow["{\iota_Y}", from=1-3, to=1-5]
			\arrow["{t_Y}", shift left=3, from=1-3, to=2-3]
			\arrow["{s_X}"', shift right=3, from=1-3, to=2-3]
			\arrow["{t_E}", shift left=3, from=1-5, to=2-5]
			\arrow["{s_E}"', shift right=3, from=1-5, to=2-5]
			\arrow["{i_X}"{description}, from=2-1, to=1-1]
			\arrow["{u^*_0}"', from=2-1, to=2-3]
			\arrow["{i_Y}"{description}, from=2-3, to=1-3]
			\arrow["{\pi_{Y, 1}}"', from=2-3, to=2-5]
			\arrow["{i_E}"{description}, from=2-5, to=1-5]
		\end{tikzcd}\]
		\[\begin{tikzcd}
			{E_{X, 1}\times_{E_{X, 0}}E_{X, 1}} && {E_{Y, 1}\times_{E_{Y, 0}}E_{Y, 1}} && {E_1\times_{E_0}E_1} \\
			{E_{X, 1}} && {E_{Y, 1}} && {E_1}
			\arrow["{\langle u^*_1\circ \pi_1, u^*_1\circ \pi_2\rangle}", from=1-1, to=1-3]
			\arrow["{c_X}"', from=1-1, to=2-1]
			\arrow["{\langle \iota_Y\circ \pi_1, \iota_Y\circ \pi_2\rangle}", from=1-3, to=1-5]
			\arrow["{c_Y}", from=1-3, to=2-3]
			\arrow["{c_E}", from=1-5, to=2-5]
			\arrow["{u^*_1}"', from=2-1, to=2-3]
			\arrow["{\iota_Y}"', from=2-3, to=2-5]
		\end{tikzcd}\]
		We have that the left inner square commutes if the outer square commutes, because the inclusion maps are monomorphisms and together they are functorial in $\mathrm{Cat}(\mathcal{E})$.
		\begin{itemize}
			\item respect for the source map:
			\begin{align*}
				s_E\circ \iota_Y\circ u^*_1 = s_E\circ \bar{u}_m \tag{UMP of $u^*_1$}\\
				= \tilde{u}_{m, 0} \tag{source of internal natural transformation}\\
				= \tilde{u}_{o, 0}\circ I_0\circ \mathrm{disc}(s_X) \tag{unwrapping $\tilde{u}_m$}\\
				= \tilde{u}_{o, 0}\circ s_X \tag{$I_0 = \mathrm{id}$ and $\mathrm{disc}(s_X)_0 = s_X$}\\
				= s_E\circ \bar{u}_o\circ s_X \tag{source of internal natural transformation}\\
				= \pi_{Y, 1}\circ u^*_0\circ s_X \tag{UMP of $u^*_0$}
			\end{align*}
			\item respect for the target map:
			\begin{align*}
				t_E\circ \iota_Y\circ u^*_1 = t_E\circ \bar{u}_m \tag{UMP of $u^*_1$}\\
				= \tilde{u}_{o, 0}\circ x_0 \tag{target of internal natural transformation}\\
				= s_E\circ \bar{u}_o\circ x_0 \tag{source of internal natural transformation}\\
				= \pi_{Y,1}\circ u^*_0\circ x_0 \tag{UMP of $u^*_0$}\\
				= \pi_{Y,1}\circ u^*_0\circ I_0\circ \mathrm{disc}(t_X) \tag{unwrapping $x_0$}\\
				= \pi_{Y, 1}\circ u^*_0\circ t_X \tag{$I_0 = \mathrm{id}$ and $\mathrm{disc}(t_X)_0 = t_X$}
			\end{align*}
			\item respect for identities:
			\begin{align*}
				i_E\circ \pi_{Y, 1}\circ u^*_0 = i_E\circ s_E\circ \bar{u}_0 \tag{UMP of $u^*_0$}\\
				= i_E\circ \tilde{u}_{o, 0} \tag{source of internal natural transformation}\\
				= \tilde{u}_{o, 1}\circ i_X \tag{$\tilde{u}_o$ is an internal functor}\\
				= \bar{u}_m\circ i_X \tag{$\bar{u}_m = \tilde{u}_{o, 1}$ on identities}\\
				= \iota_Y\circ u^*_1\circ i_X \tag{UMP of $u^*_1$}
			\end{align*}
			That $\bar{u}_m\circ i_X = \tilde{u}_m \circ i_X$ holds because $\tilde{u}_m$ is defined the same as $\bar{u}_m$ on identity morphisms. This refers to an earlier not that $\tilde{u}_{o, 1}$ can be used to define the reindexing on morphisms if $p $ is a discrete internal fibration.
			\item respect for composition:\newline
			The map $u^*_1$ is determined by the unique 2-cell $\bar{u}_m: \mathrm{disc}(E_{X, 1})_0\rightarrow E_1$, induced by the $p$-cartesian lift $\bar{u}_o$. 
			For two morphisms in $E_{X, 1}$, the composite of the two components of $\bar{u}_m$ at these morphisms in $E_1$ satisfies the same uniqueness properties as the component of $\bar{u}_m$ at the composite of these morphisms, showing that $u^*_1$ preserves composition.
			\begin{align*}
				c_E\circ \langle \iota_Y\circ \pi_1, \iota_Y\circ \pi_2\rangle \circ \langle u^*_1\circ\pi_1, u^*_1\circ\pi_2\rangle \\
				= c_E\circ \langle \iota_Y\circ u^*_1\circ \pi_1, \iota_Y\circ u^*_1\circ \pi_2\rangle \tag{pullback}\\
				= c_E\circ \langle \bar{u}_m\circ\pi_1, \bar{u}_m\circ \pi_2\rangle \tag{UMP of $u^*_1$}\\
				= \bar{u}_m\circ c_X \tag{uniqueness $\bar{u}_m$}\\
				= \iota_Y\circ u^*_1\circ c_X \tag{pullback}
			\end{align*}
			To show that $c_E\circ \langle \bar{u}_m\circ \pi_1, \bar{u}_m\circ \pi_2\rangle = \bar{u}_m\circ c_X$, we construct the 2-cells
				\begin{equation}\label{dg:compcart}
					\begin{tikzcd}[column sep=small]
						{\mathrm{disc}(E_{X, 1}\times _{E_{X, 0}}E_{X, 1})} && E && {\mathrm{disc}(E_{X, 1}\times _{E_{X, 0}}E_{X, 1})} && E \\
						& {E_X} &&& {E_X} & E & B
						\arrow[""{name=0, anchor=center, inner sep=0}, "{\tilde{u}_c}", from=1-1, to=1-3]
						\arrow["{x_c}"', from=1-1, to=2-2]
						\arrow[""{name=1, anchor=center, inner sep=0}, "{\tilde{u}_c}", from=1-5, to=1-7]
						\arrow["{x_c}"', from=1-5, to=2-5]
						\arrow["p", from=1-7, to=2-7]
						\arrow["{\iota_X}"', from=2-2, to=1-3]
						\arrow["{\tilde{u}_o}"', from=2-5, to=2-6]
						\arrow["p"', from=2-6, to=2-7]
						\arrow["{\xi_c}"', shorten <=5pt, shorten >=3pt, Rightarrow, from=0, to=2-2]
						\arrow["{\gamma_c}"', shorten <=5pt, shorten >=3pt, Rightarrow, from=1, to=2-6]
					\end{tikzcd}
				\end{equation}
			such that
			\[p\xi_c = p\tilde{u}_ox_c\circ \gamma_c\]
			Because $\bar{u}_o$ is cartesian, 2-cell $\zeta_c: \mathrm{disc}(E_{X, 1}\times_{E_{X, 0}E_{X, 1}})_0\rightarrow E_1$ from $\tilde{u}_c$ to $\tilde{u}_o\circ x_c$ is induced, unqiue such that $\xi_c = \bar{u}_ox_c\circ \zeta_c$ and $p\zeta_c = \gamma_c$.\newline
			Now we can show equality of $c_E\circ \langle \bar{u}_m\circ \pi_1, \bar{u}_m\circ \pi_2\rangle$ and $\bar{u}_m\circ c_X$ by showing that 
			\[c_E\circ \left\langle\bar{u}_ox_c, c_E\circ \langle \bar{u}_m\circ \pi_1, \bar{u}_m\circ \pi_2\rangle\right\rangle = \xi_c = c_E\circ \langle\bar{u}_ox_c,\bar{u}_m\circ c_X\rangle \]
			and
			\[p(c_E\circ \langle \bar{u}_m\circ \pi_1, \bar{u}_m\circ \pi_2\rangle) = \gamma_c = p\bar{u}_m\circ c_X\]
			Because the composition of internal natural transformations yields an internal natural transformation, the 2-cells $c_E\circ \langle \bar{u}_m\circ \pi_1, \bar{u}_m\circ \pi_2\rangle$ and $\bar{u}_m\circ c_X$ are valid 2-cells with the type
			\[\mathrm{disc}(E_{X, 1}\times_{E_{X, 0}} E_{X_1})_0 = E_{X, 1}\times_{E_{X, 0}} E_{X_1}\rightarrow E_1\]
			We have the following definitions for the 1-cells and 2-cells in diagrams~\ref{dg:compcart}.
			\begin{itemize}
				\item $\gamma_c := i_B\circ Y$
				\item $\xi_c := c_E\circ \langle \bar{u}_o\circ s_X\circ c_X, \iota_X\circ c_X\rangle$
				\item $x_c := I\circ \mathrm{disc}(t_X)\circ \mathrm{disc}(c_X)$
				\item $\tilde{u}_c := \tilde{u}_o\circ I\circ \mathrm{disc}(s_X)\circ \mathrm{disc}(c_X)$
			\end{itemize}
			Note that these 1-cells and 2-cells are the 1-cells and 2-cells defined in the diagrams~\ref{dg:cart} below composed with $\mathrm{disc}(c_X)$ and $c_X$ respectively, showing that $x_c$ and $\tilde{u}_c$ are valid internal functors, $\xi_c$ and $\gamma_c$ are valid internal natural transformations, and the equality
			\[p\xi_c = p\bar{u}_ox_c\circ \gamma_c\]
			holds.
			\begin{equation}\label{dg:cart}
				\begin{tikzcd}
				{\mathrm{disc}(E_{X, 1})} && E & {\mathrm{disc}(E_{X, 1})} && E \\
				& {E_X} && {E_X} & E & B
				\arrow[""{name=0, anchor=center, inner sep=0}, "{\tilde{u}_m}", from=1-1, to=1-3]
				\arrow["x"', from=1-1, to=2-2]
				\arrow[""{name=1, anchor=center, inner sep=0}, "{\tilde{u}_m}", from=1-4, to=1-6]
				\arrow["x"', from=1-4, to=2-4]
				\arrow["p", from=1-6, to=2-6]
				\arrow["{\iota_X}"', from=2-2, to=1-3]
				\arrow["{\tilde{u}_o}"', from=2-4, to=2-5]
				\arrow["p"', from=2-5, to=2-6]
				\arrow["\xi", shorten <=5pt, shorten >=3pt, Rightarrow, from=0, to=2-2]
				\arrow["\gamma", shorten <=5pt, shorten >=3pt, Rightarrow, from=1, to=2-5]
			\end{tikzcd}
    \end{equation}

    We now check that the maps $c_E\circ \langle \bar{u}_m\circ \pi_1, \bar{u}_m\circ \pi_2\rangle$ and $\bar{u}_m\circ c_X$ satisfy the unique lifting properties.\newline
			We externalise and check for all components $(I, A: I\rightarrow E_{X, 1}\times_{E_{X, 0}}E_{X, 1})$ in $\underline{\mathrm{disc}(E_{X, 1}\times_{E_{X, 0}E_{X, 1}})}$.
			\begin{align*}
				\underline{\bar{u}_ox_c}_A\circ \underline{(\bar{u}_m\circ c_X)}_A\\ 
				= (\mathrm{id}_I, c_E\circ \langle \bar{u}_m\circ c_X, \bar{u}_ox_c\rangle\circ A) \tag{def. vertical comp.}\\
				= (\mathrm{id}_I, c_E\circ \langle \bar{u}_m\circ c_X, \bar{u}_o\circ I_0\circ \mathrm{disc}(t_X)_0\circ \mathrm{disc}(c_X)_0\rangle\circ A) \tag{unwrapping $x_c$}\\
				= (\mathrm{id}_I, c_E\circ \langle \bar{u}_m\circ c_X, \bar{u}_o\circ I_0\circ t_X\circ c_X\rangle\circ A) \tag{discreteness}\\
				= (\mathrm{id}_I, c_E\circ \langle \bar{u}_m\circ c_X, \bar{u}_o\circ t_X\circ c_X\rangle\circ A) \tag{$I_0 = \mathrm{id}$}\\
				= (\mathrm{id}_I, c_E\circ \langle \bar{u}_o\circ s_X\circ c_X, \iota_X\circ c_X\rangle\circ A) \tag{uniqueness of $\bar{u}_m$}\\
				= \underline{\xi_c}_A
			\end{align*}
			and
			\begin{align*}
				\underline{\bar{u}_ox_c}_A\circ \underline{(c_E\circ \langle \bar{u}_m\circ \pi_1, \bar{u}_m\circ \pi_2\rangle)}_A \\
				= (\mathrm{id}_I, c_E\circ \left\langle c_E \circ \langle \bar{u}_m\circ \pi_1, \bar{u}_m\circ \pi_2\rangle, \bar{u}_ox_c\right\rangle\circ A) \tag{def. vertical comp.}\\
				= (\mathrm{id}_I, c_E\circ \left\langle c_E \circ \langle \bar{u}_m\circ \pi_1, \bar{u}_m\circ \pi_2\rangle, \bar{u}_o\circ t_X\circ c_X\right\rangle\circ A) \tag{unwrapping $x_c$}\\
				= (\mathrm{id}_I, c_E\circ \left\langle \bar{u}_m\circ \pi_1, c_E \circ \langle\bar{u}_m\circ \pi_2, \bar{u}_o\circ t_X\circ c_X\rangle\right\rangle\circ A) \tag{comp. is associative}\\
				= (\mathrm{id}_I, c_E\circ \left\langle \bar{u}_m\circ \pi_1, c_E \circ \langle\bar{u}_m\circ \pi_2, \bar{u}_o\circ t_X\circ \pi_2\rangle\right\rangle\circ A) \tag{taget of comp.}\\
				= (\mathrm{id}_I, c_E\circ \left\langle \bar{u}_m\circ \pi_1, c_E \circ \langle\bar{u}_o\circ s_X\circ \pi_2, \iota_X\circ \pi_2\rangle\right\rangle\circ A) \tag{unique property of $\bar{u}_m$}\\
				= (\mathrm{id}_I, c_E\circ \left\langle c_E \circ \langle \bar{u}_m\circ \pi_1, \bar{u}_o\circ s_X\circ \pi_2\rangle, \iota_X\circ \pi_2\right\rangle\circ A) \tag{comp. is associative}\\
				= (\mathrm{id}_I, c_E\circ \left\langle c_E \circ \langle \bar{u}_m\circ \pi_1, \bar{u}_o\circ t_X\circ \pi_1\rangle, \iota_X\circ \pi_2\right\rangle\circ A) \tag{pullback of comp.}\\
				= (\mathrm{id}_I, c_E\circ \left\langle c_E \circ \langle \bar{u}_o\circ s_X \circ \pi_1, \iota_X\circ \pi_1\rangle, \iota_X\circ \pi_2\right\rangle\circ A) \tag{unique property of $\bar{u}_m$}\\
				= (\mathrm{id}_I, c_E\circ \left\langle \bar{u}_o\circ s_X \circ \pi_1, c_E \circ \langle \iota_X\circ \pi_1, \iota_X\circ \pi_2\rangle\right\rangle\circ A) \tag{comp. is associative} \\
				= (\mathrm{id}_I, c_E\circ \left\langle \bar{u}_o\circ s_X \circ c_X, c_E \circ \langle \iota_X\circ \pi_1, \iota_X\circ \pi_2\rangle\right\rangle\circ A) \tag{source of comp.}\\
				= (\mathrm{id}_I, c_E\circ \langle \bar{u}_o\circ s_X \circ c_X, \iota_X\circ c_X\rangle\circ A) \tag{functoriality of $\iota_X$}\\
				= \underline{\xi_c}_A
			\end{align*}
			We have that $c_E\circ \langle \bar{u}_m\circ \pi_1, \bar{u}_m\circ \pi_2\rangle$ and $\bar{u}_m\circ c_X$ are both above $\mathrm{id}_Y = \gamma_c$ via $p$ because the composition maps respect the fibre stucture. \newline
			Because $c_E\circ \langle \bar{u}_m\circ \pi_1, \bar{u}_m\circ \pi_2\rangle$ and $\bar{u}_m\circ c_X$ are both unique with regards to the same properties, they must be equal.
		\end{itemize}
	\end{enumerate}
\end{proof}

We briefly review the necessary background on (2-)fibrations.
We show that the internal functor ~$p \colon \Pred \to \Univ$, constructed in~\cref{eq:internal-predicate-functor}, satisfies this definition and forms a fibration in the 2-category ~$\Cat{\Sh{\CC, J}}$ in the following sense.
\begin{definition}
\label{def:2fibration}
A \emph{fibration} in a 2-category~$\Cat{K}$ is a 1-cell~$p\colon\IntCat{E}\to\IntCat{B}$ such that for every 2-cell
\begin{equation*}
\beta\colon(\IntCat{X}\xra{b}\IntCat{B})\Rightarrow (\IntCat{X}\xra{e}\IntCat{E}\xra{p}\IntCat{B})
\end{equation*}
there is a \emph{$p$-cartesian lift}~$\alpha\colon(\IntCat{X}\xra{e'}\IntCat{E})\Rightarrow(\IntCat{X}\xra{e}\IntCat{E}).$
	
	A 2-cell $\alpha$ is $p$-cartesian when given the following data:
	\begin{itemize}
		\item all 1-cells ~$x\colon \IntCat{Y}\to \IntCat{X}$,
		\item a 1-cell ~$e^{\prime\prime}\colon \IntCat{Y}\to \IntCat{E}$,
		\item all 2-cells ~$\xi\colon e^{\prime\prime}\Rightarrow e\circ x$,
		\item and all 2-cells ~$\gamma\colon p\circ e^{\prime\prime}\Rightarrow p\circ e^\prime \circ x$
	\end{itemize}

	\noindent such that ~$p\xi = p\alpha x\circ \gamma$ (where left-and- right whiskering is given by juxtaposition), there exists a unique 2-cell ~$\colon e^{\prime\prime}\Rightarrow e^\prime$ such that
	\begin{equation*}
		p\zeta = \gamma \quad \xi = \alpha x\circ \zeta.
	\end{equation*}
	
	\noindent This condition can be seen as a lifting property of 2-cells. If one views 1-cells as objects and 2-cells as morphisms, the definition recovers the usual lifting condition of a Grothendieck fibration in an ordinary 1-category.
\end{definition}

\begin{definition}
	Given a 2-category ~$\mathcal{K}$, the 2-category ~$\co{\mathcal{K}}$ is defined as follows: 
	it has the same objects and 1-cells as ~$\mathcal{K}$, but the direction of all 2-cells is reversed. 
	That is, for any 1-cells ~$f, g\colon A \to B$, a 2-cell ~$\alpha\colon f \Rightarrow g$ in ~$\mathcal{K}^{\mathrm{co}}$ corresponds to a 2-cell ~$\alpha\colon g \Rightarrow f$ in ~$\mathcal{K}$.
	Vertical and horizontal composition of 2-cells is defined accordingly.
\end{definition}

\begin{definition}
Maintain the notation as in~\cref{def:2fibration}.
We denote 2-cells going the reverse direction with ~$\co{(-)}$.
Let ~$p\colon \IntCat{E}\to \IntCat{B}$ be a 1-cell in a 2-category.
The 1-cell $p$ is an opfibration if for all 2-cells ~$\co{\beta}\colon p\circ e\Rightarrow b$ has a $p$-cartesian 2-cell ~$\co{\alpha}\colon e\Rightarrow e^\prime$ above $\co{\beta}$.
A 2-cell ~$\co{\alpha}$  is $p$-opcartesian if for all 2-cells ~$\co{\xi}\colon e\circ x\Rightarrow e^{\prime\prime}$ and ~$\co{\gamma}\colon p\circ e^\prime\circ x\Rightarrow p\circ e^{\prime\prime}$ such that ~$p\co{\xi} = \co{\gamma}\circ p\co{\alpha}x$, there exists a unique 2-cell ~$\co{\zeta}\colon e' \Rightarrow d$ such that ~$\co{\xi} = \co{\zeta}\circ \co{\alpha}x$ and ~$p \co{\zeta} = \co{\gamma}$.
\end{definition}
	
The proof that $p$ is an opfibration proceeds by dualising the construction of cartesian 2-cells given in the previous theorem. 
In the fibration case, we are given a morphism ~$\delta\colon F\to G$ and a subobject ~$Q \subseteq G$, and construct the pullback ~$P \coloneqq \gamma^* Q \subseteq F$. 
This ensures that the triple ~$((F, \phi), (G, \psi), \delta) \in \mathsf{Pred}_1(A)$, with \( \phi \coloneqq \psi \circ \delta \), satisfies the condition for being a morphism in the internal category ~$\Pred$.
	
In the opfibration case, the construction is reversed: given a morphism ~$\delta \colon F \to G$ and a subobject ~$P \subseteq F$, we define ~$Q \subseteq G$ to be the smallest subobject such that ~$\gamma(P) \subseteq Q$, i.e., the \emph{image} of ~$ P$ along ~$\delta $. 
The subobject $ Q $ is then classified by a map ~$\psi \colon G \to \Omega$, uniquely determined by the requirement that ~$\phi = \psi \circ \delta$, where ~$\phi \colon F \to \Omega$ classifies $P$. 
This yields a triple ~$((F, \phi), (G, \psi), \delta) \in \Pred_1(A) $, now constructed via forward image rather than pullback.
	
This construction recalls the usual definition of opcartesian morphisms via left adjoints to pullback. 
Since Grothendieck toposes admit stable image factorisations, meaning for any morphism we can always define the smallest subobject through which it factors, we conclude that $p$ defines a bifibration in ~$\Cat{\Sh{\CC, J}}$.

\begin{lemma}
	Let $\mathcal{K}$ be a 2-category, and let ~$p\colon \IntCat{E}\to\IntCat{B}$ be a 1-cell in $\mathcal{K}$. A 2-cell ~$\alpha\colon f\Rightarrow g$ in $\mathcal{K}$ is opcartesian for $p$ if and only if $\alpha$ is cartesian for $p$ when regarded as a 2-cell in the 2-category $\co{\mathcal{K}}$.
\end{lemma}
\begin{proof}
	By definition, the 2-category ~$\co{\mathcal{K}}$ has the same objects and 1-cells as $\mathcal{K}$, but the direction of all 2-cells is reversed.
	
	Recall that a 2-cell ~$\alpha$  is $p$-opcartesian if for all 2-cells ~$\xi\colon e\circ x\Rightarrow e^{\prime\prime}$ and ~$\gamma\colon p\circ e^\prime\circ x\Rightarrow p\circ e^{\prime\prime}$ such that ~$p\xi = \gamma\circ p\alpha x$, there exists a unique 2-cell ~$\zeta\colon e' \Rightarrow d$ such that ~$\xi = \zeta\circ \alpha x$ and ~$p \zeta = \gamma$.
	
	But this is precisely the definition for a cartesian 2-cell in $\co{\mathcal{K}}$, showing the universal lifting property for opcartesian 2-cells in $\mathcal{K}$ is exactly the universal lifting property for cartesian 2-cells in $\co{\mathcal{K}}$, completing the proof.
\end{proof}

\begin{corollary}
	For a base category ~$\EC$ and an opfibration ~$p\colon \IntCat{E}\to \IntCat{B}$ in ~$\Int{\Cat{E}}$, for all 2-cells ~$f\colon X\Rightarrow Y$, given the fibres ~$\IntCat{E}_X, \IntCat{E}_Y$, there exists an internal functor
	\begin{equation*}
		\exists^f\colon \IntCat{E}_X\rightarrow \IntCat{E}_Y
	\end{equation*}
	in ~$\Int{\Cat{E}}$ by following the procedure in~\cref{prop:internal-reindexing}, but in ~$\co{\Int{\Cat{E}}}$ instead of in ~$\Int{\Cat{E}}$.
\end{corollary}

%% file: lan.tex
\section{The Fibred Day Convolution Map \texorpdfstring{$\Pred_{F, 0}\DaySH \Pred_{F, 0} \to \Pred_{F \DaySH F, 0}$}{Pred DaySH Pred to Pred}}
\label{apdx:map}

Let $(\CC, J)$ be a Day-stable site with monoidal product $\cdot$, and $F$ a resource sheaf. We start with the diagram below
\begin{equation*}
\begin{tikzcd}
        & \op{\mathcal{C}} 
          \arrow[dr, bend left=30, "P_{\DaySH}"{name=D, above}] 
          \arrow[dr, bend right=30, "P \DaySH P"{name=U, description}] \\
        \op{(\mathcal{C} \times \mathcal{C})} 
          \arrow[ur, "\op{{\DaySH}}"{name=L, above}] 
          \arrow[rr, swap, "P \boxtimes P"{name=R, below}] 
        & & \Set
        \arrow[Rightarrow, from=U, to=D, "\alpha", shorten <=1pt, shorten >=1pt] 
  \arrow[Rightarrow, from=R, to=L, "\lambda", shorten <=6pt, shorten >=6pt]
      \end{tikzcd}
\end{equation*}

where
\begin{itemize}
	\item $\boxtimes$ is the external product,
	\item $P$ is defined as $\Pred_{F, 0}$,
	\item $P_\otimes$ is defined as $\Pred_{F\DaySH F, 0}$,
	\item $\lambda$ is given by Day convolution $\mathrm{Lan}_{\cdot}(F\boxtimes F)\simeq F\DaySH F$.
\end{itemize}

Given ~$\lambda$, we want to produce the map $\alpha\colon P\DaySH P\to P_\otimes$ using the UMP of the left kan extension.

For this, it suffices to produce a map $\beta\colon P\boxtimes P\to P_\otimes \circ \op{\cdot}$. 

We express $P\boxtimes P$ and $P_\otimes \circ \op{\cdot}$ diagramatically:

\begin{equation*}\begin{tikzcd}
	&& {\PSh{\CC}} \\
	{\PSh{\CC}\times \PSh{\CC}} &&&& {\PSh{\CC}\times \PSh{\CC}} \\
	&&&& {\PSh{\CC}} \\
	{\PSh{\slice{\CC}{A}}\times\PSh{\slice{\CC}{B}}} &&&& {\PSh{\slice{\CC}{A\cdot B}}} \\
	{\PSh{\slice{\CC}{A}}\times\PSh{\slice{\CC}{B}}} &&&& {\PSh{\slice{\CC}{A\cdot B}}} \\
	&& {\PSh{\slice{\CC}{A}\times\slice{\CC}{B}}}
	\arrow["\Delta"', from=1-3, to=2-1]
	\arrow["\Delta", from=1-3, to=2-5]
	\arrow["{\PSh{\dom_A}\times \PSh{\dom_B}}"', from=2-1, to=4-1]
	\arrow["\DaySH", from=2-5, to=3-5]
	\arrow["\PSh{\dom_{A\cdot b}}", from=3-5, to=4-5]
	\arrow["{\Omega^-_{\slice{\CC}{A}}\times \Omega^-_{\slice{\CC}{B}}}"', from=4-1, to=5-1]
	\arrow["{\Omega^-_{A\cdot B}}", from=4-5, to=5-5]
	\arrow["{\boxtimes_{\slice{\CC}{A}, \slice{\CC}{B}}}"', from=5-1, to=6-3]
	\arrow["{\PSh{\gamma_{A,B}}}", from=5-5, to=6-3]
\end{tikzcd}
\end{equation*}

where
\begin{itemize}
    \item $\Delta$ is the diagonal functor,
    \item $\dom_{-}$ is the domain functor ~$\slice{\CC}{-}\to \CC$,
    \item $\mathbf{PSh}(-) \colon \op{\Cats}\to \op{\Cats}$ is the endofunctor sending $\CC \mapsto \PSh{\CC}$ and $F \mapsto (-)\circ \op{F}$,
    \item $\boxtimes\colon \mathbf{PSh}\times \mathbf{PSh}\to \mathbf{PSh}\circ (-\times -)$ is the external product composed with the internal product in the domain:
    \begin{equation*}
        \boxtimes_{\CC, \DC}\colon \PSh{\CC}\times \PSh{\DC}\to \PSh{\CC\times \DC}
    \end{equation*},
    \item $\gamma_{A, B}\colon \slice{\CC}{A}\times \slice{\CC}{B}\to \slice{\CC}{A\cdot B}$ is given by assumption that ~$\slice{\CC}{-}$ is lax monoidal~\cref{def:day-stable},
    \item ~$\Omega^-_{\CC}\colon \op{\PSh{\CC}}\to \PSh{\CC}$ is the functor mapping a presheaf $F$ to the hom-set $[F, \Omega]\in \PSh{\CC}$.
\end{itemize}

Now, we fill in the diagram as follows:

\begin{tikzcd}
    &                                                               & \PSh{\CC} \arrow[rdd, "\Delta"] \arrow[ld, "\Delta", swap]                         &                                                           \\
    & \PSh{\CC}\times \PSh{\CC} \arrow[ld, "\DaySH", swap] \arrow[rdd, "\Omega^-_{\CC}\DaySH\Omega^-_{\CC}", swap] \arrow[rrdd, "\PSh{\dom_A}\times \PSh{\dom_B}", sloped] &                                                          &                                                           \\
\PSh{\CC} \arrow[d, "\PSh{\dom_{A\cdot B}}", swap] \arrow[rd, "\Omega^-_{\CC}", swap]          &                                                               &                                                          & \PSh{\CC}\times \PSh{\CC} \arrow[d, "\PSh{\dom_A}\times \PSh{\dom_B}"]                       \\
\PSh{\slice{\CC}{A\cdot B}} \arrow[d, ""{name = D, above}, "\Omega^-_{\slice{\CC}{A\cdot B}}", swap]   & \PSh{\CC} \arrow[ld, "\PSh{\dom_{A\cdot B}}"] \arrow[rd, "\PSh{\cdot}", swap]                               & \PSh{\CC}\times \PSh{\CC} \arrow[d, ""{name = U, below}, "\boxtimes_{\CC, \CC}"] \arrow[l, "\DaySH"] \arrow[rd, "{\scriptscriptstyle \PSh{\dom_A}\times \PSh{\dom_B}}",  pos=0.5, sloped, xshift=6pt] & \PSh{\slice{\CC/A}}\times \PSh{\slice{\CC}{B}} \arrow[d, "\Omega^-_{\slice{\CC}{A}}\times \Omega^-_{\slice{\CC}{B}}"]  \\
\PSh{\slice{\CC}{A\cdot B}} \arrow[rrd, "\PSh{\gamma_{A, B}}", swap] &                                                               & \PSh{\CC\times \CC} \arrow[d, "\PSh{\dom_A\times \dom_B}"{description}]                            & \PSh{\slice{\CC}{A}}\times \PSh{\slice{\CC}{B}} \arrow[ld, "\boxtimes_{\slice{\CC}{A}, \slice{\CC}{B}}"] \\
    &                                                               & \PSh{\slice{\CC}{A}\times \slice{\CC}{B}}                &                                                          
\arrow[Rightarrow, from=U, to= D, shorten <=7pt, shorten >=157pt, "\boldsymbol{\lambda}", swap, pos=0.07, draw = blue]
\end{tikzcd}

The diagram below trivially commutes.
\begin{equation*}
\begin{tikzcd}
    \PSh{\CC}\times \PSh{\CC} \arrow[r, "{\scriptscriptstyle\PSh{\dom_A}\times \PSh{\dom_B}}"] & \PSh{\slice{\CC}{A}}\times \PSh{\slice{\CC}{B}}  \\
    \PSh{\CC} \arrow[u, "\Delta"] \arrow[r, "\Delta"']                     & \PSh{\CC} \arrow[u, "{\scriptscriptstyle\PSh{\dom_A}\times \PSh{\dom_B}}"']
    \end{tikzcd}
\end{equation*}

The diagram below commutes by naturality of $\boxtimes$.
\begin{equation*}
    \begin{tikzcd}
        \PSh{\CC\times \CC} \arrow[r, "{\scriptscriptstyle\PSh{\dom_A}\times \PSh{\dom_B}}"]                                                     & \PSh{\slice{\CC}{A}\times \slice{\CC}{B}}                                                                 \\
        \PSh{\CC}\times \PSh{\CC} \arrow[r, "{\scriptscriptstyle\PSh{\dom_A}\times \PSh{\dom_B}}"'] \arrow[u, "{\boxtimes_{\CC,\CC}}"] & \PSh{\slice{\CC}{A}}\times \PSh{\slice{\CC}{B}} \arrow[u, "{\boxtimes_{\slice{\CC}{A},\slice{\CC}{B}}}"']
    \end{tikzcd}
\end{equation*}

The diagram below commutes because $\slice{\CC}{-}$ is lax monoidal.
\begin{equation*}
\begin{tikzcd}
    \PSh{\slice{\CC}{A\cdot B}}  \arrow[r, "\PSh{\gamma_{A, B}}"] & \PSh{\slice{\CC}{A}\times \slice{\CC}{B}}\\
    \PSh{\CC}  \arrow[u, "\PSh{\dom_{A\cdot B}}"] \arrow[r, "\PSh{\cdot}"'] & \PSh{\CC\times \CC} \arrow[u, "\PSh{\dom_A\times\dom_B}"']
\end{tikzcd}
\end{equation*}

The squares below are more tricky
\begin{equation*}
    \begin{tikzcd}
    \PSh{\CC} \arrow[r, "\Omega^-_{\CC}"] \arrow[d, "\PSh{\dom_{A\cdot B}}"'] 
        & \PSh{\CC} \arrow[d, "\PSh{\dom_{A\cdot B}}"] \\
    \PSh{\slice{\CC}{A\cdot B}} \arrow[r, "\Omega^-_{\slice{\CC}{A\cdot B}}"']
        & \PSh{\slice{\CC}{A\cdot B}} 
    \end{tikzcd}
\end{equation*}

\begin{equation*}
    \begin{tikzcd}
    \PSh{\CC}\times \PSh{\CC} \arrow[r, "\Omega^-_\CC \times \Omega^-_\CC"] \arrow[d, "\PSh{\dom_A} \times \PSh{\dom_B}"'] 
        & \PSh{\CC} \times \PSh{\CC} \arrow[d, "\PSh{\dom_A} \times \PSh{\dom_B}"] \\
    \PSh{\slice{\CC}{A}} \times \PSh{\slice{\CC}{B}} \arrow[r, "\Omega^-_{\slice{\CC}{A}} \times \Omega^-_{\slice{\CC}{B}}"']
        & \PSh{\slice{\CC}{A}} \times \PSh{\slice{\CC}{B}} 
    \end{tikzcd}
    \end{equation*}

We have to show that $[F, \Omega]\circ \op{(\dom_A)} = [F\circ \op{(\dom_A)}, \Omega]$ for all presheaves $F$ on $\CC$ and all $A\in \CC$, which follows by a Yoneda-type argument.\todo{Give details}

Finally, to make the last square commute, we have to assume that Day convolution preserves subobjects:
\begin{equation*}
    \begin{tikzcd}
        \PSh{\slice{\CC}{A}} \arrow[r, "\Omega^-_{\slice{\CC}{A}}"]  & \PSh{\slice{\CC}{A}} \\
        \PSh{\CC} \arrow[r, "\Omega_{\CC}^-"'] \arrow[u, "\PSh{\dom_A}"]& \PSh{\CC} \arrow[u, "\PSh{\dom_A}"']
    \end{tikzcd}
\end{equation*}

Now that all squares commute under the assumptions of a Day-stable site, the diagram collapses to the following.

\begin{equation*}
    \begin{tikzcd}
    \PSh{\CC} \arrow[d, "\Delta"'] & & \PSh{\CC} \arrow[dr, "\PSh{\cdot}"']\arrow[Rightarrow, from=U, "\lambda", shorten <=19pt, shorten >=6pt] & \PSh{\slice{\CC}{A} \times \slice{\CC}{B}}\\
    \PSh{\CC} \times \PSh{\CC} \arrow[r, "\Omega_{\CC}\times \Omega_{\CC}"']  & 
    \PSh{\CC} \times \PSh{\CC} \arrow[rr, "\DaySH"{name=U, below}] \arrow[ur, "\boxtimes_{\CC, \CC}"] & & 
    \PSh{\CC \times \CC} \arrow[u, "\PSh{\dom_A}\times \PSh{\dom_B}"{description}]  
    \end{tikzcd}
    \end{equation*}

Then, we define the to be produced map $\beta\colon \Pred_{F, 0}\boxtimes \Pred_{F, 0}\to \Pred_{F\DaySH F, 0} \circ \op{\cdot}$ as the whiskering
\begin{equation*}
    \begin{aligned}
    \beta :=\; & \PSh{\dom_A \times \dom_B}\,
                \lambda\, (\Omega^-_{\CC} \times \Omega^-_{\CC} \circ \Delta) \colon \\
              & \quad \boxtimes_{\slice{\CC}{A}, \slice{\CC}{B}} 
                \circ (\Omega^-_{\slice{\CC}{A}} \times \Omega^-_{\slice{\CC}{B}}) 
                \circ (\PSh{\dom_A} \times \PSh{\dom_B}) \\
              & \Rightarrow\;
                \PSh{\gamma_{A, B}} 
                \circ \Omega^-_{\slice{\CC}{A \cdot B}} 
                \circ \PSh{\dom_{A \cdot B}} 
                \circ \DaySH 
                \circ \Delta
    \end{aligned}
    \end{equation*}

\noindent inducing the map $\alpha\colon \Pred_{F, 0}\DaySH \Pred_{F, 0}\to \Pred_{F\DaySH F, 0}$ by the UMP of the left kan extension.